\documentclass[a4paper,UKenglish,cleveref, autoref, thm-restate]{lipics-v2021}
\usepackage{mathtools}
\usepackage{multicol}
\usepackage[normalem]{ulem}

\title{Resource Optimisation of Coherently Controlled Quantum Computations with the PBS-calculus}

\author{Alexandre Cl\'ement}{Universit\'e de Lorraine, CNRS, Inria, LORIA, F-54000 Nancy, France \and \url{https://members.loria.fr/AClement} }{alexandre.clement@loria.fr}{https://orcid.org/0000-0002-7958-5712}{}

\author{Simon Perdrix}{Universit\'e de Lorraine, CNRS, Inria, LORIA, F-54000 Nancy, France \and \url{https://members.loria.fr/SPerdrix} }{simon.perdrix@loria.fr}{https://orcid.org/0000-0002-1808-2409}{}

\authorrunning{A. Cl\'ement and S. Perdrix} 



\category{} 

\relatedversion{} 

\supplement{}

\funding{This work is funded by ANR-17-CE25-0009 SoftQPro, ANR-17-CE24-0035 VanQuTe,
PIA-GDN/Quantex, and LUE / UOQ, the PEPR integrated project EPiQ ANR-22-PETQ-0007 part of Plan France 2030, and by  \emph{``Investissements d'avenir''} (ANR-15-IDEX-02) program of
the French National Research Agency; the European Project NExt ApplicationS of Quantum Computing (NEASQC), funded by Horizon 2020 Program inside the call H2020-FETFLAG-2020-01(Grant Agreement 951821), and the HPCQS European High-Performance Computing Joint Undertaking (JU) under grant agreement No 101018180.
}

\nolinenumbers
\hideLIPIcs

\EventEditors{John Q. Open and Joan R. Access}
\EventNoEds{2}
\EventLongTitle{ICALP}
\EventShortTitle{ICALP 2020}
\EventAcronym{ICALP}
\EventYear{2016}
\EventDate{December 24--27, 2016}
\EventLocation{Little Whinging, United Kingdom}
\EventLogo{}
\SeriesVolume{42}
\ArticleNo{1}

\usepackage{tikzit}

\tikzstyle{diamant}=[diamond, fill=couleurdefond, draw=black]
\tikzstyle{newe}=[fill={gray!15}, draw=black, shape=rectangle, tikzit fill={rgb,255: red,191; green,191; blue,191}]
\tikzstyle{newerouge}=[fill={gray!15}, draw=black, shape=rectangle, tikzit fill={rgb,255: red,191; green,191; blue,191}]
\tikzstyle{newebleu}=[fill={gray!15}, draw=black, shape=rectangle, tikzit fill={rgb,255: red,191; green,191; blue,191}]
\tikzstyle{cercle}=[circle, fill=couleurdefond, draw=black]
\tikzstyle{cartouche}=[rounded rectangle, fill=couleurdefond, draw=black]
\tikzstyle{neg}=[rounded rectangle, fill=couleurdefond, draw=black, execute at end node={$\neg$}]
\tikzstyle{diagrammevide}=[rectangle, fill=couleurdefond, draw=black, inner sep=1.25em, borddiagrammevide, tikzit shape=rectangle]
\tikzstyle{mdiagrammevide}=[rectangle, fill=couleurdefond, draw=black, inner sep=0.75em, sborddiagrammevide, tikzit shape=rectangle]
\tikzstyle{sdiagrammevide}=[rectangle, fill=couleurdefond, draw=black, inner sep=0.5em, sborddiagrammevide, tikzit shape=rectangle]
\tikzstyle{xsdiagrammevide}=[rectangle, fill=couleurdefond, draw=black, inner sep=0.4em, xsborddiagrammevide, tikzit shape=rectangle]
\tikzstyle{xsdiagrammevideregulier}=[rectangle, fill=couleurdefond, draw=black, inner sep=0.33em, xsborddiagrammevideregulier, tikzit shape=rectangle]
\tikzstyle{bs}=[shape=beam, fill=couleurdefond, draw, inner sep=0.25em, thick, tikzit fill=white]
\tikzstyle{sbs}=[shape=beam, fill=couleurdefond, draw, inner sep=0.2em, thick, tikzit fill=white]
\tikzstyle{boite22}=[fill=white, draw=black, shape=rectangle, minimum height=1cm, minimum width=0.5cm]
\tikzstyle{boite2}=[fill=white, draw=black, shape=rectangle, minimum height=0cm, minimum width=0cm]
\tikzstyle{snegpotentiel}=[fill=couleurdefond, draw=black, shape=rounded rectangle, inner sep=0.25em, tikzit fill={rgb,255: red,191; green,191; blue,191}, execute at end node={\footnotesize$\star$}]
\tikzstyle{xxsnegpotentiel}=[fill=couleurdefond, draw=black, shape=rounded rectangle, inner sep=0.25em, tikzit fill={rgb,255: red,191; green,191; blue,191}, execute at end node={$\star$}, scale=0.65]
\tikzstyle{negpotentiel}=[fill=couleurdefond, draw=black, shape=rounded rectangle, tikzit fill={rgb,255: red,191; green,191; blue,191}, execute at end node={$\star$}]
\tikzstyle{token}=[fill=black, draw=black, shape=circle, inner sep=0.1em]
\tikzstyle{npbs}=[shape=beam, horizontal fill={{npbsmoitiebasse}{npbsmoitiehaute}}, draw, inner sep=0.25em, thick]
\tikzstyle{npbsalenvers}=[shape=beam, horizontal fill={{npbsmoitiehaute}{npbsmoitiebasse}}, draw, inner sep=0.25em, thick]
\tikzstyle{neweviolet}=[fill={rosefonce!3}, draw=rosefonce, shape=rectangle, execute at begin node={\textcolor{rosefonce}\bgroup}, execute at end node={\egroup}, tikzit draw={rgb,255: red,170; green,0; blue,170}, tikzit fill={rgb,255: red,170; green,0; blue,170}]
\tikzstyle{point}=[fill=black, draw=white, shape=circle, inner sep=0.492pt]
\tikzstyle{newegris}=[fill={gris!3}, draw=gris, shape=rectangle, tikzit fill={rgb,255: red,191; green,191; blue,191}, tikzit draw={rgb,255: red,191; green,191; blue,191}, execute at begin node={\textcolor{gris}\bgroup}, execute at end node={\egroup}]
\tikzstyle{bouchon}=[shape=rectangle, draw=black, text width=0pt, minimum height=1em, inner sep=0pt, fill=black]
\tikzstyle{dot}=[fill=black, draw=none, shape=circle, inner sep=0pt, outer sep=0pt, scale=0.891]

\tikzstyle{new}=[-]
\tikzstyle{tirets}=[-, draw=black, dashed]
\tikzstyle{noire}=[-, draw=black]
\tikzstyle{rouge}=[-, draw=red]
\tikzstyle{bleue}=[-, draw=bleu, tikzit draw=blue]
\tikzstyle{longdashed}=[-, dash pattern=on 5pt off 5pt]
\tikzstyle{pointilles}=[-, draw=black, dotted]
\tikzstyle{grise}=[-, draw={rgb,255: red,191; green,191; blue,191}]
\tikzstyle{coulgen}=[-, draw=rosefonce]
\tikzstyle{coulgenbis}=[-, draw=orange]
\tikzstyle{coulgenter}=[-, draw=brunbleu]
\tikzstyle{borddiagrammevide}=[-, dash pattern=on 0.5em off 0.5em on 0.5em off 0.5em on 0.5em off 0em]
\tikzstyle{sborddiagrammevide}=[-, dash pattern=on 0.2em off 0.2em on 0.2em off 0.2em on 0.2em off 0em]
\tikzstyle{xsborddiagrammevide}=[-, dash pattern=on 0.16em off 0.16em on 0.16em off 0.16em on 0.16em off 0em]
\tikzstyle{violette}=[-, draw=rosefonce, tikzit draw={rgb,255: red,170; green,0; blue,170}]
\tikzstyle{xsborddiagrammevideregulier}=[-, dash pattern=on 0.132em off 0.132em on 0.132em off 0.132em on 0.132em off 0em]

\input{figures/styles-pbs.tikzdefs}

\usetikzlibrary{decorations.markings}
\tikzset{
  vertical/.style={thin, densely dotted, blue,
    decoration={markings,
   mark=between positions 1.5pt and 1 step 4pt with {
       \node[draw=white,text=blue,scale=0.9] {\tiny$\uparrow$};
      }
    },
    postaction=decorate,
  }
}
\tikzset{
  horizontal/.style={thin, densely dotted, red,
   decoration={markings,
    mark=between positions 1.5pt and 1 step 4pt with {
       \node[draw=white,text=red,scale=0.9] {\tiny$\rightarrow$};
       }
    },
    postaction=decorate,
  }
}

\tikzset{
  verticalbis/.style={thin, %
  dash pattern=on 0pt off 1pt, blue,
    decoration={markings,
   mark=between positions 0.4pt and 1 step 1.2pt with {
       \node[%
       text=blue,scale=0.36] {\tiny$\uparrow$};
      }
    },
    postaction=decorate,
  }
}
\tikzset{
  horizontalbis/.style={thin, %
  dash pattern=on 0pt off 1pt, red,
   decoration={markings,
    mark=between positions 0.4pt and 1 step 1.2pt with {
       \node[%
       text=red,scale=0.36] {\tiny$\rightarrow$};
       }
    },
    postaction=decorate,
  }
}

\usetikzlibrary{cd}

\usepackage{amsmath}
\hypersetup{hypertexnames=false,raiselinks=true}
\usepackage{stmaryrd}
\usepackage{scalerel}
\usepackage{array,xtab,tabularx,longtable}

\newcolumntype{C}{>{$}c<{$}}  
\newcolumntype{R}{>{$}r<{$}}  
\newcolumntype{L}{>{$}l<{$}}  
\newcommand{\interp}[1]{\left\llbracket #1 \right\rrbracket}
\newcommand{\ket}[1]{\left| #1 \right\rangle}
\newcommand{\bra}[1]{\left\langle #1 \right|}
\newcommand{\scalprod}[2]{\left\langle #1\middle|#2\right\rangle}

 {\everymath{\displaystyle\everymath{}}\array}%
 {\endarray}
 {\everymath{\scriptstyle\everymath{}}\array}%
 {\endarray}

\renewcommand\H{\mathcal H} 
\newcommand\Ve{\mathcal V} 
\renewcommand\S{\mathcal S}
\newcommand\M{\mathsf M}

\newcommand\CC{\mathbb{C}}
\newcommand\RR{\mathbb{R}}
\newcommand\QQ{\mathbb{Q}}
\newcommand\ZZ{\mathbb Z}

\newcommand\GG{\mathcal{G}}
\newcommand\UU{\mathcal{U}}
\newcommand\hv{\{\Vpol,\Hpol\}}
\newcommand\CPBS{\textup{CPBS}}
\newcommand\diag[1]{\textup{\bf Diag}^{#1}}

\newcommand{\Hpol}{\textbf{\textup{H}}}
\newcommand{\Vpol}{\textbf{\textup{V}}}

\newcommand\noa{\#}
\newcommand\nopbs{\#_{\textup{PBS}}}
\newcommand\noneg{\#_{\neg}}

\newcommand{\dbtilde}[1]{\ThisStyle{\tilde{\raisebox{0pt}[0.85\height]{$\SavedStyle\tilde{#1}$}}}}

\newlength{\xlutmvcyp}

\makeatletter
\newlength{\negph@wd}
\DeclareRobustCommand{\negphantom}[1]{%
  \ifmmode
    \mathpalette\negph@math{#1}%
  \else
    \negph@do{#1}%
  \fi
}
\newcommand{\negph@math}[2]{\negph@do{$\m@th#1#2$}}
\newcommand{\negph@do}[1]{%
  \settowidth{\negph@wd}{#1}%
  \hspace*{-\negph@wd}%
}
\makeatother

\makeatletter
\newlength{\halfnegph@wd}
\DeclareRobustCommand{\halfnegphantom}[1]{%
  \ifmmode
    \mathpalette\halfnegph@math{#1}%
  \else
    \halfnegph@do{#1}%
  \fi
}
\newcommand{\halfnegph@math}[2]{\halfnegph@do{$\m@th#1#2$}}
\newcommand{\halfnegph@do}[1]{%
  \settowidth{\halfnegph@wd}{#1}%
  \hspace*{-0.5\halfnegph@wd}%
}
\makeatother

\makeatletter
\newlength{\halfph@wd}
\DeclareRobustCommand{\halfphantom}[1]{%
  \ifmmode
    \mathpalette\halfph@math{#1}%
  \else
    \halfph@do{#1}%
  \fi
}
\newcommand{\halfph@math}[2]{\halfph@do{$\m@th#1#2$}}
\newcommand{\halfph@do}[1]{%
  \settowidth{\halfph@wd}{#1}%
  \hspace*{0.5\halfph@wd}%
}
\makeatother

\newcommand\changelargeur[2]{#1\negphantom{#1}\phantom{#2}}

\makeatletter
\DeclareRobustCommand{\crefnosort}[1]{%
  \begingroup\@cref@sortfalse\cref{#1}\endgroup
}
\makeatother

\makeatletter
\DeclareRobustCommand{\crefnosortnocompress}[1]{%
  \begingroup\@cref@sortfalse\@cref@compressfalse\cref{#1}\endgroup
}
\makeatother

\makeatletter
\newcommand*\vspacebeforeline[1]{%
    \ifvmode 
        \vskip #1
        \vskip \z@skip
    \else
        \@bsphack
        \vadjust pre {%
            \@restorepar
            \vskip #1
            \vskip \z@skip
        }%
        \@esphack
    \fi
}
\makeatother

\newcommand{\cnot}{\mathit{CNot}}
\newcommand{\ov}{I}
\newcommand{\interpath}[1]{\left\llbracket #1 \right\rrbracket}

\newcommand\eqeqref[1]{\overset{\ThisStyle{\raisebox{0.3em}{$\SavedStyle\eqref{#1}$}}}{=}}

\newcommand\eqdeuxeqref[2]{\overset{\ThisStyle{\raisebox{0.3em}{$\SavedStyle\eqref{#1}\eqref{#2}$}}}{=}}
\newcommand\eqtroiseqref[3]{\overset{\ThisStyle{\raisebox{0.3em}{$\SavedStyle\eqref{#1}\eqref{#2}\eqref{#3}$}}}{=}}

\newcommand\eqexpl[1]{\overset{#1}{=}}

\newcommand{\ntikzfig}[1]{{\tikzset{tikzfig/.style={baseline=-0.25em,scale=0.5,every node/.style={scale=1},borddiagrammevide/.style={-, dash pattern=on 0.5em off 0.5em on 0.5em off 0.5em on 0.5em off 0em}}}\input{#1.tikz}}}
\newcommand{\stikzfig}[1]{{\tikzset{tikzfig/.style={baseline=-0.25em,scale=0.2,every node/.style={scale=0.9}}}\input{#1.tikz}}}
\newcommand{\sgtikzfig}[1]{{\tikzset{tikzfig/.style={baseline=-0.25em,scale=0.2667,every node/.style={scale=0.925}}}\input{#1.tikz}}}
\newcommand{\xstikzfig}[1]{{\tikzset{tikzfig/.style={baseline=-0.25em,scale=0.15,every node/.style={scale=0.8}}}\input{#1.tikz}}}
\newcommand{\xgtikzfig}[1]{{\tikzset{tikzfig/.style={baseline=-0.25em,scale=0.25,every node/.style={scale=0.7}}}\input{#1.tikz}}}
\newcommand{\xntikzfig}[1]{{\tikzset{tikzfig/.style={baseline=-0.25em,scale=0.25,every node/.style={scale=0.805}}}\input{#1.tikz}}}

\newlength{\xfycguhh}
\newcommand{\ptikzfig}[4]{{\tikzset{tikzfig/.style={baseline=-0.25em,scale=0.5,scale=#2,every node/.style={scale=#3},borddiagrammevide/.style={-, dash pattern=on #4\xfycguhh off #4\xfycguhh on #4\xfycguhh off #4\xfycguhh on #4\xfycguhh off 0em}}}\input{#1.tikz}}}
\newcommand{\mtikzfig}[1]{{\tikzset{tikzfig/.style={baseline=-0.25em,scale=0.35,every node/.style={scale=0.8},borddiagrammevide/.style={-, dash pattern=on 0.35em off 0.35em on 0.35em off 0.35em on 0.35em off 0em}}}\input{#1.tikz}}}

\newcounter{eqnabc}
\newenvironment{eqnabc}{\refstepcounter{eqnabc}\equation}{\tag{\alph{eqnabc}}\endequation}

\newcounter{eqnABC}
\newenvironment{eqnABC}{\refstepcounter{eqnABC}\equation}{\tag{\Alph{eqnABC}}\endequation}

\newcommand{\urlalt}[2]{\href{#2}{\nolinkurl{#1}}}

\begin{document}

\maketitle

\begin{abstract}
Coherent control of quantum computations can be used to improve some quantum protocols and algorithms. For instance, the complexity of implementing the permutation of some given unitary transformations can be strictly decreased by allowing coherent control, rather than using the standard quantum circuit model. In this paper, we address the problem of optimising the resources of coherently controlled quantum computations. We refine the PBS-calculus, a graphical language for coherent control which is inspired by quantum optics. In order to obtain a more resource-sensitive language, it manipulates abstract gates -- that can be interpreted as queries to an oracle -- and more importantly, it avoids the representation of useless wires by allowing unsaturated polarising beam splitters. Technically the language forms a coloured prop. The language is equipped with an equational theory that we show to be sound, complete, and minimal.

Regarding resource optimisation, we introduce an efficient procedure to minimise the number of oracle queries of a given diagram.
We also consider the problem of minimising both the number of oracle queries and the number of polarising beam splitters. We show that this optimisation problem is NP-hard in general, but introduce an efficient heuristic that produces optimal diagrams when at most one query to each oracle is required.
\end{abstract}

\section{Introduction}

Most models of quantum computation (like quantum circuits) and most quantum programming languages are based on the \emph{quantum data/classical control} paradigm. In other words, based on a set of quantum primitives (e.g. unitary transformations, quantum measurements), the way these primitives are applied on a register of qubits  is either fixed or classically controlled.

However, quantum mechanics offers more general control of operations: for instance in quantum optics it is easy to control the trajectory of a system, like a photon, based on its polarisation using a \emph{polarising beam splitter}. One can then position distinct quantum primitives on the distinct trajectories. Since the polarisation of a photon can be in superposition, it achieves some form of quantum control, called coherent control: the quantum primitives are  applied in superposition depending on the state of another quantum system. Coherent control is not only a subject  of interest for foundations of quantum mechanics~\cite{hardy2005probability,oreshkov2012quantum,zych2019bell}, it also leads to advantages in solving computational problems 
~\cite{facchini2015quantum,araujo2014computational,colnaghi2012quantum,renner2021reassessing} and in designing more efficient protocols~\cite{feix2015quantum,chiribella2012perfect,Abbott2020communication,ebler2018enhanced,guerin2016exponential}.

Indeed, some problems can be solved more efficiently by using coherent control rather than the usual quantum circuits. This separation has been proved in a multi-oracle model where the measure of complexity is the number of queries to (a single or several distinct) oracles, which are generally unitary maps. The simplest example is the following problem \cite{chiribella2012perfect}: given two oracles $U$ and $V$ with the promise  that they  are either commuting or  anti-commuting, decide whether $U$  and $V$ are commuting or not. This problem can be solved using the so-called quantum switch  \cite{chiribella13}  which can be implemented using only two queries by means of coherent control, whereas solving this problem requires at least 3 queries (e.g. two queries to $U$ and one query to $V$) in the quantum circuit model (see \cref{fig:QS}).

\begin{figure}[]
\centerline{$\ntikzfig{QS-}$\qquad\qquad$\ntikzfig{circ}$}
\caption{\label{fig:QS} [\emph{Left}] Coherently controlled quantum computation for solving the commuting problem. Only two queries are used: one query to $U$ and one query to $V$. 
  [\emph{Right}] Optimal circuit for solving the commuting problem, where the 3-qubit gate is a control-swap. Notice that three queries are necessary in the quantum circuit model.}

\end{figure}

In this paper, we address the problem of optimising resources of coherently controlled quantum computations. To do so, we first refine the framework of the PBS-calculus -- a graphical language for coherently controlled quantum computation --  to make it more resource-sensitive. Then, we consider the problem of optimising the number of queries, and also the number of polarising beam splitters, of a given coherently controlled quantum computation, described as a PBS-diagram. 

\vspace{0.1cm}
\noindent  {\bf PBS-calculus.} The PBS-calculus is a graphical language that has been introduced \cite{alex2020pbscalculus} to represent and reason about quantum computations involving coherent control of quantum operations. Inspired by quantum optics~\cite{clement2022LOv}, the polarising beam splitter (PBS for  short), denoted $\stikzfig{beamsplitter}$ is at the heart of the language: when a photon enters the PBS, say from the top left, it is reflected (and hence outputted on the top right) if its polarisation is vertical; or transmitted (and hence outputted on the bottom right) if its polarisation is horizontal. If the polarisation is a superposition of vertical and horizontal, the photon is outputted in a superposition of two positions. As a consequence, the trajectory of a particle, say a photon, will depend on its polarisation. The second main ingredient of the PBS-calculus are the gates, denoted $\sgtikzfig{gateU}$ which applies some transformation $U$ on a data register. Notice that the gates never act on the polarisation of the particle.
 
PBS-diagrams, which form a traced symmetric monoidal category (more precisely a traced prop \cite{maclane1965categorical}), are equipped with an equational theory that allows one to transform a diagram. The equational theory has been proved to be sound, complete, and minimal \cite{alex2020pbscalculus}. 
 
 Notice that a PBS-diagram may have some useless wires, like in the example of the ``half quantum switch'', see Figure \ref{fig:halfQS}~(left). We refine the PBS-calculus in order to allow one to remove these useless wires, leading to unsaturated PBS (or 3-leg PBS) like  $\stikzfig{beamsplitter-BL}$ or $\stikzfig{beamsplitter-TR}$. To avoid ill-formed diagrams like $\sgtikzfig{non-valid}$,  a typing discipline is necessary. To this end, we use the framework of coloured props: each wire has 3 possible colours: black, red and blue which can be interpreted as follows: a  photon going through a blue (resp. red) wire must have a horizontal (resp. vertical) polarisation.

The introduction of unsaturated polarising beam splitters requires to revisit the equational theory of the PBS-calculus. 
The heart of the refined equational theory is  the axiomatisation of the 3-leg polarising beam splitters, together with some additional equations which govern how 4-leg polarising beam splitters can be decomposed into 3-leg ones. To show the  completeness of the refined equational theory, we introduce normal  forms and show that any diagram can be put in normal form. Finally, we also show the minimality of the equational theory by proving that none of the equations  can be derived from the other ones. 

\begin{figure}
\centerline{\ntikzfig{half-BS-trace}\qquad\ntikzfig{half-BS}}
\caption{\label{fig:halfQS} A coherent control of $U$ and $V$, also called a half quantum switch: when the initial polarisation is vertical ($\Vpol$), $U$ is applied on the data register, when the polarisation is horizontal ($\Hpol$), $V$ is applied. Whatever the polarisation is, the particle always goes out of the top port of the second beam splitter. On the right-hand side the diagram is made of beam splitters with a missing leg, whereas on the left-hand side standard beam splitters are used, and a useless trace is added.}
\end{figure}

\vspace{0.1cm}
\noindent {\bf Resource Optimisation.} The  PBS-calculus, tha\-nks to its refined  equational theory, provides a  way to detect and remove dead-code in a diagram. We exploit  this property  to address the crucial question of resource optimisation. 
We introduce a specific form of diagrams that minimises the number of gates, more precisely the number of queries to oracles, with an appropriate modelisation of oracles. 
We provide an efficient procedure to transform any diagram into this specific form.  We then focus on the problem of  optimising both the number of queries and 
 the number of polarising beam  splitters. We refine the previous procedure, leading to an efficient heuristic. We  show that the produced diagrams are optimal when every oracle is queried at  most once, but might not be optimal in general.  
  We actually show that the general optimisation problem is NP-hard using a reduction from the \emph{maximum Eulerian cycle decomposition problem} \cite{caprara1999eulerian}.

\vspace{0.1cm}
 \noindent{\bf Related  work.}
Several languages have been designed to represent coherently controlled quantum computation: some of them are extensions of quantum circuits, and other diagrammatic languages \cite{wechs2021quantum,arrighi2021addressable,wilson2020diagrammatic,vanrietvelde2021routed}; others are based on abstract programming languages \cite{altenkirch2005functional,ying2012defining,dowek2017lineal,diaz2019realizability,DBLP:journals/corr/BadescuP15}. While there are numerous works on resource-optimisation of quantum computation, in particular for quantum circuits \cite{kliuchnikov2013optimization,amy2014polynomial,nam2018automated}, there was, up to  our knowledge, no 
procedure for resource optimisation of coherently controlled quantum computation.

\section{Coloured PBS-Diagrams}\label{colouredPBSdiag}

We use the formalism of traced coloured props (i.e. small traced symmetric strict monoidal categories whose objects are freely spanned by the elements of a set of colours) to represent coherently controlled quantum computations. We are going to use the ``colours'' $\vtype$, $\htype$, $\top$, to denote respectively vertical, horizontal or possibly both polarisations. 

\begin{definition}\label{def:diag}
Given a monoid $\M$, let $\diag \M$ be  the traced coloured prop   with colours $\{\vtype, \htype, \top\}$ freely generated by the following generators, for any $U\in \M$:
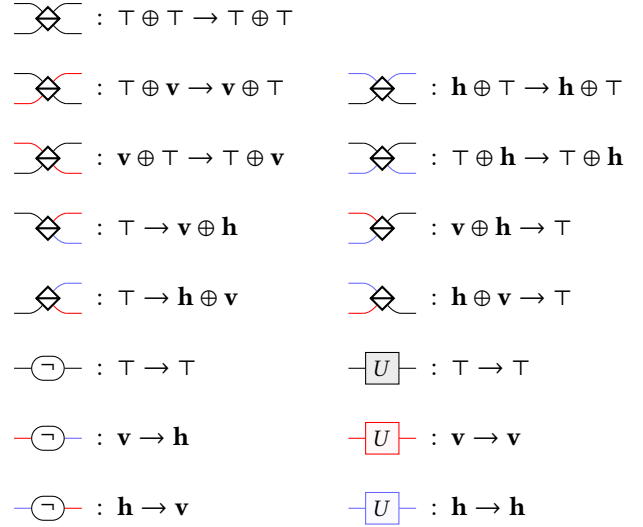
\vspace{0.2cm}

\centerline{$\begin{array}{rclcrclcrcl}
\stikzfig{beamsplitter}&:&\top \oplus \top \to \top \oplus \top &~~~~&\stikzfig{beamsplitterNRRN}&:&\top\oplus \vtype \to \vtype\oplus \top&~~~~&\stikzfig{beamsplitterwBNR}&:&\top\to \htype\oplus \vtype  \\[0.2cm]
\stikzfig{beamsplitterBBNN}&:&\htype  \oplus \top \to \htype\oplus \top &~~~~&\stikzfig{beamsplitterRNNR}&:& \vtype \oplus \top \to \top \oplus \vtype&~~~~&\stikzfig{beamsplitterRNBw}&:&\vtype\oplus \htype \to \top\\[0.2cm]
\stikzfig{beamsplitterNNBB}&:&\top \oplus \htype \to\top \oplus \htype &~~~~&\stikzfig{beamsplitterNRwB}&:&\top \to \vtype\oplus \htype&~~~~&\stikzfig{beamsplitterBwRN}&:&\htype\oplus \vtype \to \top\\[0.2cm]
\stikzfig{neg-}&:&\top \to \top&~~~~&\stikzfig{cnegRB-}&:&\vtype \to \htype&~~~~&\stikzfig{cnegBR-}&:&\htype \to \vtype\\[0.2cm]
\stikzfig{gateU-}&:&\top \to \top&~~~~&\stikzfig{cgateUR-}&:&\vtype \to\vtype&~~~~&\stikzfig{cgateUB-}&:&\htype \to\htype\\[0.2cm]
\end{array}$}

\end{definition}

The morphisms of $\diag \M$ are called $\M$-diagrams or simply diagrams when $\M$ is irrelevant or clear from the context.
Intuitively, the diagrams are inductively obtained by composition of the generators from Definition \ref{def:diag} using the sequential composition $D_2\circ D_1$, the parallel composition $D_3\oplus D_4$, and the trace $Tr_d(D)$ which are respectively depicted as follows:{$\quad\ntikzfig{composition-s-symmetrique}\quad\ntikzfig{produittensoriel-s-symmetrique}\quad\ntikzfig{trace-s-symmetrique}$}. Notice that these compositions should type-check, i.e. $D_1:a\to b$, $D_2:b\to c$ and $D:a\oplus d\to b\oplus d$ with $d\in \{\top, \vtype, \htype\}$.  The axioms of the traced coloured prop guarantee that the diagrams are defined up to deformation: two diagrams whose graphical representations are isomorphic are equal.

Regarding notations, we use  actual colours for wires:  blue  for $\htype$-wires, red for $\vtype$-wires, and black for $\top$-wires. We also add labels on the wires, which are omitted when clear from the context,
so that there is no loss of information in the case of a colour-blind reader or  black and white printing.\footnote{With the convention that the type $\top$ is always omitted (in other words, a wire of ambiguous type is black by convention). See~\cite{clem2023}, Example~1.7 for additional explanations about how to infer the omitted  types.} Two examples of diagrams are given in Figure \ref{fig:ex1}.

\begin{figure}[!]
\centerline{$\ptikzfig{ex2}{0.8}{0.8}1$\qquad\qquad\ptikzfig{NFcoloreeexemplecourtmonoidesansId}{0.41}{0.7}1}
\caption{\label{fig:ex1} (\emph{Left}) An example of diagram of type $\top \oplus \top \oplus \vtype \oplus \htype\to \top \oplus \htype \oplus \top \oplus \vtype$. (\emph{Right}) An example of a diagram of type $\top\oplus\vtype\oplus\htype\oplus\top\oplus\htype\to\htype\oplus\top\oplus\vtype\oplus\top\oplus\htype$, in a particular form that we will call \emph {normal form} (see \cref{defNFcoloree}).}%
\end{figure}

Unless  specified, the unit of $\M$ is denoted $I$ and its composition is $\cdot$ which will be generally omitted ($VU$ rather than $V\cdot U$). 
The main two  examples of monoids  we consider in the rest of the paper are:
\begin{itemize}
\item The monoid $\mathcal{U}(\H)$ of isometries of a Hilbert space $\H$ with the usual composition. When $\H$ is of finite dimension, the elements of $\mathcal{U}(\H)$ are unitary maps. With a slight abuse of notations, the corresponding traced coloured prop of diagrams is denoted $\diag \H$. 
\item The free monoid $\GG^*$ on some set  $\GG$. The gates, when the monoid is freely generated, can be interpreted as queries to oracles (each element of $\GG$ corresponds to an oracle): the gates implement \emph{a priori} arbitrary operations with no particular structures. We use the term \emph{abstract diagram} when the underlying monoid is freely generated, and we refer to the elements of $\GG$ as \emph{names}. Notice that the free monoid case can also be seen as an extension of \emph{bare diagrams}~\cite{branciard2021coherent}.
\end{itemize}

There are other examples of interests: One can consider for instance a monoid of commuting or anticommuting gates, that can be used to model the problem studied in \cite{chiribella2012perfect}. Another example is the monoid of $n$-qubit quantum circuits whose generators are layers of gates acting on $n$ qubits (e.g. $H \otimes \cnot\otimes I \otimes H $ when $n=5$ where $H$ is the $1$-qubit Hadamard gate, $\cnot$  the 2-qubit controlled-not gate, and $I$ the 1-qubit identity) and whose composition is the sequential composition of circuits. The monoid can be quotiented by equations like $(H \otimes  I)\cdot( I \otimes H) = H\otimes H$ and $(H\otimes I)\cdot(H\otimes I)= I\otimes I$. 
Finally, one can consider the monoid of unitary purifications\footnote{Given a Hilbert space $\H$, the elements of the monoid are triplets $[U,\ket\varepsilon,\mathcal E]$ where $\mathcal E$ is a Hilbert space, $U:\H\otimes \mathcal E\to \H\otimes \mathcal E$ is a unitary transformation, and $\ket{\varepsilon}\in \mathcal E$. The composition is defined as $[U_2,\ket {\varepsilon_2},\mathcal E_2]\cdot [U_1,\ket {\varepsilon_1},\mathcal E_2] = [(\sigma_{\mathcal E_1, \H}\otimes I)(I\otimes U_2) (\sigma_{\H, \mathcal E_1}\otimes I)(U_1\otimes I)\,,\, \ket{\varepsilon_1}\otimes \ket {\varepsilon_2},\,\mathcal E_1\otimes \mathcal E_2]$ where $\sigma_{\mathcal K, \mathcal K'}$ is the swap between the two Hilbert spaces $\mathcal K$, $\mathcal K'$.\label{footnotedefmonoidpurif}} used to describe coherent control of quantum channels~\cite{branciard2021coherent}.

\section{Semantics}

The input of a diagram is a single particle, which has a polarisation, a position and a data register. A basis state for the polarisation is either vertical or horizontal, and a basis state for the position is an integer which corresponds to the wire on which the particle is located. The type of a diagram restricts the possible input/output configurations: if $D:\vtype\oplus \top\to \htype \oplus \htype\oplus \vtype$ then the possible input (resp. output) configurations are the following polarisation-position pairs: $\{({\Vpol},0), ({\Vpol},1), ({\Hpol},1)\}$ (resp. $\{({\Hpol},0),({\Hpol},1),({\Vpol},2)\}$).
 More generally for any object $a$, let $[a]$ be the set of possible configurations, and $|a|$ be its size, inductively defined as follows: $|\ov| = 0$, $|a \oplus \top| = |a \oplus \vtype |=|a \oplus \htype| = |a|+1$, and $[\ov] = \emptyset$, $[a\oplus \vtype] = [a]\cup\{(\Vpol,|a|)\}$, $[a\oplus \htype] = [a]\cup\{(\Hpol,|a|)\}$ and $[a\oplus \top] = [a]\cup\{(\Vpol,|a|), (\Hpol,|a|)\}$.

The semantics of an $\M$-diagram $D:a\to b$ is a map $[a] \to [b]\times \M$ which associates with an input configuration $(c,p)$, an output configuration $(c',p')$ and a side effect $U_k\ldots U_1\in \M$ which represents the action applied on a data register of the particle. Thus the semantics of a diagram can be formulated as   follows:

\begin{definition}[Action semantics]\label{defactionsem}
Given an $\M$-diagram $D:a \to b$, let $\interpath D : [a] \to [b]\times \M$ be inductively defined as: $ \forall D_1:a\to b, D_2:b\to d, D_3 :d \to e, D_4:a\oplus f\to b\oplus f$, where $f\in \{\top, \vtype, \htype\}$:
\begin{longtable}{@{\quad\ }L@{\quad}L}
\interpath{\stikzfig{beamsplitterNRwB}}=\begin{cases}({\Vpol},0)\mapsto(({\Vpol},0),I)\\({\Hpol},0)\mapsto(({\Hpol},1),I)\end{cases}&\interpath{\stikzfig{beamsplitterwBNR}}=\begin{cases}({\Vpol},0)\mapsto(({\Vpol},1),I)\\({\Hpol},0)\mapsto(({\Hpol},0),I)\end{cases}\\[1.5em]
\interpath{\stikzfig{beamsplitterRNBw}}=\begin{cases}({\Vpol},0)\mapsto(({\Vpol},0),I)\\({\Hpol},1)\mapsto(({\Hpol},0),I)\end{cases}&\interpath{\stikzfig{beamsplitterBwRN}}=\begin{cases}({\Vpol},1)\mapsto(({\Vpol},0),I)\\({\Hpol},0)\mapsto(({\Hpol},0),I)\end{cases}\\[1.5em]
\interpath{\stikzfig{abbeamsplitter}}=(c,p)\mapsto\begin{cases}((c,p),I)&\text{if $c={\Vpol}$}\\((c,1-p),I)&\text{if $c={\Hpol}$}\end{cases}&\interpath{\stikzfig{cswap}}=(c,p)\mapsto((c,1-p),I)\\[1.7em]
\begin{array}{@{}l}\interpath{\ntikzfig{cfilcourt-s}}=(c,0)\mapsto((c,0),I)\\[1.2em]\interpath{\sgtikzfig{cgateU}}=(c,0)\mapsto((c,0),U)\end{array}&\interpath{\sgtikzfig{cneg-}}=\begin{cases}({\Vpol},0)\mapsto(({\Hpol},0),I)\\({\Hpol},0)\mapsto(({\Vpol},0),I)\end{cases}\\[1.7em]
\interpath{D_1\oplus  D_3} = (c,p)\mapsto \begin{cases}\interpath{D_1}(c,p)&\text{if $p<|a|$}\\S_a(\interpath{D_3}(c,p-|a|))&\text{otherwise}\end{cases}&\qquad\interpath{D_2\circ D_1} = \interpath {D_2} \circ \interpath {D_1}\\[1.5em]
\changelargeur{\interpath{Tr_f(D_4)}=(c,p)\mapsto \begin{cases}  \interpath{D_4}(c,p)\quad\text{if $\pi_{pos}(\interpath{D_4}(c,p))< |b|$}\\\interpath{D_4}\circ S_{a-b}(\interpath{D_4}(c,p))\\\changelargeur{\qquad\quad\text{if $\pi_{pos}(\interpath{D_4}\circ S_{a-b}(\interpath{D_4}(c,p)))< |b|\leq\pi_{pos}(\interpath{D_4}(c,p))$}}{}\\\interpath{D_4}\circ S_{a-b}(\interpath{D_4}\circ S_{a-b}(\interpath{D_4}(c,p)))\quad\text{otherwise}\end{cases}}{}
\end{longtable}
where the composition is: $g\circ f (c,p)=((c'',p''), U'U)$ with $f(c,p)=((c',p'),U)$ and $g(c',p')=((c'',p''), U')$; $\pi_{pos}:[a]\times \M\to \mathbb N:: ((c,p),U)\mapsto p$ is the projector on the position; and $S_a:[b]\times \M \to [a\oplus b] \times \M :: ((c,p),U) \mapsto ((c,p+|a|),U)$ and $S_{a-b}:[b]\times \M \to [a] \times \M :: ((c,p),U) \mapsto ((c,p+|a|-|b|),U)$ shift the position.

Given $D:a\to b$ and $(c,p)\in[a]$, we denote respectively by $c^D_{c,p}$, $p^D_{c,p}$ and $U^D_{c,p}$ the polarisation, the position and the element of $\M$, such that $\interpath{D}(c,p)=((c^D_{c,p}, p^D_{c,p}),U^D_{c,p})$.
In the case where $\M$ is the free monoid $\GG^*$, its elements can be seen as words, so we will use the notation $w^D_{c,p}$ instead of $U^D_{c,p}$.
\end{definition}

Notice that the semantics of the trace is not  defined as a fixed point but as a finite number of unfoldings.  Indeed, like for PBS-diagrams, one can show that any wire of a diagram is used at most twice, each time with a distinct polarisation.

\begin{proposition}\label{semanticsiswelldefined}
$\interpath .$ is well-defined, i.e. the axioms of the traced coloured prop are sound and the semantics of the trace is well-defined. 
\end{proposition}
\begin{proof}
This can be proved in a similar way as in Appendix~B.1.1 of~\cite{alex2020pbscalculus} and Appendix~B.1 of~\cite{branciard2021coherent}, noting that our action semantics could equivalently be defined as a big-step path semantics similar to those of \cite{alex2020pbscalculus} and \cite{branciard2021coherent} (cf. Section 4.2 of~\cite{clem2023}).
\end{proof}
\subsection{Quantum Semantics}

Any diagram whose underlying monoid consists of linear maps admits a \emph{quantum semantics}, defined as follows:

\begin{definition}[Quantum semantics]
Given a  monoid $\M$ of  linear maps (with the standard composition) on a complex vector space $\Ve$, for any $\M$-diagram $D:a\to b$ the quantum semantics of $D$ is the linear map $V_D: \mathbb C^{[a]}\otimes \Ve\to \mathbb C^{[b]}\otimes \Ve :: \ket{c,p}\otimes \ket \varphi \mapsto \ket{c^D_{c,p},p^D_{c,p}}\otimes U^D_{c,p}\ket  \varphi$.
\end{definition}

The diagrams in $\diag \H$ are valid by construction, in the sense that their semantics are  valid quantum evolutions:

\begin{proposition}\label{quantumsemiso}
For any $D\in \diag \H$, $V_D: \mathbb C^{[a]}\otimes \H\to \mathbb C^{[b]}\otimes \H$ is an isometry. 
\end{proposition}
\begin{proof}
The proof is given in Appendix \ref{preuvequantumsemiso}.
\end{proof}

Note that $\interpath{D} = \interpath {D'}$ implies $V_D = V_{D'}$; the converse is true if and only if $0\notin \M$:
\begin{proposition}\label{equivquantumaction}
Given a monoid $\M$ of complex linear maps, we have $\forall D,D', \interpath{D} = \interpath {D'} \Leftrightarrow V_D=V_{D'}$, if and only if $0\notin \M$. 
\end{proposition}
\begin{proof}
The proof is given in Appendix \ref{preuveequivquantumaction}
\end{proof}

In particular, two diagrams in $\diag \H$ have the same action semantics if and only if they have the same quantum semantics.  

\subsection{Interpretation}
Given a monoid homomorphism $\gamma: \M \to \M'$, one can transform any $\M$-diagram into an $\M'$-diagram straightforwardly, by applying $\gamma$ on each gate of the diagram:

\begin{definition}\label{def:interp}
Given an $\M$-diagram $D:a\to b$ and a monoid homomorphism $\gamma\colon\M\to\M'$, we define its $\gamma$-inter\-pretation $\gamma(D):a\to b$ as the $\M'$-diagram obtained by applying $\gamma$ to each gate of $D$. It is defined inductively as: $\gamma(\sgtikzfig{cgateU}:a\to a)=\sgtikzfig{cgategammadeU}:a\to a$; for any other generator $g$, $\gamma(g)=g$; $\gamma(D_2\circ D_1)=\gamma(D_2)\circ\gamma(D_1)$; $\gamma(D_1\oplus D_2)=\gamma(D_1)\oplus\gamma(D_2)$; and $\gamma(Tr_e(D))=Tr_e(\gamma(D))$.
\end{definition}

\begin{proposition}\label{prop:inter} Any $\M$-diagram is the interpretation of an abstract diagram.
\end{proposition}

\begin{proof}
Given an $\M$-diagram $D$, let $\GG$ be the underlying set of $\M$ and $\gamma : \GG \to \M$  s.t. $\forall U\in \GG$, $\gamma(U)=U$. The function $\gamma$ can be extended trivially into a homomorphism $\gamma\colon\GG^*\to\M$. Notice that $D$ can be seen as a (abstract) diagram of $\diag {\GG^*}$ and $\gamma(D) = D$. 
\end{proof}

It is easy to see that the action of monoid homomorphisms on diagrams is well-behaved with respect to the semantics:

\begin{proposition}
Given any $\M$-diagram $D:a\to b$ and any monoid homomorphism $\gamma\colon \M\to\M'$, for any configuration $(c,p)\in[a]$, if $\interpath{D}(c,p)=((c',p'),U)$ then $\interpath{\gamma(D)}(c,p)=((c',p'),\gamma(U))$.
\end{proposition}
\begin{proof}Straightforward by induction.\end{proof}

As a consequence, given two abstract diagrams $D_1, D_2\in \diag {\GG^*}$, if $\interpath{D_1} =\interpath{D_2}$ then for any homomorphism $\gamma:\GG^* \to \M$, $\interpath{\gamma(D_1)} =  \interpath{\gamma(D_2)}$. The converse is not true in general. Notice that in the framework of graphical languages, an equation holds in graphical languages for traced symmetric (resp. dagger compact closed) monoidal categories if and only if it holds in finite-dimensional vector (resp. Hilbert) spaces \cite{hasegawa2008finite,selinger2011finite}. We prove a similar result by showing that interpreting abstract diagrams using 2-dimensional Hilbert spaces is enough to completely characterise their semantics:

\begin{proposition}\label{matricescompletes}
Given a Hilbert space $\H$ of dimension at least $2$ and a set  $\GG$,
$\forall D_1,D_2\in \diag {\GG^*}$, there exists a monoid homomorphism $\gamma\colon\GG^*\to\UU(\H)$ s.t. $\interpath{D_1}=\interpath{D_2} \Leftrightarrow \interpath{\gamma(D_1)}=\interpath{\gamma(D_2)}$.
\end{proposition}

A stronger version of \cref{matricescompletes}, where the homomorphism $\gamma$ is independent of the diagrams, is also true, assuming the axiom of choice:

\begin{proposition}\label{interpretationuniverselle}
Given a Hilbert space $\H$ of dimension at least $2$, and a set $\GG$ of cardinality at most the cardinality of $\UU(\H)$,  there exists a monoid homomorphism $\gamma\colon\GG^*\to\UU(\H)$ s.t. $\forall D_1,D_2\in \diag {\GG^*}$, $\interpath{D_1}=\interpath{D_2} \Leftrightarrow \interpath{\gamma(D_1)}=\interpath{\gamma(D_2)}$.
\end{proposition}
\begin{proof}[Proof of \cref{matricescompletes,interpretationuniverselle}]
The proof is given in \cref{preuveinterpretationuniverselle}.
\end{proof}
\begin{remark}Notice that the cardinality of $\UU(\H)$ is $\max(2^{\aleph_0},2^{\dim(\H)})$ (where $2^{\aleph_0}$ is the cardinality of $\RR$ and $\dim(\H)$ is the Hilbert dimension of $\H$).\end{remark}

\section{Equational Theory}
\label{thecalculus}
In this section, we introduce an equational theory which allows one to transform any $\M$-diagram into an equivalent one. Indeed, all the equations we present in this section  preserve the semantics of the diagrams (see \cref{soundnessCPBS}).  

These equations are given in \cref{axiomsCPBS}. They form what we call
the \textup{CPBS}-calculus:
\begin{definition}[CPBS-calculus]Two $\M$-diagrams $D_1, D_2$ are equivalent according to the rules of the $\textup{CPBS}$-calculus, denoted $\textup{CPBS}\vdash D_1=D_2$,  if one can transform $D_1$ into $D_2$ using the equations given in Figure \ref{axiomsCPBS}. 
More precisely, $\textup{CPBS}\vdash \cdot = \cdot$ is defined as the smallest congruence 
which satisfies 
equations of Figure \ref{axiomsCPBS} in addition to the axioms of coloured traced prop.
\begin{figure*}[!htb]
\scalebox{.8333}{\begin{minipage}{1.2\textwidth}
\begin{multicols}{3}
\setcounter{equation}{0}
\newcommand{\nspazer}{-0.2em}
\begin{equation}\label{idboxR}\begin{array}{r@{\ }c@{\ }l}\input{gateIrouge-t.tikz}&=&\input{filmoyenR-label.tikz}\end{array}\end{equation}
\vspace{\nspazer}
\begin{equation}\label{fusionUVR}\begin{array}{r@{\ }c@{\ }l}\input{gateUgateVrouge-t.tikz}&=&\input{gateVUrouge-t.tikz}\qquad\end{array}\end{equation}
\vspace{\nspazer}
\begin{equation}\label{neguRB}\begin{array}{r@{\ }c@{\ }l}\input{UnegRB-label.tikz}&=&\input{negRBU-label.tikz}\end{array}\end{equation}
\vspace{\nspazer}\vspace{\nspazer}
\begin{equation}\label{nbscol}\begin{array}{r@{\ }c@{\ }l}\input{negbeamsplitterNRwB-sh.tikz}&=&\input{beamsplitterwBNRnegneg-sh.tikz}\end{array}\end{equation}
\vspace{\nspazer}\vspace{\nspazer}
\begin{equation}\label{ubscol}\begin{array}{r@{\ }c@{\ }l}\input{UbeamsplitterNRwB-sh.tikz}&=&\input{beamsplitterNRwBUU-sh.tikz}\end{array}\end{equation}
\vspace{\nspazer}\vspace{\nspazer}
\begin{equation}\label{bouclevideUR}\begin{array}{r@{\ }c@{\ }l}\ptikzfig{bouclevideUR-label}{0.65}{0.74286}{1}&=&\ptikzfig{diagrammevide}{1}{0.5}{0.5}\end{array}\end{equation}
\vspace{\nspazer}
\begin{equation}\label{negnegRBR}\begin{array}{r@{\ }c@{\ }l}\input{negnegRBR-label.tikz}&=&\input{filmoyenR-label.tikz}\end{array}\end{equation}
\vspace{\nspazer}
\begin{equation}\label{negnegBRB}\begin{array}{r@{\ }c@{\ }l}\input{negnegBRB-label.tikz}&=&\input{filmoyenB-label.tikz}\end{array}\end{equation}
\vspace{\nspazer}
\begin{equation}\label{bsbssplitmerge}\begin{array}{r@{\ }c@{\ }l}\input{NbsRBbsN-sh.tikz}&=&\input{filmoyen.tikz}
\end{array}\end{equation}
\vspace{\nspazer}
\begin{equation}\label{bsbsmergesplit}\begin{array}{r@{\ }c@{\ }l}\input{RBbsNbsRB-sh.tikz}&=&\input{filsparallelesmoyensRB-label.tikz}\end{array}\end{equation}
\vspace{\nspazer}
\begin{equation}\label{bsxcol}\begin{array}{r@{\ }c@{\ }l}\input{beamsplitterwBNR-sh.tikz}&=&\input{beamsplitterNRwBswap-sh.tikz}\end{array}\end{equation}
\vspace{\nspazer}
\begin{equation}\label{xbscol}\begin{array}{r@{\ }c@{\ }l}\input{beamsplitterBwRN-sh.tikz}&=&\input{swapbeamsplitterBwRN-sh.tikz}\end{array}\end{equation}
\vspace{\nspazer}
\begin{equation}\label{bsbigebre}\begin{array}{r@{\ }c@{\ }l}\input{beamsplitter-sh.tikz}&=&\ptikzfig{bsbigebre}{0.5}{0.8}1\end{array}\end{equation}
\vspace{\nspazer}\vspace{\nspazer}\vspace{\nspazer}
\begin{equation}\label{bsnrrn}\begin{array}{@{}r@{\ }c@{\ }l}\input{beamsplitterNRRN-label-sh.tikz}&=&\ptikzfig{defbsNRRN}{0.5}{0.8}1\end{array}\end{equation}
\vspace{\nspazer}\vspace{\nspazer}\vspace{\nspazer}
\begin{equation}\label{bsrnnr}\begin{array}{@{}r@{\ }c@{\ }l}\input{beamsplitterRNNR-label-sh.tikz}&=&\ptikzfig{defbsRNNR}{0.5}{0.8}1\end{array}\end{equation}
\vspace{\nspazer}\vspace{\nspazer}\vspace{\nspazer}
\begin{equation}\label{bsnnbb}\begin{array}{r@{\ }c@{\ }l}\input{beamsplitterNNBB-label-sh.tikz}&=&\input{defbsNNBB-s.tikz}\end{array}\end{equation}
\vspace{\nspazer}\vspace{\nspazer}\vspace{\nspazer}
\begin{equation}\label{bsbbnn}\begin{array}{r@{\ }c@{\ }l}\input{beamsplitterBBNN-label-sh.tikz}&=&\input{defbsBBNN-s.tikz}\end{array}\end{equation}\vspace{\nspazer}\vspace{\nspazer}
\end{multicols}
\end{minipage}}
\caption[Axioms of the CPBS-calculus]{Axioms of the CPBS-calculus. $U,V\in\M$. 
Equations (\ref{idboxR}) and (\ref{fusionUVR}) reflect the monoid structure of $\M$; Equations (\ref{neguRB}) to 
 (\ref{ubscol}) show how the three generators commute; Equation (\ref{bouclevideUR}) means that a disconnected diagram (with no inputs/outputs) can be removed;  
 Equations (\ref{negnegRBR}) to 
(\ref{bsbsmergesplit}) witness the fact that the negation and the 3-leg PBS are invertible;  Equations (\ref{bsxcol}) and (\ref{xbscol}) are essentially topological rules; Equations (\ref{bsbigebre}) to (\ref{bsbbnn}) show how 4-leg PBS can be decomposed into 3-leg PBS. Notice in particular that the other rules do not use 4-leg PBS, as a consequence one could define the language using 3-leg PBS only and see the 4-leg PBS as syntactic sugar. 
\label{axiomsCPBS}}
\end{figure*}
\end{definition}

Notice that the CPBS-calculus subsumes the  PBS-calculus: the fragment of monochromatic (black) $\CC^{q\times q}$-diagrams of the CPBS-calculus coincides with the set of PBS-diagrams, 
moreover, the completeness of both languages implies that for any two  
PBS-diagrams $D_1$ and $D_2$, $\textup{PBS} \vdash D_1=D_2$ if and only if $\textup{CPBS}\vdash D_1=D_2$.

\begin{proposition}[Soundness]\label{soundnessCPBS}
For any two $\M$-diagrams $D_1$ and $D_2$, if $\textup{CPBS}\vdash D_1 = D_2$ then $\interpath{D_1}=\interpath{D_2}$.
\end{proposition}
\begin{proof}
Since the semantic equality is a congruence, it suffices to check that for every equation of \cref{axiomsCPBS}, both sides have the same semantics, which is easy to do.
\end{proof}

We introduce normal forms, that
will be useful to prove that the equational theory is complete, and will also play a role in optimising the number of gates in a diagram in Section~\ref{resourceoptimisation}.

\begin{definition}\label{defNFcoloree}
A diagram is said to be in \emph{normal form} if it is of the form $M\circ P\circ F\circ G\circ S$, where:
\begin{itemize}
\item $S$ is of the form $b_1\oplus\cdots\oplus b_n$, where each $b_i$ is either $\begin{tikzpicture}
	\begin{pgfonlayer}{nodelayer}
		\node [style=none] (0) at (-0.6, 0) {};
		\node [style=none] (1) at (0.6, 0) {};
		\node [style=none] (2) at (0, 0.275) {\footnotesize $\vtype$};
	\end{pgfonlayer}
	\begin{pgfonlayer}{edgelayer}
		\draw [style=rouge] (0.center) to (1.center);
	\end{pgfonlayer}
\end{tikzpicture}
$, $\begin{tikzpicture}
	\begin{pgfonlayer}{nodelayer}
		\node [style=none] (0) at (-0.6, 0) {};
		\node [style=none] (1) at (0.6, 0) {};
		\node [style=none] (2) at (0, 0.325) {\footnotesize $\htype$};
	\end{pgfonlayer}
	\begin{pgfonlayer}{edgelayer}
		\draw [style=bleue] (0.center) to (1.center);
	\end{pgfonlayer}
\end{tikzpicture}
$ or $\xstikzfig{beamsplitterNRwB}$
\item $G$ is of the form $g_1\oplus\cdots\oplus g_k$, where each $g_i$ is
either $\begin{tikzpicture}
	\begin{pgfonlayer}{nodelayer}
		\node [style=none] (0) at (-0.6, 0) {};
		\node [style=none] (1) at (0.6, 0) {};
		\node [style=none] (2) at (0, 0.275) {\footnotesize $\vtype$};
	\end{pgfonlayer}
	\begin{pgfonlayer}{edgelayer}
		\draw [style=rouge] (0.center) to (1.center);
	\end{pgfonlayer}
\end{tikzpicture}
$, $\begin{tikzpicture}
	\begin{pgfonlayer}{nodelayer}
		\node [style=none] (0) at (-0.6, 0) {};
		\node [style=none] (1) at (0.6, 0) {};
		\node [style=none] (2) at (0, 0.325) {\footnotesize $\htype$};
	\end{pgfonlayer}
	\begin{pgfonlayer}{edgelayer}
		\draw [style=bleue] (0.center) to (1.center);
	\end{pgfonlayer}
\end{tikzpicture}
$, $\xntikzfig{cgateUiR}$ or $\xntikzfig{cgateUiB}$, with $U_i\neq I$
\item $F$ is of the form $n_1\oplus\cdots\oplus n_k$, where each $n_i$ is either $\begin{tikzpicture}
	\begin{pgfonlayer}{nodelayer}
		\node [style=none] (0) at (-0.6, 0) {};
		\node [style=none] (1) at (0.6, 0) {};
		\node [style=none] (2) at (0, 0.275) {\footnotesize $\vtype$};
	\end{pgfonlayer}
	\begin{pgfonlayer}{edgelayer}
		\draw [style=rouge] (0.center) to (1.center);
	\end{pgfonlayer}
\end{tikzpicture}
$, $\begin{tikzpicture}
	\begin{pgfonlayer}{nodelayer}
		\node [style=none] (0) at (-0.6, 0) {};
		\node [style=none] (1) at (0.6, 0) {};
		\node [style=none] (2) at (0, 0.325) {\footnotesize $\htype$};
	\end{pgfonlayer}
	\begin{pgfonlayer}{edgelayer}
		\draw [style=bleue] (0.center) to (1.center);
	\end{pgfonlayer}
\end{tikzpicture}
$, $\xntikzfig{cnegRB}$ or $\xntikzfig{cnegBR}$
\item $P$ is a permutation of the wires, that is, a trace-free diagram in which all generators are identity wires or swaps
\item $M$ is of the form $w_1\oplus\cdots\oplus w_m$, where each $w_i$ is either $\begin{tikzpicture}
	\begin{pgfonlayer}{nodelayer}
		\node [style=none] (0) at (-0.6, 0) {};
		\node [style=none] (1) at (0.6, 0) {};
		\node [style=none] (2) at (0, 0.275) {\footnotesize $\vtype$};
	\end{pgfonlayer}
	\begin{pgfonlayer}{edgelayer}
		\draw [style=rouge] (0.center) to (1.center);
	\end{pgfonlayer}
\end{tikzpicture}
$, $\begin{tikzpicture}
	\begin{pgfonlayer}{nodelayer}
		\node [style=none] (0) at (-0.6, 0) {};
		\node [style=none] (1) at (0.6, 0) {};
		\node [style=none] (2) at (0, 0.325) {\footnotesize $\htype$};
	\end{pgfonlayer}
	\begin{pgfonlayer}{edgelayer}
		\draw [style=bleue] (0.center) to (1.center);
	\end{pgfonlayer}
\end{tikzpicture}
$ or $\xstikzfig{beamsplitterRNBw}$\nolinebreak.
\end{itemize}
\end{definition}

For example, the diagram shown in \cref{fig:ex1}~(right)
is in normal form.

\begin{theorem}\label{existenceNF}
For any $\M$-diagram $D$, there exists an $\M$-diagram in normal form $N$ such that $\CPBS\vdash D=N$.
\end{theorem}
\begin{proof}
The proof is given in \cref{preuveexistenceNF}.
\end{proof}

Note that the structure of the normal form as well as the proof of \cref{existenceNF} use in an essential way the removal of useless wires made possible by the use of colours, and in particular Equation \eqref{bsbsmergesplit}, which has no equivalent in the monochromatic \textup{PBS}-calculus of 
\cite{alex2020pbscalculus}.
An example of a diagram and its normal form are given in Figure \ref{fig:ex2}. 

\begin{figure}[!]
\centerline{$\ptikzfig{exampleNF}{0.8}{0.8}1$\qquad\qquad$\ptikzfig{example-NF-2}{0.8}{0.8}1$}
\caption{\label{fig:ex2} An example of a diagram (\emph{left}) and its equivalent diagram in normal form (\emph{right}).} %
\end{figure}

Now we use the normal form to prove the completeness of the \CPBS-calculus:

\begin{lemma}[Uniqueness of the normal form]\label{uniquenessNF}
For any two diagrams in normal form $N$ and $N'$, if $\interpath N=\interpath{N'}$ then $N=N'$.
\end{lemma}
\begin{proof}
The proof is given in \cref{preuveuniquenessNF}.
\end{proof}

\begin{theorem}[Completeness]\label{completenessCPBS}
Given any two $\M$-diagrams $D_1$ and $D_2$, if $\interpath{D_1}=\interpath{D_2}$ then $\CPBS\vdash D_1=D_2$.%
\end{theorem}
\begin{proof}
By \cref{existenceNF}, there exist $N_1,N_2$ in normal form such that $\CPBS\vdash D_1=N_1$ and $\CPBS\vdash D_2=N_2$. By \cref{soundnessCPBS}, $\interpath{N_1}=\interpath{D_1}=\interpath{D_2}=\interpath{N_2}$. Therefore, by \cref{uniquenessNF}, $N_1=N_2$. By transitivity, this proves that $\CPBS\vdash D_1=D_2$.
\end{proof}

Finally, each equation of Figure \ref{axiomsCPBS} is necessary for the completeness:

\begin{theorem}[Minimality]\label{minimalityCPBS}
None of the equations of Figure \ref{axiomsCPBS} is a consequence of the others.
\end{theorem}
\begin{proof}
The proof is given in Appendix \ref{preuveminimalityCPBS}.
\end{proof}

\section{Resource Optimisation}\label{resourceoptimisation}

We show in this section that the equational theory of  the CPBS-calculus can be used for resource optimisation.

\subsection{Minimising the Number of Oracle Queries}

We consider the problem of minimising the number of oracle queries: given a set $\GG$ of (distinct) oracles and a $\GG^*$-diagram $D$, the objective is to find a diagram $D'$ equivalent to  $D$  (i.e. $\interpath D = \interpath{D'}$) such that $D'$ uses a minimal number of queries to each oracle. Since there are several oracles, the definition of the optimal diagrams should be made  precise.

First, we define the number of queries to a given oracle:

\begin{definition}
Given a $\GG^*$-diagram $D$, for any $U\in \GG$, let $\noa_U(D)$ be the number of queries to $U$ in $D$, inductively defined as follows: $\noa_U(\sgtikzfig{cgatew}) = |w|_U$; $\noa_U(g)=0$ for all the other generators; $\noa_U(D_1\oplus D_2) = \noa_U(D_2\circ D_1) = \noa_U(D_1)+\noa_U(D_2)$; and $\noa_U(Tr_a(D)) = \noa_U(D)$, where $|w|_U$ is the number of occurrences of $U$ in the word $w\in \GG^*$. 
\end{definition}

We can now define a query-optimal diagram as follows: 

\begin{definition}
A $\GG^*$-diagram $D$ is query-optimal if $\forall D'\in \diag{\GG^*}$, $\forall U\in \GG$, $\interpath{D} = \interpath{D'}$ implies $\noa_U(D)\le \noa_U(D')$.
\end{definition}

Note that given a diagram, it  is not \emph{a priori} guaranteed that there exists an equivalent diagram which is query-optimal: for instance, it might be that all the diagrams which minimise the number of queries  to some oracle $U$ do not  minimise the number of queries  to another oracle $V$. We actually show (Proposition \ref{gateoptimisationworks}) that any diagram can be turned into a query-optimal one. To this end, we first need a lower bound on the number of queries to a given oracle: 

\begin{proposition}[Lower bound]\label{prop:lowerbound}
For any $\GG^*$-diagram $D:a\to b$ and any $U\in \GG$, $\noa_U(D)\ge \left\lceil\sum_{(c,p)\in [a]}\frac{|w_{c,p}^D|_U}2\right\rceil$ 
where $w_{c,p}^D\in \GG^*$ is such that $\interpath D(c,p)=((c',p'),w^D_{c,p})$.
\end{proposition}
\begin{proof}Notice that each gate \sgtikzfig{cgatew} of the diagram $D$ is used at most twice according to the semantics,\footnote{This can be stated more formally by replacing the contents of the gates by distinct names in order to get a (coloured) bare diagram (see \cite{branciard2021coherent}), and then proved in a similar way as Proposition 3 of \cite{branciard2021coherent}.} in other words, there are either at most two pairs $(c,p)$, $(c',p')$ such that $w$ contributes once to $w^D_{c,p}$ and once to $w^D_{c',p'}$; or at most a single pair $(c,p)$ such that  $w$ contributes twice to $w^D_{c,p}$. As a consequence, $\sum_{(c,p)\in [a]} |w_{c,p}^D|_U \le 2\noa_U(D)$, which leads to the  lower bound.
\end{proof}

Note that Proposition \ref{prop:lowerbound} provides a lower bound on the minimal number of queries to $U$ one can reach in optimising a diagram since the right-hand side of the inequality only depends on the semantics of the diagram.

We are now ready to introduce an optimisation procedure that transforms any diagram into an equivalent query-optimal one:

\paragraph*{Query optimisation procedure of a $\GG^*$-diagram $D$:}
\begin{enumerate}
\item\label{mettreenNF} Transform $D$ into its normal form $D_{\mathit{NF}}$. A recursive procedure for doing this can easily be deduced from the proof of \cref{existenceNF} (given in \cref{preuveexistenceNF}).
\item Split all gates into elementary gates (that is, gates whose label is a single letter), using the  following variants of Equation \eqref{fusionUVR}, which are consequences of the equations of \cref{axiomsCPBS} (see \cref{preuveequationsgateoptproc}): $\forall U\in \GG$, $\forall w\in \GG^*$, $w\neq I$:\\[-0.5cm]\scalebox{0.9}{\begin{tabularx}{1.11\linewidth}{@{}XXX@{}}\begin{equation}\label{fusionUV-}\begin{array}{@{}r@{\ }c@{\ }l@{}}\begin{tikzpicture}
	\begin{pgfonlayer}{nodelayer}
		\node [style=newerouge] (0) at (0, 0) {$wU$};
		\node [style=none] (1) at (-1.8, 0) {};
		\node [style=none] (2) at (1.7, 0) {};
		\node [style=none] (3) at (-1.5, 0.45) {\footnotesize \vtype};
	\end{pgfonlayer}
	\begin{pgfonlayer}{edgelayer}
		\draw [style=rouge] (1.center) to (2.center);
	\end{pgfonlayer}
\end{tikzpicture}
&\to &\input{gateUgateVrouge-.tikz}\qquad\end{array}\end{equation}&
\begin{equation}\label{fusionUVB}\begin{array}{@{}r@{\ }c@{\ }l@{}}\begin{tikzpicture}
	\begin{pgfonlayer}{nodelayer}
		\node [style=newebleu] (0) at (0, 0) {$wU$};
		\node [style=none] (1) at (-1.8, 0) {};
		\node [style=none] (2) at (1.7, 0) {};
		\node [style=none] (3) at (-1.5, 0.45) {\footnotesize \htype};
	\end{pgfonlayer}
	\begin{pgfonlayer}{edgelayer}
		\draw [style=bleue] (1.center) to (2.center);
	\end{pgfonlayer}
\end{tikzpicture}
&\to &\input{gateUgateVbleu-.tikz}\qquad\end{array}\end{equation}&
\begin{equation}\label{fusionUV}\begin{array}{@{}r@{\ }c@{\ }l@{}}\begin{tikzpicture}
	\begin{pgfonlayer}{nodelayer}
		\node [style=newe] (0) at (0, 0) {$wU$};
		\node [style=none] (1) at (-1.5, 0) {};
		\node [style=none] (2) at (1.5, 0) {};
	\end{pgfonlayer}
	\begin{pgfonlayer}{edgelayer}
		\draw [style=noire] (1.center) to (0);
		\draw [style=noire] (0) to (2.center);
	\end{pgfonlayer}
\end{tikzpicture}
~&~\to &~\input{gateUgateVbis-.tikz}\qquad\end{array}\end{equation}\end{tabularx}}\vspace{-0.25cm}
\item As long as the diagram contains two non-black gates with the same label, merge them. To do so, deform the diagram to put one over the other, and apply one of the following equations, which are also consequences of the equations of \cref{axiomsCPBS}:\vspace{0.325cm}
\begingroup
\tikzset{tikzfig/.style={baseline=-0.25em,scale=0.3,every node/.style={scale=0.75}}}\\[-0.5cm]\begin{tabularx}{\linewidth}{XX}
\begin{equation}\label{mergegatesRB}\input{gateURsurgateUB-label.tikz}\ \to \ \input{RBbsNUbsRB.tikz}\end{equation}&
\begin{equation}\label{mergegatesBR}\input{gateUBsurgateUR-label.tikz}\ \to \ \input{BRbsNUbsBR.tikz}\end{equation}\\
\begin{equation}\label{mergegatesRR}\input{gateURsurgateUR-label.tikz}\ \to \ \input{RRnbbsNUbsnbRR.tikz}\end{equation}&
\begin{equation}\label{mergegatesBB}\input{gateUBsurgateUB-label.tikz}\ \to \ \input{BBnhbsNUbsnhBB.tikz}\end{equation}\end{tabularx}
\endgroup
\end{enumerate}

An example of query-optimised diagram is given in Figure \ref{fig:exopt}. 
The query-optimisation procedure transforms any diagram into an equivalent query-optimal one:

\begin{proposition}\label{gateoptimisationworks}
The diagram $D_0$ output by the query optimisation procedure is query-optimal: for any U and any $D'$ s.t. $\interpath{D'}=\interpath {D_0}$, 
one has $\noa_U(D_0)\le \noa_U(D')$.
\end{proposition}
\begin{proof}
Notice that in $D_{\mathit{NF}}$, for each gate there is one and only one input state $(c,p)$ which goes to this gate. As a consequence, $\forall U, \noa_U(D_{\mathit{NF}}) = \sum_{(c,p)\in [a]}{|w_{c,p}^D|_U}$ (where $D_{\mathit{NF}}:a\to b$). Moreover, 
$\forall U, \noa_U(D_0) = \left\lceil \frac {\noa_U(D_{\mathit{NF}})}{2}\right\rceil$, thus $D_0$ meets the  lower bound of \cref{prop:lowerbound} and hence is query-optimal.
\end{proof}

Note that the query-optimisation procedure is efficient: one can naturally define the size $|D|$ of a diagram $D\in \diag{\GG^*}$ as follows: $|\sgtikzfig{cgatew}|  = |w|$; $|g|=1$ for all the other generators; $|D_1\oplus D_2| = |D_2\circ D_1| = |D_1|+|D_2|$; and $|Tr_a(D)| = |D|+1$. Step \ref{mettreenNF} of the procedure, which consists in putting the diagram in normal  form, can be done using a number of elementary equations of Figure \ref{axiomsCPBS} which  is quadratic in the size of the diagram, the other two steps being linear.  Notice that here we only count the number of basic equations. The procedure also requires some diagrammatic transformations (that is, deformations), which can be handled efficiently (more precisely, at most in quadratic time) 
using  appropriate data  structures.

\subsection{Optimising Both Queries and PBS}

We refine the resource optimisation of a diagram by considering not only the number of queries but also the number of instructions, and in particular the number of polarising beam splitters. Note that the number of beam splitters and the number of queries cannot be minimised independently, in the sense that there might not exist a diagram that  is both query-optimal and PBS-optimal (see such an example  in Figure \ref{fig:notoptimal}).  As the implementation of an oracle is \emph{a priori} more expensive than the implementation of  a single PBS, we optimise the number of queries and then the number of PBS  in this order, i.e. the  measure of complexity is the lexicographic order number of queries, number of polarising beam splitters.

\begin{figure}
\centerline{$\ptikzfig{gateUBsurgateUR-label-}{0.6}{0.75}1\qquad\qquad \ptikzfig{BRbsNUbsBR-}{0.6}{0.75}1$}\caption{\label{fig:notoptimal}Two equivalent diagrams: the diagram on the left  is optimal in terms of the number of polarising beam splitters, the diagram on the right is optimal in terms of queries. Note that there is no equivalent diagram with no polarising beam splitter and at most a single query.}
\end{figure}

\begin{definition} A diagram $D$ is query-PBS-optimal if $D$ is  query-optimal and for  any  query-optimal diagram $D'$ equivalent to $D$ (i.e. $\interpath D = \interpath {D'}$), $\nopbs(D)\le \nopbs(D')$, where $\nopbs(D)$ be the number of PBS of $D$.\end{definition}

We introduce an efficient heuristic, called \emph{PGT  procedure} that, when applied on a query-optimal diagram $D_0$, preserves the number of queries. The produced diagram, called in PGT form (see \cref{fig:PGT}), contains at most as many PBS as the original diagram, and moreover is query-PBS-optimal when there is at most one query to each oracle (see \cref{pasplusdePBSquavant,optimalitequanduneportedechaque}).

\begin{figure}[tb]
\vspace{-0.75cm}
\begin{tabularx}{\linewidth}{X@{\qquad\qquad\qquad}X}\begin{eqnABC}\label{snail}\ptikzfig{CPNFnoire-presquecentree}{0.5}{0.7}1\end{eqnABC} &
\begin{eqnABC}\label{stairs}\hspace{-1.5em}\ptikzfig{superpetpermsnegpotentiels-ms}{0.55714}{0.7}1\end{eqnABC}\end{tabularx}
$\ptikzfig{escalierbsnoir-abrege}{0.43}{0.688}1\quad\ \ \ptikzfig{escalierbsrouge-label-abrege}{0.43}{0.688}1\quad\ \ \ptikzfig{escalierbsbleu-label-abrege}{0.43}{0.688}1\quad\ \ \ptikzfig{escalierbsrougemerge-label-abrege}{0.43}{0.688}1\quad\ \ \ptikzfig{escalierbsrougemergeinverse-label-abrege}{0.43}{0.688}1$
\caption[Schematic description of a diagram in PGT form]{\label{fig:PGT}Schematic description of a diagram in PGT form (for Permutation, Gates and Traces). A diagram is in PGT form if it is of the form \eqref{snail}, with $P$ of the form \eqref{stairs}, and the $C_i$ of the forms depicted on the second line. $\ptikzfig{negpotentiel-m}{0.4}{0.8}1$ denotes either $\input{cfilcourt-s.tikz}$ or $\xntikzfig{cneg}$ with $a\in\{\vtype,\htype\}$, and $\sigma_1,\sigma_2$ are permutations of the wires.}
\end{figure}

More precisely, the PGT procedure consists in putting $D_0$ in the so-called PGT form, which we prove to contain few PBS. 
First, we consider query-free diagrams:
\begin{definition}\label{defNMF}
A diagram {D} is in \emph{stair form} if it is of the form
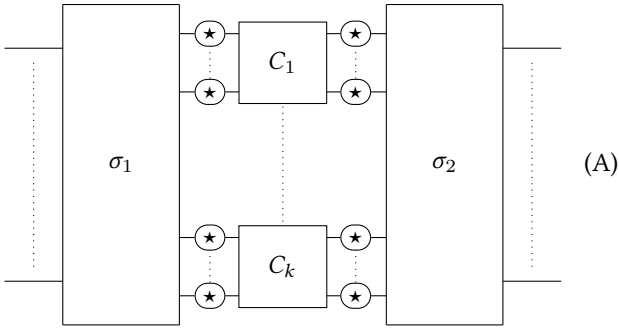
\begin{equation}\tag{\ref{stairs}}
\ptikzfig{superpetpermsnegpotentiels-m}{0.75}{0.75}1
\end{equation}
where $\sigma_1$ and $\sigma_2$ are permutations of the wires, $\ptikzfig{negpotentiel-m}{0.4}{0.8}1$ denotes either $\begin{tikzpicture}
	\begin{pgfonlayer}{nodelayer}
		\node [style=none] (0) at (-0.6, 0) {};
		\node [style=none] (1) at (0.6, 0) {};
		\node [style=none] (2) at (0, 0.325) {\footnotesize $a$};
	\end{pgfonlayer}
	\begin{pgfonlayer}{edgelayer}
		\draw [style=noire] (0.center) to (1.center);
	\end{pgfonlayer}
\end{tikzpicture}
$ or $\xntikzfig{cneg}$ with $a\in\{\vtype,\htype\}$, and $C_1,...,C_k$ are each of one of the following forms:
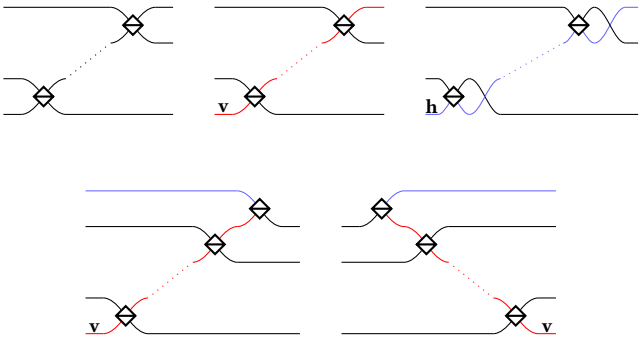
\[\ptikzfig{escalierbsnoir-abrege}{0.43}{0.688}1\quad\ \ \ptikzfig{escalierbsrouge-label-abrege}{0.43}{0.688}1\quad\ \ \ptikzfig{escalierbsbleu-label-abrege}{0.43}{0.688}1\quad\ \ \ptikzfig{escalierbsrougemerge-label-abrege}{0.43}{0.688}1\quad\ \ \ptikzfig{escalierbsrougemergeinverse-label-abrege}{0.43}{0.688}1\]
The diagrams of these forms will be called \emph{staircases}. The $C_i$ will be called the staircases of $D$.
\end{definition}

\begin{remark}Note that in the diagram \eqref{stairs}, all wires can be of arbitrary colours. We did not represent the labels in order to not overload the figures.
\end{remark}
Diagrams in stair form are optimal in terms of the number of polarising beam splitters:  

\begin{theorem}\label{stairsoptimal}
Any diagram $D:a\to b$ in stair form is PBS-optimal (that is, for any diagram $D':a\to b$, $\interpath{D}=\interpath{D'} \Rightarrow \nopbs(D)\leq \nopbs(D')$).
\end{theorem}
\begin{proof}
The proof is given in Appendix \ref{preuvestairsoptimal}.
\end{proof}

We extend the stair form to diagrams with queries as follows, leading to the  PGT form (for Permutation/Gates/Traces):

\begin{definition}
A $\GG^*$-diagram is in \emph{PGT form} (for Permutation, Gates and Traces) if it is of the form
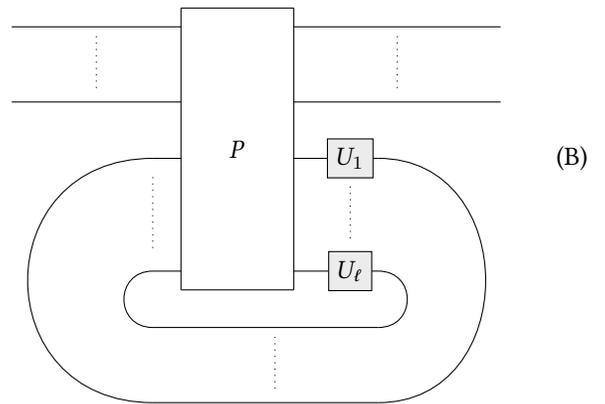
\begin{equation}\tag{\ref{snail}}
\mtikzfig{CPNFnoire}\end{equation}
where $P$ is in stair form {and $U_1, ...., U_\ell\in\GG$}.

\end{definition}
\begin{remark}Like in the diagram  \eqref{stairs}, all wires of  \eqref{snail} can be of arbitrary colours. 
\end{remark}

Contrary to the stair form, the PGT form is  not  optimal (see as an example Figure \ref{fig:notoptimal-PBS}). Intuitively, if there are several queries to an oracle $U$, then decomposing the corresponding gates into blue and red  gates and then recomposing them in a different way may lead to a diagram with a  smaller number of PBS. However, we will prove that applying the PGT procedure after the query optimisation procedure gives us a query-PBS-optimal diagram when there is at most one query to each oracle (see \cref{optimalitequanduneportedechaque}).

\begin{figure}
~\\\centerline{$\ntikzfig{PGTUUnonoptimal}$\qquad\qquad\qquad\qquad\qquad$\ntikzfig{PGTUUoptimal}$}
\caption{\label{fig:notoptimal-PBS} [\emph{Left}] An example of diagram in PGT form which is optimal in the number of  queries but not in the number of polarising beam splitters. Indeed it is equivalent to  the diagram on the right which is query-optimal and PBS-free.}\end{figure}

\bigskip

 The procedure relies on equations of Figure \ref{axiomsCPBS}, together with easy to derive variants of these equations. The derivations of the additional equations are given in \cref{derivationsforPGTproc}. The procedure, with all steps detailed, more pictures and explicit statement of the variants of the equations, is given in Appendix \ref{PBSoptprocillustree}.

\paragraph*{PGT  procedure:}
Given a query-optimal diagram $D_0$: 
\begin{enumerate}
\setcounter{enumi}{-1}
\item During all the procedure, every time there are two consecutive negations, we remove them using Equation \eqref{negnegRBR}, \eqref{negnegBRB} or their all-black version.
\item\label{putinsnailformabbr} Deform the query-optimal diagram $D_0$ to put it in the form \eqref{snail} with $P$ gate-free. The goal of the following steps is to put $P$ in stair form.
\item Split all PBS of the form $\stikzfig{abbeamsplitter}$ into combinations of $\stikzfig{beamsplitterNRwB}$, $\stikzfig{beamsplitterwBNR}$, $\stikzfig{beamsplitterRNBw}$ and $\stikzfig{beamsplitterBwRN}$, using Equations \eqref{bsbigebre} to \eqref{bsbbnn}.
\item As long as there are two PBS connected by a black wire, with possibly a black negation on this wire, push the possibly remaining negation out using Equation \eqref{nbscol}, and cancel the PBS together using Equation \eqref{bsbsmergesplit} and its variants. For example:
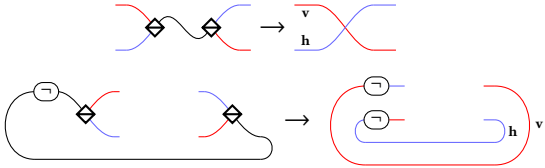
\[\tikzset{tikzfig/.style={baseline=-0.25em,scale=0.3,every node/.style={scale=0.6857}}}\input{RBbsNbsBRcompresse.tikz}\to\input{swapRB-label.tikz}\quad\ \ 
\tikzset{tikzfig/.style={baseline=-0.25em,scale=0.3,every node/.style={scale=0.6857}}}
\quad\input{tracenbsNRwBgapbsBwRN.tikz}\to\input{traceBdanstraceRgapnegs-label.tikz}\]
When there are not two such PBS anymore, all black wires are connected to at least one side of $P$ (possibly through negations), and the PBS are connected together with red and blue wires with possibly negations on them.
\item Remove all isolated loops. Note that since $D_0$ is query-optimal, there cannot be loops containing gates at this point.
\item Deform $P$ to put it in the form\eqref{stairs} with the $C_i$ of the form \smallskip\ptikzfig{escalierbsnegpotentielshaut}{0.6}{0.6857}1 and $\sigma_1$ and $\sigma_2$ being wire permutations, where $\ptikzfig{negpotentiel-m}{0.4}{0.8}1$ is either $\begin{tikzpicture}
	\begin{pgfonlayer}{nodelayer}
		\node [style=none] (0) at (-0.6, 0) {};
		\node [style=none] (1) at (0.6, 0) {};
		\node [style=none] (2) at (0, 0.325) {\footnotesize $a$};
	\end{pgfonlayer}
	\begin{pgfonlayer}{edgelayer}
		\draw [style=noire] (0.center) to (1.center);
	\end{pgfonlayer}
\end{tikzpicture}
$ or $\xntikzfig{cneg}$ with $a\in\{\vtype,\htype\}$, $\xstikzfig{bs3pattessymsplit}$ is either $\xstikzfig{beamsplitterNRwB}$ or $\xstikzfig{beamsplitterwBNR}$ and $\xstikzfig{bs3pattessymmerge}$ is either $\xstikzfig{beamsplitterRNBw}$ or $\xstikzfig{beamsplitterBwRN}$.
\item\label{pushnegations} Remove the negations in the middle of the $C_i$ by pushing them to the bottom by means of variants of Equation \eqref{nbscol}.
\item Transform each $C_i$, which is now, up to deformation, a ladder of PBS without negations, into one of the five kinds of stairs depicted in Definition \ref{defNMF}, depending on its type. To do so, deform it and apply Equations \eqref{bsxcol} and \eqref{xbscol} appropriately, and repeatedly apply the appropriate equation among \eqref{bsnrrn}, \eqref{bsrnnr}, \eqref{bsnnbb}, and a variant of \eqref{bsbigebre}. This gives us $D_1$.
\end{enumerate}

An example of diagram  produced by the PGT  procedure is given in \cref{fig:exopt}. 
\begin{figure}[]
\centerline{$\ptikzfig{ex-query-opt}{0.8}{0.8}{0.8}\qquad\qquad\qquad\ptikzfig{ex-query-pbs-opt}{0.8}{0.8}{0.8}$}
\caption{\label{fig:exopt} The diagram on the left is the obtained by applying the query-optimisation procedure on the example of Figure \ref{fig:ex2}. The diagram on the right is (up to deformation) obtained by applying the PGT procedure to the diagram on the left. Note that this diagram is both query- and PBS-optimal.} 
\end{figure}

Since the PGT procedure consists in putting a subdiagram of $D_0$ in stair form (except Step \ref{putinsnailformabbr} which is just deformation and does not change the number of PBS), Theorem \ref{stairsoptimal} implies in particular that this procedure does not increase the number of PBS in $D_0$:

\begin{proposition}\label{pasplusdePBSquavant}
The diagram $D_1$ output by the PGT procedure contains at most as many PBS as  the initial diagram $D_0$.\end{proposition}

This also implies that given any diagram $D$, there exists an equivalent query-PBS-optimal diagram in PGT form. Indeed, by \cref{gateoptimisationworks}, there exist query-optimal diagrams equivalent to $D$, and among these diagrams, some of them have minimal number of PBS and are therefore query-PBS-optimal. Finally, applying the PGT procedure to one of these diagrams gives us an equivalent diagram in PGT form, which, since the PGT procedure does not change the gates or increase the number of PBS, is still query-PBS-optimal.

Applying the PGT procedure after the query optimisation procedure produces an interesting heuristic: the output diagram is necessarily query-optimal and, although it is not necessarily query-PBS-optimal in general, it is whenever it does not contain two queries to the same oracle:

\begin{theorem}\label{optimalitequanduneportedechaque}
Given a diagram $D_1$ obtained by applying first the query optimisation procedure then the PGT procedure to a diagram $D$, if $D_1$ does not contain two queries to the same oracle (i.e. $\forall U\in \GG, \noa_U(D_1)\le 1$), then it is query-PBS-optimal.
\end{theorem}
\begin{proof}
The proof is given in Appendix \ref{preuveoptimalitequanduneportedechaque}.
\end{proof} 

Finally, note that, like the query optimisation procedure, the PGT procedure is efficient: it can be done using a number of elementary graphical transformations (those of  Figure \ref{axiomsCPBS}) which is 
linear in the  size of  the diagram. It  also requires  some diagrammatic transformations, which can  be handled  using  appropriate data  structures, leading to a quadratic algorithm.

\subsection{Hardness}

We show in this section that the query-PBS optimisation problem is actually NP-hard.

\begin{theorem}\label{gatePBSNPhard}
The problem of, given an abstract diagram, finding an equivalent query-PBS-optimal diagram, is NP-hard.
\end{theorem}
The proof, given in \cref{preuvegatePBSNPhard}, is based on a reduction from the maximum Eulerian cycle decomposition problem (MAX-ECD) which  is known to be NP-hard \cite{caprara1999eulerian}. The MAX-ECD  problem consists, given a graph, in finding a partition of its set of edges into the maximum number of cycles.  Intuitively, the reduction goes as follows: given an Eulerian  graph $G=(V=\{v_0,\ldots v_{n-1}\},E)$, let $\sigma$ be a permutation of the vertices of the graph s.t. $\forall i, (v_i,\sigma(v_i))\in E$ (such a $\sigma$ exists since $G$ is Eulerian), we construct a $V^*$-diagram $D$ such that the number of occurrences of each $v_i$ in $D$ is half its degree in $G$; and such that $\forall i$, $\interpath{D}(\Vpol,i)=((\Vpol,i), v_i)$ and $\interpath{D}(\Hpol,i)=((\Hpol,i), \sigma(v_i))$. Roughly speaking, we show that the edge-partitions of $G$ into cycles correspond to the possible implementations of $D$, and that an implementation with a minimal number of PBS leads to a partition with a maximal number of cycles.

In the following, we explore a few variants of the problem, which remain NP-hard.

First, query-PBS optimisation is still hard when restricted to negation-free diagrams:
\begin{corollary}\label{gatePBSnegfreeNPhard}
The problem of, given a negation-free abstract diagram, finding an equivalent diagram which is query-PBS-optimal among negation-free diagrams, is NP-hard.
\end{corollary}
\begin{proof}
The proof is given in \cref{preuvegatePBSnegfreeNPhard}.
\end{proof}

Additionally, it is also hard, in a query-optimal diagram, to optimise the PBS and the negations together,
respectively: with respect to a cost function (at least in the case where the cost of a negation is not less than the cost of a PBS); with the negations prioritised over the PBS; and with the PBS prioritised over the negations. Note that the NP-hardness is clear in the third case since the considered problem is a refinement of the query-PBS-optimisation problem addressed in \cref{gatePBSNPhard}.

\begin{corollary}\label{gateplusPBSNPhard}
For any $\alpha\geq 1$, the problem of, given an abstract diagram $D$, finding an equivalent query-optimal diagram $D'$ such that $\nopbs(D')+\alpha \noneg(D')$ is minimal,  is NP-hard, where $\noneg(D)$ is the number of negations in $D$.
\end{corollary}
\begin{proof}
The proof is given in \cref{preuvegateplusPBSNPhard}.
\end{proof}

Although we were only able to prove \cref{gateplusPBSNPhard} for $\alpha\geq1$, we conjecture that the optimisation is actually NP-hard even if the negations cost less than the PBS:
\begin{conjecture}
For any $\alpha\geq 0$, the problem of, given an abstract diagram $D$, finding an equivalent query-optimal diagram $D'$ such that $\nopbs(D')+\alpha \noneg(D')$ is minimal,  is NP-hard. 
\end{conjecture}

\begin{corollary}\label{querynegPBSNPhard}
The problem of, given an abstract diagram $D$, finding an equivalent query-$\neg$-PBS-optimal\footnote{A diagram is query-$\neg$-PBS-optimal if it is optimal according to the  lexicographic order: the number of queries then the number of negations and finally the number of polarising beam splitters. The definition of a query-PBS-$\neg$-optimal diagram is analogous.} diagram is NP-hard.
\end{corollary}
\begin{proof}
The NP-hardness of this problem directly follows from \cref{gatePBSnegfreeNPhard}. Indeed, given a negation-free diagram $D$, the query optimisation procedure gives us a negation-free query-optimal diagram $D'$ equivalent to $D$. Any query-$\neg$-PBS-optimal diagram equivalent to $D$ has to contain at most as many negations as $D'$, namely $0$, that is, be negation-free. Thus, finding a query-$\neg$-PBS-optimal equivalent to $D$ amounts to finding a negation-free query-PBS-optimal diagram equivalent to $D$.
\end{proof}
Finally, as noted above, the NP-hardness when the PBS are prioritised over the negations is a direct consequence of \cref{gatePBSNPhard}:
\begin{remark}
The problem of, given an abstract diagram $D$, finding an equivalent query-PBS-$\neg$-optimal
diagram is NP-hard.
\end{remark}

\section{Discussions and Future Work}

The power and limits of quantum coherent control is an intriguing question. Maybe surprisingly,\footnote{One may argue that this could be due to the simplicity of the language. It belongs to future work to know whether things would be as simple if the language were to be extended to allow for a more general quantum control.} we have proved that coherently controlled quantum computations, when expressed in the PBS-calculus, can be efficiently optimised: any PBS-diagram can be transformed in polynomial time into a diagram that is optimal in terms of oracle queries. We have refined the procedure to also decrease the number of polarising beam splitters. It leads to an optimal diagram when each oracle is queried only once, but the corresponding optimisation problem is NP-hard in general. We leave to future work an experimental evaluation of the PGT procedure when each oracle is not necessarily queried only once.

It might be that the NP-hardness result is even more significant than the optimisation heuristic, as the hardness might scale up as the language is further developed. There is however no certainty that things will necessarily happen as badly, and it might be a perspective for further developments of this language to find extensions of it in which such optimisation problems are easy to solve.

To perform the resource optimisation, we have introduced a few add-ons to the framework of the PBS-calculus. First, we have refined the syntax in order to allow the representation of unsaturated (or 3-leg) polarising beam splitters. They are essential ingredients for resource optimisation, as they provide a way to decompose a diagram into elementary components and then remove the useless ones.  However, note that one can perform resource optimisation of vanilla PBS-diagrams, using the refined one only as an intermediate language. Indeed, given a vanilla PBS-diagram (where all wires are black), one can apply  the  optimisation procedures described in this paper. The resulting optimised PBS-diagram may contain some unsaturated PBS, but all these 3-leg PBS can be saturated by adding useless traces and then one can make the diagram monochromatic. The resulting vanilla PBS-diagram keeps the same number of queries and PBS.

We have also generalised the gates of the diagrams, by considering arbitrary monoids. This is a natural abstraction that allows one to consider various examples and in particular the one of the free monoid which is appropriate to model the oracle queries. The query complexity is a convenient model to prove lower bounds, but note that the optimisation procedures described in this paper can be applied with any arbitrary monoid (for instance using Proposition \ref{prop:inter}). However, there is no guarantee of minimality with an arbitrary monoid. 

Another direction of research is to consider resource optimisation in a more expressive language for quantum control. Indeed, the polarisation of a particle can only be flipped within a PBS-diagram. The PBS-calculus is well suited for most applications of coherent control in quantum computing, by allowing the description of superpositions of classical controls (in particular superposition of causal orders) since the input particle can be in any superposition of polarisations. However, it would be interesting to develop  resource optimisation techniques for quantum computation involving arbitrary quantum control.

\bibliographystyle{plainurl}
\bibliography{bibpbsdiagrams}

\begin{thebibliography}{10}

\bibitem{Abbott2020communication}
Alastair~A. {Abbott}, Julian {Wechs}, Dominic {Horsman}, Mehdi {Mhalla}, and
  Cyril {Branciard}.
\newblock Communication through coherent control of quantum channels.
\newblock {\em {Quantum}}, 4:333, September 2020.
\newblock \href {https://arxiv.org/abs/1810.09826} {\path{arXiv:1810.09826}},
  \href {https://doi.org/10.22331/q-2020-09-24-333}
  {\path{doi:10.22331/q-2020-09-24-333}}.

\bibitem{altenkirch2005functional}
Thorsten Altenkirch and Jonathan Grattage.
\newblock A functional quantum programming language.
\newblock In {\em 20th Annual IEEE Symposium on Logic in Computer Science
  (LICS'05)}, pages 249--258. IEEE, 2005.

\bibitem{amir2010permutation}
Amihood Amir, Tzvika Hartman, Oren Kapah, Avivit Levy, and Ely Porat.
\newblock On the cost of interchange rearrangement in strings.
\newblock {\em SIAM Journal on Computing}, 39(4):1444--1461, 2010.
\newblock \href {https://doi.org/10.1137/080712969}
  {\path{doi:10.1137/080712969}}.

\bibitem{amy2014polynomial}
Matthew Amy, Dmitri Maslov, and Michele Mosca.
\newblock Polynomial-time {T}-depth optimization of {Clifford+T} circuits via
  matroid partitioning.
\newblock {\em IEEE Transactions on Computer-Aided Design of Integrated
  Circuits and Systems}, 33(10):1476--1489, 2014.

\bibitem{araujo2014computational}
Mateus Ara{\'u}jo, Fabio Costa, and {\v{C}}aslav Brukner.
\newblock Computational advantage from quantum-controlled ordering of gates.
\newblock {\em Physical Review Letters}, 113(25):250402, 2014.
\newblock \href {https://arxiv.org/abs/1401.8127} {\path{arXiv:1401.8127}},
  \href {https://doi.org/10.1103/PhysRevLett.113.250402}
  {\path{doi:10.1103/PhysRevLett.113.250402}}.

\bibitem{arrighi2021addressable}
Pablo Arrighi, Christopher Cedzich, Marin Costes, Ulysse R\'{e}mond, and
  Beno\^{\i}t Valiron.
\newblock Addressable quantum gates.
\newblock {\em ACM Transactions on Quantum Computing}, 4(3), April 2023.
\newblock \href {https://doi.org/10.1145/3581760} {\path{doi:10.1145/3581760}}.

\bibitem{DBLP:journals/corr/BadescuP15}
Costin Badescu and Prakash Panangaden.
\newblock Quantum alternation: Prospects and problems.
\newblock In Chris Heunen, Peter Selinger, and Jamie Vicary, editors, {\em
  Proceedings 12th International Workshop on Quantum Physics and Logic, {QPL}
  2015, Oxford, UK, July 15-17, 2015}, volume 195 of {\em {EPTCS}}, pages
  33--42, 2015.
\newblock \href {https://doi.org/10.4204/EPTCS.195.3}
  {\path{doi:10.4204/EPTCS.195.3}}.

\bibitem{branciard2021coherent}
Cyril Branciard, Alexandre Cl{\'e}ment, Mehdi Mhalla, and Simon Perdrix.
\newblock Coherent control and distinguishability of quantum channels via
  {PBS}-diagrams.
\newblock In Filippo Bonchi and Simon~J. Puglisi, editors, {\em 46th
  International Symposium on Mathematical Foundations of Computer Science (MFCS
  2021)}, volume 202 of {\em Leibniz International Proceedings in Informatics
  (LIPIcs)}, pages 22:1--22:20, Dagstuhl, Germany, August 2021. Schloss
  Dagstuhl -- Leibniz-Zentrum f{\"u}r Informatik.
\newblock URL: \url{https://hal.science/hal-03325456}, \href
  {https://arxiv.org/abs/2103.02073} {\path{arXiv:2103.02073}}, \href
  {https://doi.org/10.4230/LIPIcs.MFCS.2021.22}
  {\path{doi:10.4230/LIPIcs.MFCS.2021.22}}.

\bibitem{caprara1999eulerian}
Alberto Caprara.
\newblock Sorting permutations by reversals and {E}ulerian cycle
  decompositions.
\newblock {\em SIAM Journal on Discrete Mathematics}, 12(1):91--110, 1999.
\newblock \href {https://doi.org/10.1137/S089548019731994X}
  {\path{doi:10.1137/S089548019731994X}}.

\bibitem{chiribella2012perfect}
Giulio Chiribella.
\newblock Perfect discrimination of no-signalling channels via quantum
  superposition of causal structures.
\newblock {\em Physical Review A}, 86(4):040301, 2012.
\newblock \href {https://arxiv.org/abs/1109.5154} {\path{arXiv:1109.5154}},
  \href {https://doi.org/10.1103/PhysRevA.86.040301}
  {\path{doi:10.1103/PhysRevA.86.040301}}.

\bibitem{chiribella13}
Giulio Chiribella, Giacomo~Mauro D'Ariano, Paolo Perinotti, and Beno{\^\i}t
  Valiron.
\newblock Quantum computations without definite causal structure.
\newblock {\em Physical Review A}, 88:022318, August 2013.
\newblock \href {https://arxiv.org/abs/0912.0195} {\path{arXiv:0912.0195}},
  \href {https://doi.org/10.1103/PhysRevA.88.022318}
  {\path{doi:10.1103/PhysRevA.88.022318}}.

\bibitem{clement2022LOv}
Alexandre Cl\'{e}ment, Nicolas Heurtel, Shane Mansfield, Simon Perdrix, and
  Beno\^{i}t Valiron.
\newblock $\textup{LO}_\textup{v}$-calculus: A graphical language for linear
  optical quantum circuits.
\newblock In Stefan Szeider, Robert Ganian, and Alexandra Silva, editors, {\em
  47th International Symposium on Mathematical Foundations of Computer Science
  (MFCS 2022)}, volume 241 of {\em Leibniz International Proceedings in
  Informatics (LIPIcs)}, pages 35:1--35:16, Dagstuhl, Germany, 2022. Schloss
  Dagstuhl -- Leibniz-Zentrum f{\"u}r Informatik.
\newblock URL: \url{https://hal.science/hal-03926660}, \href
  {https://arxiv.org/abs/2204.11787} {\path{arXiv:2204.11787}}, \href
  {https://doi.org/10.4230/LIPIcs.MFCS.2022.35}
  {\path{doi:10.4230/LIPIcs.MFCS.2022.35}}.

\bibitem{alex2020pbscalculus}
Alexandre Cl{\'e}ment and Simon Perdrix.
\newblock {PBS}-calculus: A graphical language for coherent control of quantum
  computations.
\newblock In Javier Esparza and Daniel Kr{\'a}{\v l}, editors, {\em 45th
  International Symposium on Mathematical Foundations of Computer Science (MFCS
  2020)}, volume 170 of {\em Leibniz International Proceedings in Informatics
  (LIPIcs)}, pages 24:1--24:14, Dagstuhl, Germany, August 2020. Schloss
  Dagstuhl--Leibniz-Zentrum f{\"u}r Informatik.
\newblock URL: \url{https://hal.science/hal-02929291}, \href
  {https://arxiv.org/abs/2002.09387} {\path{arXiv:2002.09387}}, \href
  {https://doi.org/10.4230/LIPIcs.MFCS.2020.24}
  {\path{doi:10.4230/LIPIcs.MFCS.2020.24}}.

\bibitem{clem2023}
Alexandre Clément.
\newblock {\em Graphical Languages for Quantum Control and Linear Optics}.
\newblock PhD thesis, Universit{\'e} de Lorraine, May 2023.
\newblock URL: \url{http://www.theses.fr/en/2023LORR0093}.

\bibitem{colnaghi2012quantum}
Timoteo Colnaghi, Giacomo~Mauro D'Ariano, Stefano Facchini, and Paolo
  Perinotti.
\newblock Quantum computation with programmable connections between gates.
\newblock {\em Physics Letters A}, 376(45):2940--2943, 2012.
\newblock \href {https://arxiv.org/abs/1109.5987} {\path{arXiv:1109.5987}},
  \href {https://doi.org/10.1016/j.physleta.2012.08.028}
  {\path{doi:10.1016/j.physleta.2012.08.028}}.

\bibitem{diaz2019realizability}
Alejandro D{\'\i}az-Caro, Mauricio Guillermo, Alexandre Miquel, and Beno{\^\i}t
  Valiron.
\newblock Realizability in the unitary sphere.
\newblock In {\em 2019 34th Annual ACM/IEEE Symposium on Logic in Computer
  Science (LICS)}, pages 1--13. IEEE, 2019.

\bibitem{dowek2017lineal}
Gilles Dowek and Pablo Arrighi.
\newblock Lineal: A linear-algebraic lambda-calculus.
\newblock {\em Logical Methods in Computer Science}, 13, 2017.

\bibitem{ebler2018enhanced}
Daniel Ebler, Sina Salek, and Giulio Chiribella.
\newblock Enhanced communication with the assistance of indefinite causal
  order.
\newblock {\em Physical Review Letters}, 120(12):120502, March 2018.
\newblock \href {https://arxiv.org/abs/1711.10165} {\path{arXiv:1711.10165}},
  \href {https://doi.org/10.1103/PhysRevLett.120.120502}
  {\path{doi:10.1103/PhysRevLett.120.120502}}.

\bibitem{facchini2015quantum}
Stefano Facchini and Simon Perdrix.
\newblock Quantum circuits for the unitary permutation problem.
\newblock In {\em International Conference on Theory and Applications of Models
  of Computation}, pages 324--331. Springer, 2015.
\newblock \href {https://arxiv.org/abs/1405.5205} {\path{arXiv:1405.5205}},
  \href {https://doi.org/10.1007/978-3-319-17142-5_28}
  {\path{doi:10.1007/978-3-319-17142-5_28}}.

\bibitem{feix2015quantum}
Adrien Feix, Mateus Ara{\'u}jo, and {\v{C}}aslav Brukner.
\newblock Quantum superposition of the order of parties as a communication
  resource.
\newblock {\em Physical Review A}, 92(5):052326, 2015.
\newblock \href {https://arxiv.org/abs/1508.07840} {\path{arXiv:1508.07840}},
  \href {https://doi.org/10.1103/PhysRevA.92.052326}
  {\path{doi:10.1103/PhysRevA.92.052326}}.

\bibitem{Fleischner:Eulerian1991}
Herbert Fleischner.
\newblock {\em Eulerian Graphs and Related Topics}, volume~50 of {\em Annals of
  Discrete Mathematics}.
\newblock Elsevier, 1991.
\newblock URL:
  \url{https://www.elsevier.com/books/eulerian-graphs-and-related-topics/fleischner/978-0-444-89110-5}.

\bibitem{giles2019remarks}
Brett Giles and Peter Selinger.
\newblock Remarks on {M}atsumoto and {A}mano's normal form for single-qubit
  {Clifford+$T$} operators, 2019.
\newblock \href {https://arxiv.org/abs/1312.6584} {\path{arXiv:1312.6584}}.

\bibitem{guerin2016exponential}
Philippe~Allard Gu{\'e}rin, Adrien Feix, Mateus Ara{\'u}jo, and {\v{C}}aslav
  Brukner.
\newblock Exponential communication complexity advantage from quantum
  superposition of the direction of communication.
\newblock {\em Physical Review Letters}, 117(10):100502, 2016.
\newblock \href {https://arxiv.org/abs/1605.07372} {\path{arXiv:1605.07372}},
  \href {https://doi.org/10.1103/PhysRevLett.117.100502}
  {\path{doi:10.1103/PhysRevLett.117.100502}}.

\bibitem{hardy2005probability}
Lucien Hardy.
\newblock Probability theories with dynamic causal structure: a new framework
  for quantum gravity.
\newblock {\em arXiv preprint gr-qc/0509120}, 2005.

\bibitem{hasegawa2008finite}
Masahito Hasegawa, Martin Hofmann, and Gordon Plotkin.
\newblock Finite dimensional vector spaces are complete for traced symmetric
  monoidal categories.
\newblock In {\em Pillars of computer science}, pages 367--385. Springer, 2008.

\bibitem{holyer1981edgepartition}
Ian Holyer.
\newblock The {NP}-completeness of some edge-partition problems.
\newblock {\em SIAM Journal on Computing}, 10(4):713--717, 1981.
\newblock \href {https://doi.org/10.1137/0210054} {\path{doi:10.1137/0210054}}.

\bibitem{kliuchnikov2013optimization}
Vadym Kliuchnikov and Dmitri Maslov.
\newblock Optimization of {Clifford} circuits.
\newblock {\em Physical Review A}, 88(5):052307, 2013.

\bibitem{maclane1965categorical}
Saunders MacLane.
\newblock Categorical algebra.
\newblock {\em Bulletin of the American Mathematical Society}, 71(1):40--106,
  1965.

\bibitem{matsumoto2008representation}
Ken Matsumoto and Kazuyuki Amano.
\newblock Representation of quantum circuits with {C}lifford and $\pi/8$ gates,
  2008.
\newblock \href {https://arxiv.org/abs/0806.3834} {\path{arXiv:0806.3834}}.

\bibitem{nam2018automated}
Yunseong Nam, Neil~J. Ross, Yuan Su, Andrew~M. Childs, and Dmitri Maslov.
\newblock Automated optimization of large quantum circuits with continuous
  parameters.
\newblock {\em npj Quantum Information}, 4(1):1--12, 2018.

\bibitem{oreshkov2012quantum}
Ognyan Oreshkov, Fabio Costa, and {\v{C}}aslav Brukner.
\newblock Quantum correlations with no causal order.
\newblock {\em Nature communications}, 3(1):1--8, 2012.

\bibitem{renner2021reassessing}
Martin~J. Renner and {\v{C}}aslav Brukner.
\newblock Reassessing the computational advantage of quantum-controlled
  ordering of gates.
\newblock {\em Physical Review Research}, 3(4):043012, 2021.

\bibitem{selinger2011finite}
Peter Selinger.
\newblock Finite dimensional hilbert spaces are complete for dagger compact
  closed categories.
\newblock {\em Electronic Notes in Theoretical Computer Science},
  270(1):113--119, 2011.

\bibitem{vanrietvelde2021routed}
Augustin Vanrietvelde, Hl{\'e}r Kristj{\'a}nsson, and Jonathan Barrett.
\newblock Routed quantum circuits.
\newblock {\em Quantum}, 5:503, 2021.

\bibitem{wechs2021quantum}
Julian Wechs, Hippolyte Dourdent, Alastair~A. Abbott, and Cyril Branciard.
\newblock Quantum circuits with classical versus quantum control of causal
  order.
\newblock {\em PRX Quantum}, 2:030335, August 2021.
\newblock \href {https://doi.org/10.1103/PRXQuantum.2.030335}
  {\path{doi:10.1103/PRXQuantum.2.030335}}.

\bibitem{wilson2020diagrammatic}
Matt Wilson and Giulio Chiribella.
\newblock A diagrammatic approach to information transmission in generalised
  switches.
\newblock {\em Electronic Proceedings in Theoretical Computer Science},
  340:333–348, September 2021.
\newblock \href {https://doi.org/10.4204/eptcs.340.17}
  {\path{doi:10.4204/eptcs.340.17}}.

\bibitem{ying2012defining}
Mingsheng Ying, Nengkun Yu, and Yuan Feng.
\newblock Defining quantum control flow.
\newblock {\em arXiv preprint arXiv:1209.4379}, 2012.

\bibitem{zych2019bell}
Magdalena Zych, Fabio Costa, Igor Pikovski, and {\v{C}}aslav Brukner.
\newblock Bell's theorem for temporal order.
\newblock {\em Nature communications}, 10(1):1--10, 2019.

\end{thebibliography}

\appendix

\section{Proofs}

\subsection{Proof of Proposition \ref{quantumsemiso}}\label{preuvequantumsemiso}

Since there exists an orthonormal basis of $\CC^{[a]}\otimes\Ve$ composed of vectors of the form $\ket{c,p}\otimes\ket\varphi$, it suffices to check that $V_D$ preserves all scalar products of vectors of this form. For any $c$, $p$, $c'$, $p'$, $\ket\varphi$ and $\ket{\varphi'}$, one has $(\bra{c,p}\otimes\bra\varphi)V_D^\dag V_D^{}(\ket{c'\!,p'}\otimes\ket{\varphi'})=\scalprod{c^D_{c,p},p^D_{c,p}}{c^D_{c'\!,p'},p^D_{c'\!,p'}}\otimes\bra\varphi {U^D_{c,p}}^{\!\dag} U^D_{c'\!,p'}\ket{\varphi'}$. On the one hand, it can be proved in the same way as in~\cite{alex2020pbscalculus} that the function $(c,p)\mapsto(c^D_{c,p},p^D_{c,p})$ is a bijection, so that $(c^D_{c,p},p^D_{c,p})=(c^D_{c'\!,p'},p^D_{c'\!,p'})$ if and only if $(c,p)=(c'\!,p')$\smallskip. That is, $\scalprod{c^D_{c,p},p^D_{c,p}}{c^D_{c'\!,p'},p^D_{c'\!,p'}}=\scalprod{c,p}{c'\!,p'}=\begin{cases}1&\text{if $(c,p)=(c'\!,p')$}\\0&\text{if $(c,p)\neq(c'\!,p')$}\end{cases}$\smallskip. On the other hand, since $U^D_{c,p}$ is an isometry, if $(c,p)=(c'\!,p')$ then $\bra\varphi {U^D_{c,p}}^{\!\dag} U^D_{c,p}\ket{\varphi'}=\scalprod{\varphi}{\varphi'}$. Thus,\smallskip\\ $(\bra{c,p}\otimes\bra\varphi)V_D^\dag V_D^{}(\ket{c'\!,p'}\otimes\ket{\varphi'})=\begin{cases}\scalprod{\varphi}{\varphi'}&\text{if $(c,p)=(c'\!,p')$}\\0&\text{if $(c,p)\neq(c'\!,p')$}\end{cases}\allowbreak=(\bra{c,p}\otimes\bra\varphi)(\ket{c'\!,p'}\otimes\ket{\varphi'})$\smallskip.

\subsection{Proof of Proposition \ref{equivquantumaction}}\label{preuveequivquantumaction}

Let us assume that $\M$ is a monoid of linear maps on a complex vector space $\Ve$.

Since the quantum semantics is defined from the action semantics, it is clear that $\forall D,D', \interpath{D} = \interpath {D'} \Rightarrow V_D=V_{D'}$.

Given an $\M$-diagram $D$, if $0\notin \M$, then for all $c,p$, $U^D_{c,p}\neq 0$, so that there exists $\ket\varphi\in\Ve$ such that $U^D_{c,p}\ket\varphi\neq0$. Then $\ket{c^D_{c,p},p^D_{c,p}}\otimes U^D_{c,p}\ket\varphi\neq0$, which implies that $c^D_{c,p}$ and $p^D_{c,p}$ are uniquely determined from the data of $c$, $p$ and $V_D$. Since in any case, $U^D_{c,p}$ is uniquely determined from the data of $c$, $p$ and $V_D$, this implies that if $0\notin \M$ then $\interpath{D}$ is uniquely determined from $V_D$. Hence if $0\notin \M$ then for any two $\M$-diagrams $D$ and $D'$, $V_D=V_{D'} \Rightarrow \interpath{D} = \interpath {D'}$.

Conversely, if $0\in \M$, then for example, with $D=\sgtikzfig{gate0}$ and $D'=\sgtikzfig{gate0neg}$, both of type $\top\to\top$, one has $V_D=V_D'=0$ but $\interpath{D}(\Vpol,0)=((\Vpol,0),0)\neq\interpath{D'}(\Vpol,0)=((\Hpol,0),0)$.

\subsection{Proof of Propositions \ref{interpretationuniverselle} and \ref{matricescompletes}}\label{preuveinterpretationuniverselle}

Given a monoid homomorphism $\gamma\colon\GG^*\to\UU(\H)$, a $\GG^*$-diagram $D$ and any $c,p$, one has $\interpath{\gamma(D)}(c,p)=((c^{D}_{c,p},p^{D}_{c,p}),\gamma(w^{D}_{c,p}))$. Therefore, to prove that $\forall D_1,D_2,\interpath{\gamma(D_1)}=\interpath{\gamma(D_2)}\Rightarrow\interpath{D_1}=\interpath{D_2}$, it suffices to prove that for any two words $w_1,w_2\in\GG^*$, if $\gamma(w_1)=\gamma(w_2)$ then $w_1=w_2$.

We first prove \cref{interpretationuniverselle}.

By Zorn's lemma, there exists a maximal family $(\alpha_i)_{i\in I}$ of $\QQ$-algebraically independent complex numbers of absolute value $1$. Such a family must have the cardinality of $\CC$ (that is, $2^{\aleph_0}$). Indeed, the cardinality of the set of polynomials in one variable with coefficients in the field extension of $\QQ$ generated by the $\alpha_i$, is $\max(\aleph_0,\mathrm{card}(I))$, and since each of these polynomials has finitely many roots, the set of their roots has cardinality at most $\max(\aleph_0,\mathrm{card}(I))$. If $\mathrm{card}(I)$ is strictly less than $2^{\aleph_0}$, then so is $\max(\aleph_0,\mathrm{card}(I))$; therefore, since the set $\bigl\{\alpha\in\CC~\bigm|~|\alpha|=1\bigr\}$ has cardinality $2^{\aleph_0}$, it contains an element $\alpha_\bot$ which is not a root of any of these polynomials, so that by adding $\alpha_\bot$ to the family $(\alpha_i)_{i\in I}$, we still have a family of $\QQ$-algebraically independent complex numbers of absolute value $1$, which contradicts the maximality of $(\alpha_i)_{i\in I}$\bigskip.

If the cardinality of $\GG$ is no greater than $2^{\aleph_0}$, then without loss of generality, we can assume that $\GG\subseteq I$. We start with the case where $\H=\CC^2$. We consider the function $\gamma\colon U\nolinebreak\in\nolinebreak\GG\mapsto H\begin{pmatrix}1&0\\0&\alpha_U\end{pmatrix}$, extended into a monoid homomorphism $\gamma\colon \GG^*\to\UU(\CC^2)$ (where $H=\frac{1}{\sqrt2}\begin{pmatrix}1&1\\1&-1\end{pmatrix}$). Given two words $w_1,w_2\in\GG^*$ such that $w_1\neq w_2$, the entries of $\gamma(w_1)$ and $\gamma(w_2)$ are polynomials in the $\alpha_U$ with coefficients in $\QQ$. The two matrices of polynomials obtained by replacing each $\alpha_U$ by a variable $X_U$ in $\gamma(w_1)$ and $\gamma(w_2)$ differ by at least one entry: indeed, by instantiating each variable $X_U$ by either $e^{\mathrm{i}\pi/4}$ or $e^{3\mathrm{i}\pi/4}$ in such a way that the sequence of angles induced by $w_1$ and $w_2$ are different, we get two different sequences of the patterns $HT$ and $HTS$ with $T=\begin{pmatrix}1&0\\0&e^{\mathrm{i}\pi/4}\end{pmatrix}$ and $S=T^2$, and it follows from
Theorem 4.1 of \cite{giles2019remarks} (which is Theorem 1(II) of \cite{matsumoto2008representation}) that these two products of matrices have distinct values.\footnote{Indeed, the regular expression $(HT|HTS)^*$ describes the same set of words as $\epsilon|(HT(HT|SHT)^*(\epsilon|S))$, which, since both the identity operator and $S$ belong to the Clifford group, clearly describes a subset of the Matsumoto-Amano normal forms defined in \cite{giles2019remarks} (Equation (2)).} Since the $\alpha_U$ are algebraically independent, this implies that $\gamma(w_1)\neq \gamma(w_2)$.

Still in the case where the cardinality of $\GG$ is no greater than $2^{\aleph_0}$, if $\H\neq\CC^2$, then it suffices to consider a subspace of $\H$ of dimension $2$, and to define for any $U\in\G$, $\gamma(U)$ as having matrix $H\begin{pmatrix}1&0\\0&\alpha_U\end{pmatrix}$ on this subspace (in an arbitrary, fixed, othonormal basis) and as being the identity on the orthogonal complement\bigskip.

If the (Hilbert) dimension of $\H$ is strictly greater than $\aleph_0$, then Zorn's lemma implies that $\H$ can be decomposed into a direct sum of $\dim(\H)$ orthogonal subspaces of dimension $2$: $\H=\bigoplus_{j\in J}\H_j$ with $\mathrm{card}(J)=\dim(\H)$ and $\forall j,\dim(\H_j)=2$. For each of the $(2^{\aleph_0})^{\dim(\H)}=2^{\dim(\H)}$ possible families $(i_j)_{j\in J}$ of elements of $I$ indexed by $J$, one can define a linear map $\delta\!\left((i_j)_{j\in J}\right)\in\UU(\H)$ as having matrix $H\begin{pmatrix}1&0\\0&\alpha_i\end{pmatrix}$ in an arbitrary orthonormal basis of $\H_j$ (chosen with the help of the axiom of choice) for every $j$. If the cardinality of $\GG$ is no greater than $2^{\dim(\H)}$, then without loss of generality, we can assume that $\GG\subseteq I^J$. We define the function $\gamma\colon\GG\to\UU(\H)$ by $\forall U,\gamma(U)=\delta(U)$, and extend it into a monoid homomorphism $\gamma\colon \GG^*\to\UU(\CC^2)$. Given two words $w_1,w_2\in\GG^*$ such that $w_1\neq w_2$, there exists an index $j\in J$ such that the two sequence of elements of $i$ induced by $w_1$ and $w_2$ at index $j$ are distinct, which, by the argument given above, implies that the unitary maps on $\H_j$ induced respectively by $\gamma(w_1)$ and $\gamma(w_2)$ are distinct. Hence, $\gamma(w_1)\neq \gamma(w_2)$\bigskip.

Finally, to prove Proposition~\ref{matricescompletes} without using the axiom of choice, it suffices to exhibit an infinite family of $\QQ$-algebraically independent complex numbers of absolute value $1$. One can consider for example the $e^{\mathrm{i}\pi^k}$, for $k\geq 2$, whose algebraic independence follows from the Lindemann-Weierstrass theorem and the fact that $\pi$ is transcendental. Given such a family, one can use a similar argument as above to prove a weaker version of Proposition~\ref{interpretationuniverselle} in which the cardinality of $\GG$ is required to be at most $\aleph_0$, which immediately implies Proposition~\ref{matricescompletes}.

\subsection{Proof of Theorem \ref{existenceNF}}\label{preuveexistenceNF}

The proof is by structural induction on $D$.
\begin{itemize}
\item \ntikzfig{diagrammevide-xs}, \begin{tikzpicture}
	\begin{pgfonlayer}{nodelayer}
		\node [style=none] (0) at (-0.6, 0) {};
		\node [style=none] (1) at (0.6, 0) {};
		\node [style=none] (2) at (0, 0.275) {\footnotesize $\vtype$};
	\end{pgfonlayer}
	\begin{pgfonlayer}{edgelayer}
		\draw [style=rouge] (0.center) to (1.center);
	\end{pgfonlayer}
\end{tikzpicture}
, \begin{tikzpicture}
	\begin{pgfonlayer}{nodelayer}
		\node [style=none] (0) at (-0.6, 0) {};
		\node [style=none] (1) at (0.6, 0) {};
		\node [style=none] (2) at (0, 0.325) {\footnotesize $\htype$};
	\end{pgfonlayer}
	\begin{pgfonlayer}{edgelayer}
		\draw [style=bleue] (0.center) to (1.center);
	\end{pgfonlayer}
\end{tikzpicture}
, \xntikzfig{cnegRB}, \xntikzfig{cnegBR}, \xgtikzfig{cgateUR}, \xgtikzfig{cgateUB}, \xstikzfig{beamsplitterNRwB}, \xstikzfig{beamsplitterRNBw}, \xstikzfig{swapRR-s-label}, \xstikzfig{swapBB-s-label}, \xstikzfig{swapRB-s-label} and \xstikzfig{swapBR-s-label} are already in normal form.

\item The normal forms of \xstikzfig{filcourt}, \xstikzfig{beamsplitterwBNR}, \xstikzfig{beamsplitterBwRN}, \xstikzfig{beamsplitter}, \xstikzfig{beamsplitterNRRN-s-label}, \xstikzfig{beamsplitterRNNR-s-label}, \xstikzfig{beamsplitterNNBB-s-label} and \xstikzfig{beamsplitterBBNN-s-label} are given by Equation \eqref{bsbssplitmerge}, \eqref{bsxcol}, \eqref{xbscol}, \eqref{bsbigebre}, \eqref{bsnrrn}, \eqref{bsrnnr}, \eqref{bsnnbb} and \eqref{bsbbnn} respectively.

\item The normal forms of \xntikzfig{neg}, \xgtikzfig{gateU}, \xstikzfig{swap}, \xstikzfig{swapnoirR-s-label}, \xstikzfig{swapRnoir-s-label}, \xstikzfig{swapnoirB-s-label} and \xstikzfig{swapBnoir-s-label} are obtained as follows:\bigskip

\begin{longtable}{RCL@{\qquad}RCL}\begin{tikzpicture}
	\begin{pgfonlayer}{nodelayer}
		\node [style=cartouche] (0) at (0, 0) {$\neg$};
		\node [style=none] (1) at (-1.5, 0) {};
		\node [style=none] (2) at (1.5, 0) {};
	\end{pgfonlayer}
	\begin{pgfonlayer}{edgelayer}
		\draw [style=noire] (1.center) to (0);
		\draw [style=noire] (0) to (2.center);
	\end{pgfonlayer}
\end{tikzpicture}
&\eqeqref{bsbssplitmerge}&\input{negbsRBbs.tikz}&\begin{tikzpicture}
	\begin{pgfonlayer}{nodelayer}
		\node [style=newe] (0) at (0, 0) {$U$};
		\node [style=none] (1) at (-1.5, 0) {};
		\node [style=none] (2) at (1.5, 0) {};
	\end{pgfonlayer}
	\begin{pgfonlayer}{edgelayer}
		\draw [style=new] (1.center) to (2.center);
	\end{pgfonlayer}
\end{tikzpicture}&\eqeqref{bsbssplitmerge}&\input{UbsRBbs.tikz}\\\\
&\eqeqref{nbscol}&\input{bsBRnegnegRBbs.tikz}&&\eqeqref{ubscol}&\input{gateUNF.tikz}\\\\
&\eqeqref{bsxcol}&\input{negNFmiroir.tikz}
\\\\\\\\\\
\input{swapnoirR-label.tikz}&\eqeqref{bsbssplitmerge}&\input{swapnoirRNF-label.tikz}&\input{swapRnoir-label.tikz}&\eqeqref{bsbssplitmerge}&\input{swapRnoirNF-label.tikz}\\\\\\\\\\
\input{swapnoirB-label.tikz}&\eqeqref{bsbssplitmerge}&\input{swapnoirBNF-label.tikz}&\input{swapBnoir-label.tikz}&\eqeqref{bsbssplitmerge}&\input{swapBnoirNF-label.tikz}
\end{longtable}\vspace{1.5em}
\[\input{swap.tikz}\quad\eqeqref{bsbssplitmerge}\quad\input{swapNF.tikz}\]
\item If $D=D_1\oplus D_2$, then by induction hypothesis there exist two diagrams in normal form $N_1$ and $N_2$ such that $\CPBS\vdash D_1=N_1$ and $\CPBS\vdash D_2=N_2$. Then $\CPBS\vdash D=N_1\oplus N_2$ and it is easy to see that $N_1\oplus N_2$ is in normal form.

\item If $D=D_2\circ D_1$, then by induction hypothesis, let $N_1$ and $N_2$ be two diagrams in normal form such that $\CPBS\vdash D_1=N_1$ and $\CPBS\vdash D_2=N_2$. Let us decompose them as $N_1=M_1\circ P_1\circ F_1\circ G_1\circ S_1$ and $N_2=M_2\circ P_2\circ F_2\circ G_2\circ S_2$, following \cref{defNFcoloree}. One has $\CPBS\vdash D=N_2\circ N_1=M_2\circ P_2\circ F_2\circ G_2\circ S_2\circ M_1\circ P_1\circ F_1\circ G_1\circ S_1$. \Cref{bsbsmergesplit} makes $S_2\circ M_1$ equal to a parallel composition of red and blue identity wires, so that $\CPBS\vdash D=M_2\circ P_2\circ F_2\circ G_2\circ P_1\circ F_1\circ G_1\circ S_1$. By naturality of the swap, one has $G_2\circ P_1=P_1\circ G'_2$, where $G'_2$ is a parallel composition of coloured non-identity gates and identity wires, obtained by permuting the ``rows'' of $G_2$. One has $\eqref{neguRB},\eqref{neguBR}\vdash G'_2\circ F_1=F_1\circ G''_2$, where $G''_2$ is obtained by changing some colours in $G'_2$, and \cref{neguBR} is the following variant of \cref{neguRB}:
\begin{equation}\label{neguBR}\input{UnegBR-label.tikz}\ =\ \input{negBRU-label.tikz}\end{equation}
which is derived from the equations of \cref{axiomsCPBS} as follows:\medskip
\begin{longtable}{RCL}\input{negBRU-label.tikz}&\eqeqref{negnegRBR}&\input{negUnegnegBRBR.tikz}\\
&\eqeqref{neguRB}&\input{negnegUnegBRBR.tikz}\\
&\eqeqref{negnegBRB}&\input{UnegBR-label.tikz}\end{longtable}
Thus, $\CPBS\vdash D=M_2\circ P_2\circ F_2\circ P_1\circ F_1\circ G''_2\circ G_1\circ S_1$. By naturality of the swap, one has $F_2\circ P_1=P_1\circ F'_2$, where $F'_2$ is a parallel composition of coloured identities and negations (obtained by permuting $F_2$). One has $\eqref{negnegRBR},\eqref{negnegBRB}\vdash F'_2\circ F_1=F''$, where $F''$ is obtained by removing all double negations in $F'_2\circ F_1$. Finally, $\eqref{fusionUVR},\eqref{idboxR}\vdash G''_2\circ G_1=G'''$, where $G'''$  is still a parallel composition of coloured non-identity gates and identity wires. Thus,  $\CPBS\vdash D=M_2\circ (P_2\circ P_1)\circ F''\circ G'''\circ S_1$, with $S_1$, $G'''$, $F''$, $(P_2\circ P_1)$ and $M_2$ respectively of the forms described in \cref{defNFcoloree}, so that their composition is in normal form. This gives us the result.
\item If $D=Tr_\vtype(D'):a\to b$, then by induction hypothesis, let $N'$ be a diagram in normal form such that $\CPBS\vdash D'=N'$. Let us decompose it as $N'=M'\circ P'\circ F'\circ G'\circ S'$, following \cref{defNFcoloree}. Since $N'$ is of type $a\oplus\vtype\to b\oplus\vtype$, $S'$ (resp. $M'$) is of the form $S\oplus\ntikzfig{cfilcourtR-s}$ (resp. $M\oplus\ntikzfig{cfilcourtR-s}$) where $S$ (resp. $M$) is a parallel composition of coloured identity wires and copies of $\xstikzfig{beamsplitterNRwB}$ (resp. $\xstikzfig{beamsplitterRNBw}$). Using the structural congruence, one can write $P'$ in the form $\ntikzfig{PsurfilR}$ or $\ntikzfig{permswaprougeenbas}$, where $P$, or $P_1$, $P_2$ and $P_3$, are permutations of the wires. In the first case, $Tr_\vtype(N')$ can (still using the structural congruence) be written in the form \[\ntikzfig{SGFPMsurbouclevideUR}\] with $S$, $G$, $F$, $P$ and $M$ of the forms demanded by \cref{defNFcoloree} (in particular, $F'$ cannot have a negation on its bottom wire since this would prevent $N'$ from being of type $a\oplus\vtype\to b\oplus\vtype$), so that $\eqref{bouclevideUR}\vdash Tr_\vtype(N')=\ntikzfig{SGFPM}$, which is in normal form. In the second case, $Tr_\vtype(N')$ can be written in the form \[\ntikzfig{SGFP1P2P3M}\] where $\xntikzfig{negpotentielRa}$ is either $\ntikzfig{cfilcourtR-s}$ or $\xntikzfig{cnegRB}$. Then using the structural congruence (in particular the yanking axiom), one can write it in the form \[\ntikzfig{SGFP1P2P3Metire}.\] By naturality of the swap, one can slide the gate $U$ and the possible negation through $P_1$. Then, possibly using \cref{neguBR}, one can move the gate $U$ to the other side of $F$. Finally, it may remain to merge $U$ with a gate of $G$ using \cref{fusionUVR} or its following variant:
\begin{equation}\label{fusionUVBgenerale}\input{gateUgateVbleu-t.tikz}=\begin{tikzpicture}
	\begin{pgfonlayer}{nodelayer}
		\node [style=newebleu] (0) at (0, 0) {$VU$};
		\node [style=none] (1) at (-1.8, 0) {};
		\node [style=none] (2) at (1.7, 0) {};
		\node [style=none] (3) at (-1.5, 0.45) {\footnotesize \htype};
	\end{pgfonlayer}
	\begin{pgfonlayer}{edgelayer}
		\draw [style=bleue] (1.center) to (2.center);
	\end{pgfonlayer}
\end{tikzpicture}
\end{equation}
and/or to remove a double negation using \cref{negnegBRB}. Then one gets a diagram in normal form.

\Cref{fusionUVBgenerale} is derived from the equations of \cref{axiomsCPBS} as follows:\medskip
\begin{longtable}{RCL}
\input{gateUgateVbleu-t.tikz}&\eqdeuxeqref{negnegBRB}{neguRB}&\input{neggateUgateVrougeneg-t.tikz}\\\\
&\eqeqref{fusionUVR}&\input{neggateVUrougeneg-t.tikz}\\\\
&\eqdeuxeqref{neguRB}{negnegBRB}&\begin{tikzpicture}
	\begin{pgfonlayer}{nodelayer}
		\node [style=newebleu] (0) at (0, 0) {$VU$};
		\node [style=none] (1) at (-1.8, 0) {};
		\node [style=none] (2) at (1.7, 0) {};
		\node [style=none] (3) at (-1.5, 0.45) {\footnotesize \htype};
	\end{pgfonlayer}
	\begin{pgfonlayer}{edgelayer}
		\draw [style=bleue] (1.center) to (2.center);
	\end{pgfonlayer}
\end{tikzpicture}

\end{longtable}\medskip

\item The case $D=Tr_\htype(D'):a\to b$ is analogous to the previous case. Instead of using Equations \eqref{neguRB} and \eqref{bouclevideUR}, one uses respectively \cref{neguBR} and the following variant of \cref{bouclevideUR}:
\begin{equation}\label{bouclevideUB}\ptikzfig{bouclevideUB-label}{0.65}{0.74286}{1}\ =\ \ptikzfig{diagrammevide}1{0.5}{0.5}\end{equation}
which is derived from the equations of \cref{axiomsCPBS} as follows:\medskip

\begin{longtable}{RCL}\ptikzfig{bouclevideUB-label}{0.65}{0.74286}{1}&\eqeqref{negnegBRB}&\ptikzfig{bouclevideUBnegneg-label}{0.65}{0.74286}{1}\\\\
&=&\ptikzfig{bouclevidenegUBneg-label}{0.65}{0.74286}{1}\\\\
&\eqeqref{neguRB}&\ptikzfig{bouclevideURnegneg-label}{0.65}{0.74286}{1}\\\\
&\eqeqref{negnegRBR}&\ptikzfig{bouclevideUR-label}{0.65}{0.74286}{1}\\\\
&\eqeqref{bouclevideUR}&\ptikzfig{diagrammevide}1{0.5}{0.465}\end{longtable}\medskip

\item If $D=Tr_\top(D'):a\to b$, then by induction hypothesis, let $N'$ be a diagram in normal form such that $\CPBS\vdash D'=N'$. $Tr_\top(N')$ can be written in the form $\ntikzfig{tracebsNRwBNsecondebsRNBw}$, which by dinaturality and \cref{bsbsmergesplit}, can be transformed into $\ntikzfig{traceRBNseconde}$. It suffices then to proceed successively as in the two preceding cases to get a diagram in normal form.
\end{itemize}

\subsection{Proof of Lemma \ref{uniquenessNF}}\label{preuveuniquenessNF}

If $\interpath{N}=\interpath{N'}$, then in particular $N$ and $N'$ have same type: $N,N':a\to b$ for some $a,b$.

Let us decompose $N$ and $N'$ into $N=M\circ P\circ F\circ G\circ S$ and $N'=M'\circ P'\circ F'\circ G'\circ S'$.

It follows directly from the definition that $S$ and $S'$ are uniquely determined by their input type, so that since they both have input type $a$, $S=S'$. Similarly, $M$ and $M'$ are uniquely determined by their output type, so that since they both have output type $b$, $M=M'$.

Let $S^{-1}$ and $M^{-1}$ be the horizontal reflections of respectively $S$ and $M$, that is, the diagrams obtained by replacing $\xstikzfig{beamsplitterNRwB}$ by $\xstikzfig{beamsplitterRNBw}$ in $S$ and $\xstikzfig{beamsplitterRNBw}$ by $\xstikzfig{beamsplitterNRwB}$ in $M$. One has $\eqref{bsbsmergesplit}\vdash M^{-1}\circ N\circ S^{-1} = P\circ F\circ G$ and $\eqref{bsbsmergesplit}\vdash M^{-1}\circ N'\circ S^{-1} = P'\circ F'\circ G'$, so that by \cref{soundnessCPBS}, $\interpath{M^{-1}\circ N\circ S^{-1}} = \interpath{P\circ F\circ G}=\interpath{M^{-1}\circ N'\circ S^{-1}} = \interpath{P'\circ F'\circ G'}$. For any $c,p$, one has $\interpath{P\circ F\circ G}(c,p)=((c^F_{c,p},p^P_{c,p}),U^G_{c,p})$ and
$\interpath{P'\circ F'\circ G'}(c,p)=((c^{F'}_{c,p},p^{P'}_{c,p}),U^{G'}_{c,p})$, so that $U^{G}_{c,p}=U^{G'}_{c,p}$, $c^{F}_{c,p}=c^{F'}_{c,p}$ and $p^{P}_{c,p}=p^{P'}_{c,p}$. Because of their respective forms required by \cref{defNFcoloree}, \linebreak[3]$G$, $G'$, $F$, $F'$, $P$ and $P'$ are uniquely determined by the family of, respectively, the $U^{G}_{c,p}$, the $U^{G'}_{c,p}$, the $c^{F}_{c,p}$, the $c^{F'}_{c,p}$, the $p^{P}_{c,p}$, and the $p^{P'}_{c,p}$. Hence, $G=G'$, $F=F'$ and $P=P'$.

\subsection{Proof of Theorem \ref{minimalityCPBS}}\label{preuveminimalityCPBS}

For
\begin{itemize}
\item each of \cref{idboxR,nbscol,negnegRBR,negnegBRB,bsbssplitmerge,bsbsmergesplit,bsxcol,xbscol,bsbigebre,bsnrrn,bsrnnr,bsnnbb,bsbbnn}
\item each instance of Equations \eqref{neguRB} and \eqref{ubscol}
\item the class of all instances of Equation \eqref{fusionUVR} without $I$ gates in the left-hand side
\item each class of instances of Equation \eqref{bouclevideUR} given by an equivalence class of elements of $\M$ for the equivalence relation $\sim_{\mathrm{conj}}^*$, defined as the transitive closure of $\sim_{\mathrm{conj}}$, itself defined by $U\sim_{\mathrm{conj}}V$ if there exist $W,T\in\M$ such that $U=WT$ and $V=TW$
\end{itemize}
we give an invariant that is satisfied by exactly one side of the considered equation (or of each element of the considered class of instances of Equation \eqref{fusionUVR} or \eqref{bouclevideUR}), and such that for any diagram $D$, applying any other equation or instance inside $D$ (that is, replacing a sub-diagram of $D$ that matches one side of the equation by the other side) preserves the fact that $D$ satisfies the invariant or not. In each case, this proves that the equations that break the invariant are not consequences of those that preserve it in any diagram.

Note that the instances of Equation \eqref{fusionUVR} with an $I$ gate in the left-hand side are consequences of Equation \eqref{idboxR}, and that the elements of a class of instances of Equation \eqref{bouclevideUR} are consequences of any particular instance of Equation \eqref{bouclevideUR} of the same class together with Equation \eqref{fusionUVR}.

\begin{itemize}
\item For Equation \eqref{idboxR}, the invariant is that at least one gate can be reached by a particle from an input wire.
\item For the class of all instances of Equation \eqref{fusionUVR} without $I$ gates in the left-hand side, the invariant is the maximum number of non-$I$ gates that a particle coming from an input wire can traverse along its path in the diagram.
\item For each instance of Equation \eqref{neguRB} given by a particular $U$, the invariant is that all gates labelled with $U$ are red.
\item For Equation \eqref{nbscol}, the invariant is that the diagram contains a (black) $\xntikzfig{neg}$.
\item For each instance of Equation \eqref{ubscol} given by a particular $U$, the invariant is that the diagram contains a (black) $\xgtikzfig{gateU}$.
\item For each class of instances of Equation \eqref{bouclevideUR}, the invariant is that there exists a wire in the diagram and a polarisation $\Vpol$ or $\Hpol$ such that the path of a particle starting from this wire with this polarisation is a closed loop, and that the product of the labels of the gates traversed by the particle before getting back to its starting point with its initial polarisation for the first time, is an element of the equivalence class (note that this does not depend on the choice of the starting point).
\item For Equation \eqref{negnegRBR}, the invariant is that all wires are red.
\item For Equation \eqref{negnegBRB}, the invariant is that no particle entering the diagram by a blue input wire can reach the output without passing through a negation at some point in the diagram. Note that \cref{negnegRBR} cannot change this invariant because in order to reach a red wire, the particle coming from a blue wire has to get its polarisation changed, and therefore to pass through a negation.
\item For Equation \eqref{bsbssplitmerge}, the invariant is that all wires are black and the diagram is non-empty and does not contain any $\xstikzfig{beamsplitter}$.
\item For Equation \eqref{bsbsmergesplit}, the invariant is that the diagram contains at least one black wire.
\item For Equation \eqref{bsxcol}, the invariant is that the diagram contains at least one generator among $\xstikzfig{beamsplitterwBNR}$, $\xstikzfig{beamsplitter}$, $\xstikzfig{beamsplitterRNNR-s-label}$, $\xstikzfig{beamsplitterBBNN-s-label}$ and $\xntikzfig{neg}$.
\item For Equation \eqref{xbscol}, the invariant is that the diagram contains at least one generator among $\xstikzfig{beamsplitterBwRN}$, $\xstikzfig{beamsplitter}$, $\xstikzfig{beamsplitterNRRN-s-label}$ and $\xstikzfig{beamsplitterBBNN-s-label}$.
\item For Equation \eqref{bsbigebre}, the invariant is that the diagram contains a $\xstikzfig{beamsplitter}$.
\item For Equation \eqref{bsnrrn}, the invariant is that the diagram contains a $\xstikzfig{beamsplitterNRRN-s-label}$.
\item For Equation \eqref{bsrnnr}, the invariant is that the diagram contains a $\xstikzfig{beamsplitterRNNR-s-label}$.
\item For Equation \eqref{bsnnbb}, the invariant is that the diagram contains a $\xstikzfig{beamsplitterNNBB-s-label}$.
\item For Equation \eqref{bsbbnn}, the invariant is that the diagram contains a $\xstikzfig{beamsplitterBBNN-s-label}$.
\end{itemize}

\subsection{Proof of Theorem \ref{stairsoptimal}}\label{preuvestairsoptimal}

Given any gate-free diagram $Q:d\to e$, we denote by $\{d^{Q}_i\}_{i=1,...,k_Q}$ the finest partition of $\{0,...,|d|-1\}$ such that there exists a partition $\{e^Q_i\}_{i=1,...,k_Q}$ of $\{0,...,|e|-1\}$ satisfying $\forall i,\forall c,p,(p\in d^{Q}_i\Leftrightarrow p^P_{c,p}\in e^{Q}_i)$. It is easy to see that the partition $\{e^Q_i\}_{i=1,...,k_Q}$ is unique and that symmetrically, it is the finest partition of $\{0,...,|e|-1\}$ such that there exists a partition $\{{f}^{Q}_i\}_{i=1,...,k_Q}$ of $\{0,...,|d|-1\}$ satisfying $\forall i,\forall c,p,(p\in {f}^Q_i\Leftrightarrow p^P_{c,p}\in e^{Q}_i)$ (which of course implies that $\forall i,{f}^Q_i=d^Q_i$)\bigskip.

Given $P=\ptikzfig{superpetpermsnegpotentiels-xl}{1.085}{0.7}1:d\to e$ a diagram in stair form, let $\{d_i\}_{i=1,...,k}$ be the partition of $\{0,...,|d|-1\}$ such that an index $j$ is in $d_i$ if the input wire of $P$ of index $j$ is connected to $C_i$. Similarly, let $\{e_i\}_{i=1,...,k}$ be the partition of $\{0,...,|e|-1\}$ such that an index $j$ is in $e_i$ if the output wire of $P$ of index $j$ is connected to $C_i$. One has $\forall i,\forall c,p,(p\in d_i\Leftrightarrow p^P_{c,p}\in e_i)$. It is easy to see that $\{d_i\}_{i=1,...,k}$ is the finest partition of $\{0,...,|d|-1\}$ such that there exists $\{e_i\}_{i=1,...,k}$ satisfying this property, that is, up to reordering the partitions, one has $k=k_P$ and $\forall i,d_i= d_i^{P}\text{ and }e_i=e_i^P$\bigskip.

Again given an arbitrary gate-free diagram $Q:d\to e$, let us decompose $d=x_1\oplus\cdots\oplus x_n$ and $e=y_1\oplus\cdots\oplus y_m$, with $\forall j, x_j,y_j\in\{\vtype,\htype,\top\}$. Since any gate-free diagram is equivalent to a diagram in stair form (indeed, by applying Steps \ref{splitallPBS} to \ref{recomposestairs} of the PGT procedure described below --- which does not rely on \cref{stairsoptimal} --- one can put any gate-free diagram in stair form), the preceding paragraph, because of the input/output types of the five kinds of staircases, implies that for every $i$ there are four cases:
\begin{enumerate}
\item $\left|{d}^Q_i\right|=\left|{e}^Q_i\right|$,\quad $\forall j\in {d}^Q_i, x_j=\top$\quad and\quad $\forall j\in {e}^Q_i, y_j=\top$
\item\label{equilibrecolore} $\left|{d}^Q_i\right|=\left|{e}^Q_i\right|$\quad and exactly one element of ${d}^Q_i$ and one element of ${e}^Q_i$ are not equal to $\top$
\item\label{splitL} $\left|{d}^Q_i\right|=\left|{e}^Q_i\right|+1$,\quad $\forall j\in {e}^Q_i, y_j=\top$\quad and exactly two elements of ${d}^Q_i$ are not equal to $\top$
\item\label{splitR} $\left|{e}^Q_i\right|=\left|{d}^Q_i\right|+1$,\quad $\forall j\in {d}^Q_i, x_j=\top$\quad and exactly two elements of ${e}^Q_i$ are not equal to $\top$
\end{enumerate}

We denote by $s_L(Q)$ the number of indices $i$ for which we are in Case \ref{splitL}\medskip.

Moreover, by examining more in details the semantics of the five kinds of staircases, one can show that for every index $i\in\{1,...,k_Q\}$, there exists two bijections $\rho_i\colon\ZZ/\!\left(\left|d^{Q}_i\right|\ZZ\right)\to d^{Q}_i$ and $\tau_i\colon\ZZ/\!\left(\left|e^{Q}_i\right|\ZZ\right)\to e^{Q}_i$ such that for any $p\in\{1,...,|d^{Q}_i|\}$, if $(\Vpol,\rho_i(\pi(p)))\in[d]$ then $\interp{Q}\left(\Vpol,\rho_i(\pi(p))\right)=\left(\left(c^Q_{\Vpol,\rho_i(\pi(p))},\tau_i(\pi(p))\right),I\right)$, and if $(\Hpol,\rho_i(\pi(p)))\in[d]$ then\linebreak $\interp{Q}\left(\Hpol,\rho_i(\pi(p))\right)=\left(\left(c^Q_{\Hpol,\rho_i(\pi(p))},\tau_i(\pi(p+1))\right),I\right)$, where $\pi\colon\ZZ\to\ZZ/k\ZZ$ denotes the canonical projection.

Concrete instances of the bijections $\rho_i$ and $\tau_i$ can be built by starting from any element $j\in d^Q_i$ and defining $\rho_i(1)=x$. Then, the properties of $\rho_i$ and $\tau_i$ imply that knowing the action semantics of $Q$, for any $p\in\ZZ/\left(\left|d^{Q}_i\right|\ZZ\right)$, the data of $\rho_i(p)$ uniquely determines $\tau_i(p)$ and $\tau_i(p+1)$, and the data of $\tau_i(p)$ uniquely determines $\rho_i(p)$ and $\rho_i(p-1)$, so that $\rho_i$ and $\tau_i$ can be built incrementally\bigskip.

It is easy to see that given a diagram in stair form $P:d\to e$, one has $\nopbs(P)= |e|-k_P+s_L(P)$. In the rest of this proof, our goal is to prove that for any gate-free diagram $Q:d\to e$, one has $\nopbs(Q)\geq |e|-k_Q+s_L(Q)$. Then, since $|e|-k_Q+s_L(Q)$ only depends on the semantics of $Q$, and diagrams in stair form reach this lower bound, this will imply that they are PBS-optimal\medskip.

Since any gate-free diagram $Q:d\to e$ can be deformed into a diagram of the form \eqref{snail} with $P$ trace-free, it suffices to prove, on the one hand, that the inequality holds for trace-free diagrams, and on the other hand, that it is preserved by the trace operation\medskip. 

To prove that the trace preserves the inequality, given a gate-free diagram $Q:d\oplus a\to e\oplus a$ with $a\in\{\vtype, \htype, \top\}$, it suffices to consider the sets $d^{Q}_i$ and $e^{Q}_j$ that contain the index of the bottom input (resp. output) wire, and to examine the possible cases (essentially, whether $i=j$, to which of the four cases described above the pairs $(d^{Q}_i,e^{Q}_i)$ and $(d^{Q}_j,e^{Q}_j)$ correspond, and whether the traced wire has type $\top$ or not). In each case, it suffices to build the bijections $\rho_i$ and $\tau_i$ and to look at the effect of applying the trace\medskip.

To prove that it holds for trace-free diagrams, we remark that up to deformation, a trace-free (gate-free) diagram can be written as a sequential composition of diagrams of the form $id_f\oplus g\oplus id_{f'}$ with $f,f'\in\{\vtype, \htype, \top\}^*$ and $g\in\biggl\{\stikzfig{beamsplitter},\linebreak[0]\stikzfig{beamsplitterNRRN-s-label},\linebreak[0]\stikzfig{beamsplitterRNNR-s-label},\linebreak[0]\stikzfig{beamsplitterNNBB-s-label},\linebreak[0]\stikzfig{beamsplitterBBNN-s-label},\linebreak[0]\stikzfig{beamsplitterRNBw},\linebreak[0]\stikzfig{beamsplitterBwRN},\linebreak[0]\stikzfig{beamsplitterNRwB},\linebreak[0]\stikzfig{beamsplitterwBNR},\linebreak[0]\sgtikzfig{cneg},\linebreak[0]\stikzfig{cswap}\biggr\}$ (where $id_f$ is inductively defined by $id_{\ov}=\ntikzfig{diagrammevide-xs}$ and $id_{f\oplus a}=id_f\oplus\mtikzfig{cfilcourt-s}$ for any $f\in\{\vtype, \htype, \top\}^*$ and $a\in\{\vtype, \htype, \top\}$). We call such diagrams \emph{layers}. Then we proceed by induction on the number of layers. The base case is that of an identity diagram $id_d:d\to d$, for which $k_{id_d}=|d|$ and $S_L(id_d)=0$, so that the inequality holds. It remains to prove that given any trace-free diagram $Q:d\to e$ satisfying the inequality, and any layer $id_f\oplus g\oplus id_{f'}$ of input type $e$, the composition $(id_f\oplus g\oplus id_{f'})\circ Q$ still satisfies the inequality. This can be done by considering the set(s) $e^Q_i$ and $e^Q_j$ that contain the indice(s) of the wire(s) of $Q$ where $g$ is plugged (together with the corresponding $d^Q_i$ and $d^Q_j$), and examining the possible cases. In each case, it suffices to build the bijections $\rho_i$ and $\tau_i$ and to look at the effect of appending the layer $(id_f\oplus g\oplus id_{f'})$.

\subsection{Proof of Theorem \ref{optimalitequanduneportedechaque}}\label{preuveoptimalitequanduneportedechaque}

Let $D_1:a\to b$ be an abstract diagram obtained from applying the query optimisation procedure followed by the PGT procedure, in which all gates bear different labels. We write it in the form
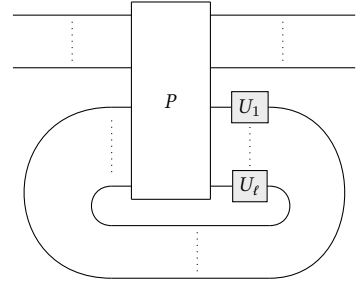
\[\input{CPNFnoire.tikz}\]
with $P$ in stair form and $U_1, ..., U_\ell\in\GG$ (where all wires and gates can be of arbitrary colours).

For each $(c,p)\in [a]$, let $p^{(1)}_{c,p}, ..., p^{(\ell_{c,p})}_{c,p}$ be the sequence of positions such that $w^{D_1}_{c,p}=U_{p^{(1)}_{c,p}}...U_{p^{(\ell_{c,p})}_{c,p}}$ (with $\ell_{c,p}=|w^{D_1}_{c,p}|$). This sequence is determined without ambiguity since the names $U_i$ are pairwise distinct. There exists a sequence of polarisations $c^{(1)}_{c,p}, ..., c^{(\ell_{c,p})}_{c,p}$ such that $\interpath{P}(c,p)=((c^{(1)}_{c,p},|b|+p^{(1)}_{c,p}),\epsilon)$, $\forall i\in\{1,...,\ell_{c,p}-1\},\interpath{P}(c^{(i)}_{c,p},|a|+p^{(i)}_{c,p})=((c^{(i+1)}_{c,p},|b|+p^{(i+1)}_{c,p}),\epsilon)$, and $\interpath{P}(c^{(\ell_{c,p})}_{c,p},|a|+p^{(\ell_{c,p})}_{c,p})=((c^{D_1}_{c,p},p^{D_1}_{c,p}),\epsilon)$ (where $\epsilon$ denotes the empty word).

Given a query-PBS-optimal diagram $D_1'$ equivalent to $D_1$, up to applying the query optimisation procedure and the PGT procedure, we can assume that $D_1'$ is in PGT form. Note that any diagram $E$ obtained from applying the query optimisation procedure necessarily satisfies that, for every $U\in\GG$, it contains exactly $\left\lfloor\sum_{(c,p)\in [a]}\frac{|w_{c,p}^E|_U}2\right\rfloor$ black gates labelled with $U$, and one red or blue gate labelled with $U$ if and only if $\sum_{(c,p)\in [a]}|w_{c,p}^E|_U$ is odd. Since the PGT procedure does not change the gates, it preserves this property. Therefore, $D_1$ and $D_1'$ both satisfy this property, and since they have the same semantics, this implies that they have the same gates up to turning some red gates into blue gates and vice-versa. That is, up to slightly deforming it in order to permute the gates, we can put $D_1'$ in the form
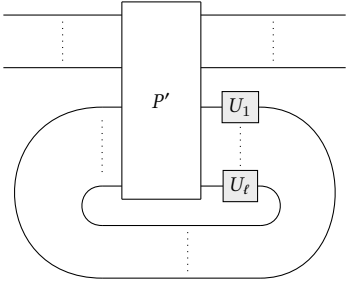
\[\input{CPNFnoirePprime.tikz}\]
with $P'$ in stair form. For each $(c,p)\in[a]$, there also exists a sequence of polarisations $c^{\prime(1)}_{c,p}, ..., c^{\prime(\ell_{c,p})}_{c,p}$ such that $\interpath{P'}(c,p)=((c^{\prime(1)}_{c,p},|b|+p^{(1)}_{c,p}),\epsilon)$, $\forall i\in\{1,...,\ell_{c,p}-1\},\interpath{P'}(c^{\prime(i)}_{c,p},|a|+p^{(i)}_{c,p})=((c^{\prime(i+1)}_{c,p},|b|+p^{(i+1)}_{c,p}),\epsilon)$, and $\interpath{P'}(c^{\prime(\ell_{c,p})}_{c,p},|a|+p^{(\ell_{c,p})}_{c,p})=((c^{D_1}_{c,p},p^{D_1}_{c,p}),\epsilon)$.

Let $P''$ be the diagram obtained from $P'$ by adding, for every position $q$ such that there exist $c,p$ and $i\in\{1,...,\ell_{c,p}\}$ satisfying $q=p^{(i)}_{c,p}$ and $c^{(i)}_{c,p}\neq c^{\prime(i)}_{c,p}$, a negation on input wire $|a|+q$ and on output wire $|b|+q$. Let $d$ be such that $P:a\oplus d\to b\oplus d$. It is easy to see that for every $(c,p)\in[a]$, and for every couple $(c,p)\in[a\oplus d]$ with $p\geq |a|$ that can be written as $(c^{\prime(i)}_{c',p'},|a|+p^{(i)}_{c',p'})$ for some $c',p'\in[a]$ and $i\in\{1,...,\ell_{c',p'}\}$, one has $\interpath{P''}(c,p)=\interpath{P}(c,p)$. Since $D_1$ is query-optimal, every black gate can be reached from two basis states $(c,p)\in[a]$ and every non-black gate can be reached from one basis state, which implies that every couple $(c,p)\in[b\oplus d]$ with $p\geq |b|$ can be written as $(c^{\prime(i)}_{c',p'},|b|+p^{(i)}_{c',p'})$, and therefore, every couple $(c,p)\in[a\oplus d]$ with $p\geq |a|$ can be written as $(c^{\prime(i)}_{c',p'},|a|+p^{(i)}_{c',p'})$. Hence, $P''$ has the same semantics as $P$. Since by construction, $P''$ contains the same number of PBS as $P'$, and by \cref{stairsoptimal}, $P$ is PBS-optimal, this implies that $P$ contains at most as many PBS as $P'$, that is, $D_1$ contains at most as many PBS as $D_1'$. Hence, $D_1$ is query-PBS-optimal.

\begin{remark}\label{rqpasoptimal}
The proof of \cref{optimalitequanduneportedechaque} uses the fact that the diagrams output by the query optimisation procedure, in addition of being query-optimal, have the property that if a gate is used only once (that is, if it is accessible from only one input state $(c,p)$, and a particle with this input state traverses only once the gate), then it is represented as red or blue. Note that a diagram in PGT form with only one query to each oracle may not be query-PBS-optimal if it contains a black gate used only once. For instance, \vspacebeforeline{3pt plus 1pt minus 1pt}$\input{ex-query-pbs-nonopt.tikz}$ is in PGT form but not query-PBS-optimal as it is equivalent to $\begin{tikzpicture}
	\begin{pgfonlayer}{nodelayer}
		\node [style=newerouge] (0) at (0, 0) {$U$};
		\node [style=none] (1) at (-1.8, 0) {};
		\node [style=none] (2) at (1.7, 0) {};
		\node [style=none] (3) at (-1.5, 0.325) {\footnotesize \vtype};
	\end{pgfonlayer}
	\begin{pgfonlayer}{edgelayer}
		\draw [style=rouge] (1.center) to (2.center);
	\end{pgfonlayer}
\end{tikzpicture}$.
\end{remark}

\subsection{NP-Hardness Results}

\subsubsection{Proof of Theorem \ref{gatePBSNPhard}}\label{preuvegatePBSNPhard}

Let $\GG$ be a set of names. We will prove that the problem is already NP-hard when we restrict the input diagram to the family $\mathcal P$ defined as follows:

\begin{definition}
Given a word $w=w_0...w_{n-1}$ with $w_0,...,w_{n-1}\in\GG$ and a permutation $\sigma$ of $[n]$, we define $\sigma(w)$ as the rearranged word $w_{\sigma(0)}...w_{\sigma(n-1)}$.
\end{definition}

\begin{definition}
We denote by $\mathcal P$ the set of $\GG^*$-diagrams $D:\top^{\oplus n}\to\top^{\oplus n}$ such that there exists a word $w=w_0...w_{n-1}\in\GG^n$ and a permutation $\sigma$ of $[n]$ such that for every $p\in[n]$, $\interpath{D}(\Vpol,p)=((\Vpol,p),w_p)$ and $\interpath{D}(\Hpol,p)=((\Hpol,p),w_{\sigma(p)})$.
\end{definition}

We polynomially reduce this restricted problem from the \emph{maximum Eulerian cycle decomposition problem}, also called MAX-ECD \cite{caprara1999eulerian}, which consists in, given an Eulerian undirected graph $G$, finding a maximum-cardinality edge-partition of $G$ into cycles (that is, partitioning the set of edges of $G$ into the maximum number of cycles). Note that the NP-hardness of MAX-ECD follows directly from the NP-completeness of the problem of deciding whether $G$ can be edge-partitioned into triangles, which is proved in \cite{holyer1981edgepartition} (it corresponds to the case of the edge-partition into copies of the complete graph $K_3$).

The MAX-ECD problem is equivalent to the problem of, given an Eulerian graph $G$, 
finding a suitable orientation of its edges together with an edge-partition of the resulting directed graph into directed cycles, so that the number of cycles is maximal among all possible choices of orientation and partition. Indeed, given these, it suffices to erase the directions of the edges to get an undirected edge-partition into cycles, and given such a partition, it suffices to choose, for each cycle, one of the two possible ways of orienting it.

Given an Eulerian graph $G$, we construct a diagram of $\mathcal P$ as follows: first, we choose an arbitrary orientation of the edges of $G$ so as to get an Eulerian directed graph $\vec G$ (which can be done by following an Eulerian circuit of $G$, which itself can be found in polynomial time~\cite{Fleischner:Eulerian1991}) and we associate a label, more precisely an element of $\GG$, with each vertex of $G$, in such a way that any two distinct vertices bear distinct labels. Without loss of generality, we can assume that the vertices of $G$ are elements of $\GG$ and thereby identify them with their labels. We enumerate the edges of $\vec G$ as 
$e_0,...,e_{n-1}$. In $\vec G$ --- since it is Eulerian --- each vertex has in- and out-degree equal, that is, each vertex appears as many times as the head of an arrow as as the tail of an arrow, hence there exists a permutation $\sigma$ of $[n]$ and a word $w=w_0...w_{n-1}\in\GG^n$ such that for any $p\in[n]$, $e_p$ is of the form $(w_p,w_{\sigma(p)})$. We consider the following diagram:
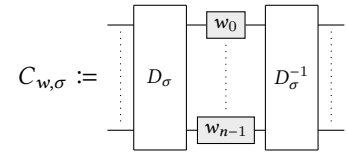
\[C_{w,\sigma}\coloneqq\input{sigmawsigmainv.tikz}\] 
where $D_\sigma:\top^{\oplus n}\to\top^{\oplus n}$ is a $\xntikzfig{neg}$-free diagram in stair form\footnote{Note that the type of $D_\sigma$ forces all of its staircases to be made only of all-black PBS.}
such that for any $p\in[n]$, $\interpath{D_\sigma}(\Vpol,p)=((\Vpol,p),\epsilon)$ and $\interpath{D_\sigma}(\Hpol,p)=((\Hpol,\sigma(p)),\epsilon)$, and $D_\sigma^{-1}$ is the horizontal reflection of $D_\sigma$, which therefore satisfies that for any $p\in[n]$, $\interpath{D_\sigma^{-1}}(\Vpol,p)=((\Vpol,p),\epsilon)$ and $\interpath{D_\sigma^{-1}}(\Hpol,p)=((\Hpol,\sigma^{-1}(p)),\epsilon)$.

For any $p\in[n]$, one has $\interpath{C_{w,\sigma}}(\Vpol,p)=((\Vpol,p),w_p)$ and $\interpath{C_{w,\sigma}}(\Hpol,p)=((\Hpol,p),w_{\sigma(p)})$. In particular, $C_{w,\sigma}$ is in $\mathcal P$, and for any $p\in[n]$, $w^{C_{w,\sigma}}_{{\Vpol},p}$ is the tail of $e_p$ and $w^{C_{w,\sigma}}_{{\Hpol},p}$ is the head of $e_p$.\footnote{See the end of \cref{defactionsem} for the definition of $w^{C_{w,\sigma}}_{c,p}$. Note that we identify words of length $1$ with their single letter.}

Let $C_{w,\sigma}^{\mathrm{opt}}$ be a query-PBS-optimal diagram equivalent to $C_{w,\sigma}$. Up to applying the PGT procedure, which can be done in polynomial time and neither changes the gates nor increases the number of PBS, we can assume that $C_{w,\sigma}^{\mathrm{opt}}$ is in PGT form. That is, up to reordering some wires, it is of the form
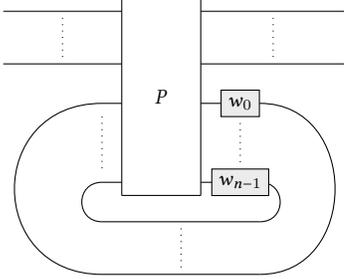
\[\input{PGTw.tikz}\]
with $P$ in stair form. Since for every $c,p$, the word $w^{C_{w,\sigma}^{\mathrm{opt}}}_{c,p}$ has length $1$, $P$ is such that for any $c\in\hv$ and $p\in[n]$, one has $p^P_{c,p}\in\{n,...,2n-1\}$ and $p^P_{c,p+n}\in[n]$.

By looking at the semantics of a generic diagram in stair form (in particular by considering the functions $\rho_i$ and $\tau_i$ defined in the proof of \cref{stairsoptimal}), it is easy to see that this implies that up to reordering the wires on the sides of $P$, we can write $C_{w,\sigma}^{\mathrm{opt}}$ in the form
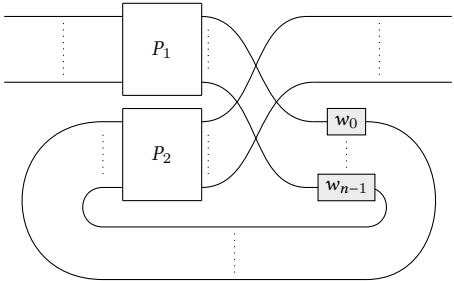
\[\input{PGTwsplit.tikz}\]
where $P_1$ and $P_2$ are two diagrams in stair form. Up to a few more deformations, $C_{w,\sigma}^{\mathrm{opt}}$ is of the form
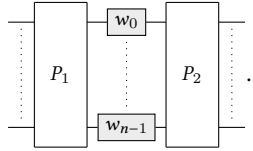
\[\input{P1wP2.tikz}.\]
Due to the semantics of $C_{w,\sigma}^{\mathrm{opt}}$, for any $c,c'\in\hv$ and $p,p'\in[n]$, if $\interpath{P_1}(c,p)=((c',p'),\epsilon)$ then $\interpath{P_2}(c',p')=((c,p),\epsilon)$. Hence, one can replace $P_1$ or $P_2$ by the horizontal reflection of the other without changing the semantics. This implies that $P_1$ and $P_2$ contain the same number of PBS (otherwise, by replacing the one with more PBS by the horizontal reflection of the other, one would obtain a diagram equivalent to $C_{w,\sigma}^{\mathrm{opt}}$ with strictly fewer PBS, which would contradict its query-PBS-optimality), and subsequently, that the diagram $C_{w,\sigma}^{\mathrm{opt}\prime}$ obtained by replacing $P_2$ by the horizontal reflection $P_1^{-1}$ of $P_1$ is still query-PBS-optimal.

Up to slightly deforming $P_1$, we can write it in the form
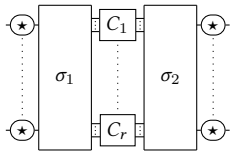
\[\input{superpetpermsnpext.tikz}\]
where $\sigma_1$ and $\sigma_2$ are permutations of the wires, the $C_k$ are of the form \ptikzfig{escalierbsnoir-abrege}{0.375}{0.6}1\nolinebreak, and \ptikzfig{negpotentiel-m}{0.4}{0.8}1 denotes either \ptikzfig{filcourt}{0.4}{0.8}1 or \ptikzfig{neg}{0.4}{0.8}1. Using this, we can write $C_{w,\sigma}^{\mathrm{opt}\prime}$ in the form
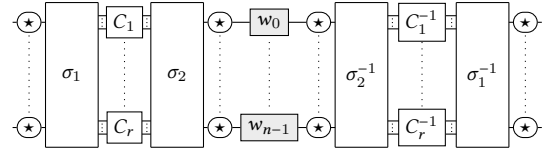
\[\input{P1wP1invsuperpetpermsnpext.tikz}\]
where given any gate-free diagram $D$, $D^{-1}$ denotes its horizontal reflection.

Since the diagram is symmetric, we can remove the negations in the middle without changing the semantics of the diagram or its query-PBS-optimality. This gives us
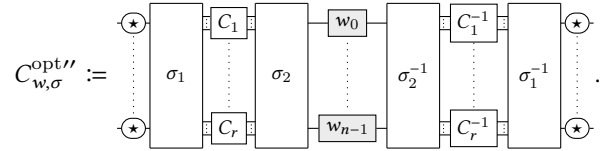
\[C_{w,\sigma}^{\mathrm{opt}\prime\prime}\coloneqq\input{P1wP1invsuperpetpermsnpextext.tikz}.\]
Let us consider the diagram
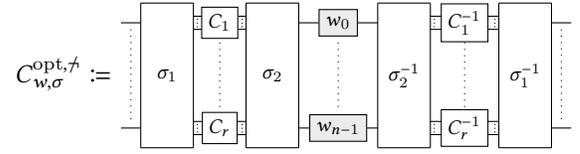
\[C_{w,\sigma}^{\mathrm{opt,\not\neg}}\coloneqq\input{P1wP1invsuperpetperms.tikz}\]
obtained by removing all negations on the sides of $C_{w,\sigma}^{\mathrm{opt}\prime\prime}$. For each $p\in[n]$ such that there was a negation on the $p$th input and output wire, one now has $\interpath{C_{w,\sigma}^{\mathrm{opt,\not\neg}}}({\Vpol},p)=(({\Vpol},p),w_{\sigma(p)})$ and $\interpath{C_{w,\sigma}^{\mathrm{opt,\not\neg}}}({\Hpol},p)=(({\Hpol},p),w_p)$. Let us consider the directed graph $\tilde G$ obtained by reversing the edge $e_p$ in $\vec G$ for every such $p$. For every $p\in[n]$, we denote by $\tilde e_p$ the $p$th edge of $\tilde G$, which is either $e_p$ or its reverse $(w_{\sigma(p)},w_p)$. Then for every $p\in[n]$, $w^{C_{w,\sigma}^{\mathrm{opt,\not\neg}}}_{{\Vpol},p}$ is the tail of $\tilde e_p$ and $w^{C_{w,\sigma}^{\mathrm{opt,\not\neg}}}_{{\Hpol},p}$ is the head of $\tilde e_p$.

$\tilde G$ can be edge-partitioned into $r$ cycles as follows: For each $k\in\{1,...,r\}$, let $n_k$ be such that $C_k:\top^{\oplus n_k}\to\top^{\oplus n_k}$. Let also $N_k\coloneqq\sum_{j=1}^{k-1}n_k$. By abuse of notation, we denote by $\sigma_1$ and $\sigma_2$ the permutations of $[n]$ respectively associated with the diagrams $\sigma_1$ and $\sigma_2$, so that $\forall c,p, \interpath{\sigma_1}(c,p)=((c,\sigma_1(p)),\epsilon)$ and $\interpath{\sigma_2}(c,p)=((c,\sigma_2(p)),\epsilon)$. Note that for any $k\in\{1,...,r\}$ and any $p\in[n_k]$, one has
\begin{itemize}
\item $\forall i\in[n_k], \interpath{C_{w,\sigma}^{\mathrm{opt,\not\neg}}}({\Vpol},\sigma_1^{-1}(N_k+i))=(({\Vpol},\sigma_1^{-1}(N_k+i)),w_{\sigma_2(N_k+i)})$
\item $\forall i\in[n_k-1], \interpath{C_{w,\sigma}^{\mathrm{opt,\not\neg}}}({\Hpol},\sigma_1^{-1}(N_k+i))=(({\Hpol},\sigma_1^{-1}(N_k+i)),w_{\sigma_2(N_k+i+1)})$
\item $\interpath{C_{w,\sigma}^{\mathrm{opt,\not\neg}}}({\Hpol},\sigma_1^{-1}(N_k+n_k-1))=(({\Hpol},\sigma_1^{-1}(N_k+n_k-1)),w_{\sigma_2(N_k)})$.
\end{itemize}
Hence, there is a cycle $w_{\sigma_2(N_k)}\to w_{\sigma_2(N_k+1)}\to\cdots\to w_{\sigma_2(N_k+n_k-1)}\to w_{\sigma_2(N_k)}$ in $\tilde G$, associated with $C_k$. Considering the cycle associated with each $C_k$ gives us an edge-partition of $\tilde G$ into $r$ cycles, since these cycles are edge-disjoint and cover all edges of $\tilde G$.

It remains to prove that there is no orientation of the edges of $G$ such that the resulting directed graph can be edge-partitioned into more than $r$ cycles. Reasoning by contradiction, assume that there exists such an orientation yielding an Eulerian directed graph $\dbtilde G$ with an edge-partition into $r'$ cycles with $r'>r$. We enumerate these cycles in an arbitrary order, and denote by $m_k$ the length of the $k$th cycle, for $k\in\{1,...,r'\}$. We denote by $\dbtilde e_p$ the $p$th edge of $\dbtilde G$, which is either $\tilde e_p$ or its reverse. Note that the in- and out-degree of each vertex are the same in $\dbtilde G$ as in $\vec G$ and $\tilde G$, so that there exist two permutations $\tau_1$ and $\tau_2$ of $[n]$ such that $\forall p\in[n], \dbtilde e_p=(w_{\tau_1(p)},w_{\tau_2(p)})$. Therefore, there exists an enumeration of $[n]$ as $(i^k_\ell)_{k\in\{1,...,r'\},\ell\in[m_k]}$, such that the $k$th cycle can be written \[w_{\tau_1(i_0^k)}\xrightarrow{\dbtilde e_{i_0^k}}w_{\tau_1(i_1^k)}\xrightarrow{\dbtilde e_{i_1^k}}\cdots \xrightarrow{\dbtilde e_{i_{m_k-2}^k}}w_{\tau_1(i_{m_k-1}^k)}\xrightarrow{\dbtilde e_{i_{m_k-1}^k}}w_{\tau_1(i_0^k)}.\]

Let $s$ be the permutation of $[n]$ such that $\forall k,\ell,s(i^k_\ell)=M_k+\ell$, where $M_k\coloneqq\sum_{j=1}^{k-1}m_k$. We make the same abuse of notation as for $\sigma_1$ and $\sigma_2$ by also denoting by $s$ the diagram that is a permutation of the wires according to $s$. We consider the following diagram, where for each $k\in\{1,...,r'\}$, $C'_k:\top^{\oplus m_k}\to\top^{\oplus m_k}$ is of the form \ptikzfig{escalierbsnoir-abrege}{0.375}{0.6}1, and for each $p\in[n]$, the \!\ptikzfig{negpotentiel-m}{0.4}{0.8}1 on wire $p$ is \ptikzfig{filcourt}{0.4}{0.8}1 if $\dbtilde e_p$ and $\tilde e_p$ have the same direction, or \ptikzfig{neg}{0.4}{0.8}1 otherwise:
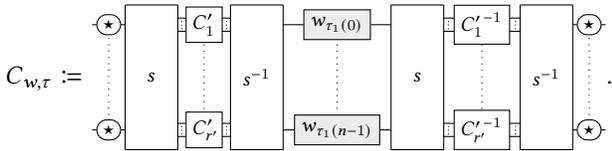
\[C_{w,\tau}\coloneqq\input{diagGdbtildesuperpetpermsnpext.tikz}.\]
For any $k\in\{1,...,r'\}$ and $\ell\in[m_k]$, if $\dbtilde e_p$ and $\tilde e_p$ have the same direction then one has $\interpath{C_{w,\tau}}({\Vpol},i^k_\ell)=(({\Vpol},i^k_\ell),w_{\tau_1(i^k_\ell)})$ and $\interpath{C_{w,\tau}}({\Hpol},i^k_\ell)=(({\Hpol},i^k_\ell),w_{\tau_1(i^k_{\ell+1\bmod m_k})})$, and if they have opposite directions then one has $\interpath{C_{w,\tau}}({\Vpol},i^k_\ell)=(({\Vpol},i^k_\ell),w_{\tau_1(i^k_{\ell+1\bmod m_k})})$ and $\interpath{C_{w,\tau}}({\Hpol},i^k_\ell)=(({\Hpol},i^k_\ell),w_{\tau_1(i^k_\ell)})$. That is, in any case, $w^{C_{w,\tau}}_{{\Vpol},i^k_\ell}$ is the tail of $\tilde e_{i^k_\ell}$ and $w^{C_{w,\tau}}_{{\Hpol},i^k_\ell}$ is its head. Since the indices $i^k_\ell$ span $[n]$ entirely, this implies that $C_{w,\tau}$ has the same semantics as $C_{w,\sigma}^{\mathrm{opt,\not\neg}}$. But $C_{w,\tau}$ contains $n-r'$ PBS whereas $C_{w,\sigma}^{\mathrm{opt,\not\neg}}$ contains $n-r$ PBS, so that $C_{w,\tau}$ contains strictly fewer PBS than $C_{w,\sigma}^{\mathrm{opt,\not\neg}}$, which contradicts the query-PBS-optimality of $C_{w,\sigma}^{\mathrm{opt,\not\neg}}$.

This proves that the edge-partition of $\tilde G$ into cycles obtained from $C_{w,\sigma}^{\mathrm{opt,\not\neg}}$ has maximum number of cycles among all possible choices of orientation and partition. In other words, the undirected edge-partition of $G$ obtained by erasing the directions of the edges in this edge-partition of $\tilde G$ has maximum number of cycles. This finishes the reduction.

\subsubsection{Proof of Corollary \ref{gatePBSnegfreeNPhard}}\label{preuvegatePBSnegfreeNPhard}

We reduce this problem from the problem maxDCD of, given an Eulerian directed graph $\vec G$, finding a maximum-cardinality edge-partition of $\vec G$ into directed cycles. This problem is defined and proved to be NP-hard in \cite{amir2010permutation}.

The proof has the same structure as the proof of \cref{gatePBSNPhard} : we define $C_{w,\sigma}$ in the same way, and we consider an equivalent diagram $C_{w,\sigma}^{\mathrm{opt}}$ which is now query-PBS-optimal only among negation-free diagrams. Since the PGT procedure preserves the property of being negation-free, we can still assume that it is in PGT form. With the same arguments as in the proof of \cref{gatePBSNPhard}, we can do the same deformations and define $C_{w,\sigma}^{\mathrm{opt}\prime}$ in the same way. This time, $C_{w,\sigma}^{\mathrm{opt}\prime}$ is negation-free, so that $C_{w,\sigma}^{\mathrm{opt,\not\neg}}=C_{w,\sigma}^{\mathrm{opt}\prime}$ and $\tilde G=\vec G$, so the construction of the proof of \cref{gatePBSNPhard} gives us an edge-partition of $\vec G$. To prove that this edge-partition has maximum cardinality, we only have to prove that there is no edge-partition of $\vec G$ into strictly more cycles, and the proof of this is the same as for \cref{gatePBSNPhard} (with the difference that we necessarily have $\dbtilde G=\vec G$, which allows for many simplifications).

\subsubsection{Proof of Corollary \ref{gateplusPBSNPhard}}\label{preuvegateplusPBSNPhard}

The proof relies on the following lemma:
\begin{lemma}\label{negfreeisbetter}
Given any diagram $D$ of $\mathcal P$ which is query-optimal and contains at least one negation, there exists an equivalent negation-free diagram with the same gates containing at most $\nopbs(D)+\noneg(D)-1$ PBS.
\end{lemma}
\begin{proof}[Proof of \cref{gateplusPBSNPhard}]
Note that the proofs of \cref{gatePBSNPhard,gatePBSnegfreeNPhard} actually give us slightly stronger results than the exact statements of \cref{gatePBSNPhard,gatePBSnegfreeNPhard}, since they in fact consider the restricted versions of their respective problems in which the input diagram is required to be in $\mathcal P$.

Given \cref{negfreeisbetter}, \cref{gateplusPBSNPhard} follows from this stronger version of \cref{gatePBSnegfreeNPhard}. Indeed, it suffices to prove that the problem of optimising $\nopbs(D)+\alpha \noneg(D)$ together with the queries is already NP-hard when restricted to the case where the input diagram $D$ is negation-free and in $\mathcal P$. Given such a diagram $D$, any query-optimal diagram $D'$ equivalent to $D$ such that $\nopbs(D')+\alpha\noneg(D')$ is minimal, is negation-free. Indeed, if it was not, then, since it is in $\mathcal P$, by \cref{negfreeisbetter} there would exist an equivalent query-optimal, negation-free diagram $D''$ that would satisfy $\nopbs(D'')+\alpha\noneg(D'')=\nopbs(D'')\leq\nopbs(D')+\noneg(D')-1<\nopbs(D')+\alpha\noneg(D')$, which would contradict the fact that $\nopbs(D')+\alpha\noneg(D')$ is minimal. Thus, finding a query-optimal diagram $D'$ equivalent to $D$ such that $\nopbs(D')+\alpha\noneg(D')$ is minimal, amounts to finding a diagram equivalent to $D$ and query-PBS-optimal among negation free diagrams.
\end{proof}

It remains to prove \cref{negfreeisbetter}:

\begin{proof}[Proof of \cref{negfreeisbetter}]
Let $D:\top^{\oplus n}\to\top^{\oplus n}$ be a query-optimal diagram of $\mathcal P$ containing at least one negation. Let us first apply Step \ref{putinsnailform} of the PGT procedure, that is, by mere deformation, we put $D$ in the form
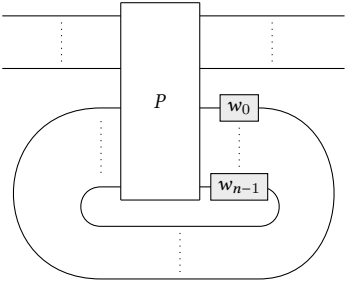
\[\input{PGTw.tikz}\]
with $P$ gate-free. Let $f(P)$ be the number of positions $p$ such that $c^P_{\Vpol,p}=\Hpol$ (note that this number does not depend on the way of deforming $D$). Since the semantics of $P$ applies a permutation to the couples $(c,p)$, there are the same number of positions $p$ such that $c^P_{\Hpol,p}=\Vpol$, so that there are $2f(P)$ couples $(c,p)$ such that $c^D_{c,p}\neq c$. Each photon that enters $P$ with a basis state corresponding to one of these couples gets its polarisation changed while traversing $P$, which means that it traverses at least one negation. Since each negation can be reached from at most two basis states, this implies that $f(P)\leq \noneg(D)$.

Note that additionally, due to the semantics of $D$ (since it is in $\mathcal P$), for any $c,c'\in\hv$ and $p,p'\in[2n]$ such that $\interpath{P}(c,p)=(c',p')$, one has $p\in[n]$ if and only if $p'\in\{n,...,2n-1\}$ and vice-versa, and $\interpath{P}(c',p')=(c,p)$. Combined with the fact that $\interpath P$ applies a permutation to the couples $(c,p)$, this implies that there are the same number of positions $p$ such that respectively: $p\in[n]$ and $c^P_{\Vpol,p}=\Hpol$; $p\in[n]$ and $c^P_{\Hpol,p}=\Vpol$; $p\in\{n,...,2n-1\}$ and $c^P_{\Vpol,p}=\Hpol$; $p\in\{n,...,2n-1\}$ and $c^P_{\Hpol,p}=\Vpol$. Since the sum of these four numbers of positions is equal to $2f(P)$, this implies that $f(P)$ is even and that the number of positions $p$ is equal to $\dfrac{f(P)}{2}$ in each case.

By applying the rest of the PGT procedure, we put $P$ in stair form and thereby transform $D$ into a diagram $D'$ in PGT form. With a similar argument as in the proof of \cref{gatePBSNPhard}, we can put $D'$ in the form
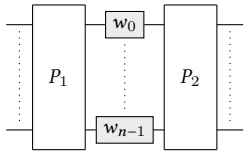
\[\input{P1wP2.tikz}\]
where $P_1$ and $P_2$ are in stair form. Since $D\in\mathcal P$, for any $c,c'\in\hv$ and $p,p'\in[n]$, if $\interpath{P_1}(c,p)=((c',p'),\epsilon)$ then $\interpath{P_2}(c',p')=((c,p),\epsilon)$. Hence, $P_1$ has the same semantics as the horizontal reflection of $P_2$ and vice-versa. By \cref{stairsoptimal}, this implies that $P_1$ and $P_2$ contain the same number of PBS. Therefore, by replacing $P_2$ by the horizontal reflection of $P_1$, we get a diagram $D''$ which is still equivalent to $D$ and still has at most as many PBS as $D$. As in the proof of \cref{gatePBSNPhard}, we can write $D''$ in the form
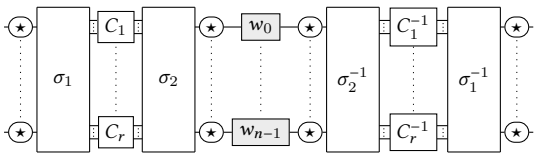
\[\input{P1wP1invsuperpetpermsnpext.tikz}\]
where $\sigma_1$ and $\sigma_2$ are permutation of the wires, the $C_k$ are staircases (see \cref{defNMF}), and given any gate-free diagram $E$, $E^{-1}$ denotes its horizontal reflection. Since $D''$ is symmetric, we can remove the negations in the middle without changing its semantics, which gives us
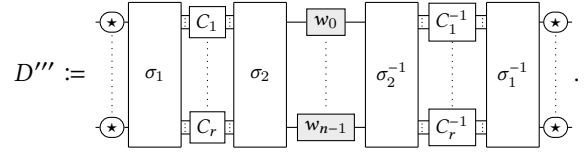
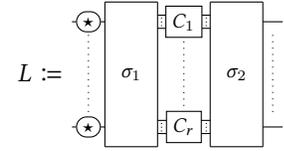
\[D'''\coloneqq\input{P1wP1invsuperpetpermsnpextext.tikz}.\]
Let
\[L\coloneqq\input{superpetpermsnpextgauche.tikz}.\]
For every letter $U\in\{w_0,...,w_{n-1}\}$, let $d_U(D)$ be the number of positions $p$ such that for some $p_1,p_2$, one has $p^L_{\Vpol,p_1}=p^L_{\Vpol,p_2}=p$ and $w_{p}=U$. Since $D\in\mathcal P$, $U$ appears as many times among the $w^D_{\Vpol,p}$ as among the $w^D_{\Hpol,p}$. Since $\forall c,p, w^D_{c,p}=w_{p^L_{c,p}}$, this implies that the number of positions $p'$ such that for some $p'_1,p'_2$ one has $p^L_{\Hpol,p'_1}=p^L_{\Hpol,p'_2}=p'$ and $w_{p'}=U$ is also $d_U(D)$. We arbitrarily associate a position $p'$ of the second kind with each position $p$ of the first kind, so as to distribute these $2d_U(D)$ positions into $d_U(D)$ couples $(p,p')$. By doing so for every $U\in\{w_0,...,w_{n-1}\}$, we obtain $d(D)$ couples, where $d(D)\coloneqq\displaystyle\sum_{U\in\{w_0,...,w_{n-1}\}}d_U(D)$, with the property that every position $q$ satisfying for some $p_1,p_2,c$, $p^L_{c,p_1}=p^L_{c,p_2}=q$, appears exactly once among these couples (as a left element if $c=\Vpol$, and as a right element if $c=\Hpol$).

For each position $q$ among these $2d(D)$ positions, there is exactly one polarisation $c$ such that $c^L_{c,p}\neq c$. This property is not affected by appending negations at the right of $L$, so that there is also exactly one polarisation $c$ (for each $q$) such that $c^{P_1}_{c,q}\neq c$. By definition of $P_1$, there is also exactly one polarisation $c$ such that $c^{P}_{c,q}\neq c$. Since all of these $2d(D)$ positions $q$ are in $[n]$, this implies that $2\,\dfrac{f(P)}{2}\geq 2d(D)$. Since $f(P)\leq \noneg(D)$, this inequality implies that $2d(D)\leq \noneg(D)$.

For each of the $d(D)$ couples $(p,p')$, we do the following transformation in $D'''$ (up to deformation):
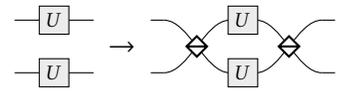
\[\input{gateUsurgateU.tikz}\ \to\ \input{bsUhUbbs.tikz}.\]
Each time, we put $L$ in stair form again, we transform $L^{-1}$ symmetrically so that it remains the horizontal reflection of $L$, and we remove any negations at the right of $L$ and at the left of $L^{-1}$, which is possible because $D'''$ remains symmetric.
One can check that if the PBS appended to $L$ is connected to two different $C_i$s, then this results in merging them together, so that the number of PBS stays the same (after adding the additional PBS), and if it is connected to a single $C_i$ then this results in splitting it into two staircases, so that the number of PBS in $L$ decreases by $2$. The behaviour of $L^{-1}$ is symmetric. At the end, the total number of PBS is at most $\nopbs(D)+2d(D)$, and the equality can be reached only if at every step two $C_i$s have been merged.
This gives us a diagram
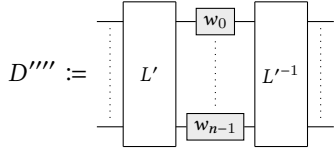
\[D''''\coloneqq\input{LwLinvprimes.tikz}\]
with $L'$ of the form
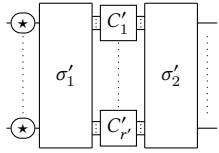
\[\input{superpetpermsnpextgaucheprimes.tikz}\]
in which there are no couples of positions $p_1,p_2$ such that $p^L_{\Vpol,p_1}=p^L_{\Vpol,p_2}$ anymore. In particular, for each position $p$ such that for some $p_1$, $\interpath{L'}(\Vpol,p_1)=((\Hpol,p),\epsilon)$, there exists $p_2$ such that $\interpath{L'}(\Hpol,p_2)=((\Vpol,p),\epsilon)$. For each of these positions, we apply the following transformation:
\[\begin{tikzpicture}
	\begin{pgfonlayer}{nodelayer}
		\node [style=newe] (0) at (0, 0) {$U$};
		\node [style=none] (1) at (-1.5, 0) {};
		\node [style=none] (2) at (1.5, 0) {};
	\end{pgfonlayer}
	\begin{pgfonlayer}{edgelayer}
		\draw [style=new] (1.center) to (2.center);
	\end{pgfonlayer}
\end{tikzpicture}\ \to\ \input{negUneg.tikz}\]
This gives us a diagram
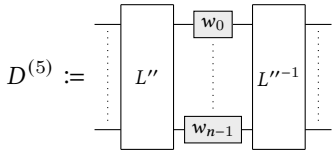
\[D^{(5)}\coloneqq\input{LwLinvsecondes.tikz}\]
with $L''$ such that for all $c,p$, $c^{L''}_{c,p}=c$. By putting $L''$ in normal form, then in stair form again, we get a diagram $L'''$ without negations and with at most as many $PBS$ as $L''$. In particular, the resulting diagram $D^{(6)}$ (after proceeding symmetrically in ${L''}^{-1}$) contains at most $\nopbs(D)+2d(D)$ PBS. If it has strictly fewer PBS, or if $2d(D)< \noneg(D)$, then we have the desired result. If it has exactly $\nopbs(D)+2d(D)$ PBS and $2d(D)=\noneg(D)$, then this means in particular that at each of the steps of the transformation of $D'''$ into $D''''$, two $C_i$s have been merged. By hypothesis, $\noneg(D)\geq1$, so the fact that $2d(D)=\noneg(D)$ implies that $d(D)>0$. This implies that there has been at least one step in the transformation of $D'''$ into $D''''$, in which two staircases connected to two gates with the same label have been merged. Since these staircases have not been split, there is at least one couple of gates in $D^{(5)}$ that have the same label and are connected to the same staircase. Then by applying to them the same transformation as before:
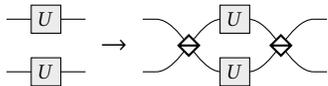
\[\input{gateUsurgateU.tikz}\ \to\ \input{bsUhUbbs.tikz}\]
and putting $L'''$ (and ${L'''}^{-1}$) in stair form again, we get a diagram with $\nopbs(D)+\noneg(D)-2$ PBS and no negations, which is equivalent to $D$.
\end{proof}

\section{Derivations of Equations \eqref{fusionUV-} to \eqref{mergegatesBB}}\label{preuveequationsgateoptproc}
Note that \cref{fusionUV-} is a particular case of \cref{fusionUVR} and that \cref{fusionUVB} is a particular case of \cref{fusionUVBgenerale}.
To prove \cref{fusionUV}, we derive a more general version, analogous to \cref{fusionUVR,fusionUVBgenerale}: for any monoid $\M$ and any $U,V\in\M$,\medskip

\begin{longtable}{RCL}
\input{gateUgateVbis.tikz}&\eqdeuxeqref{bsbssplitmerge}{ubscol}&\input{NbsRBgateUgateVRBbsN.tikz}\\\\
&\eqdeuxeqref{fusionUVR}{fusionUVBgenerale}&\input{NbsRBgateVURBbsN.tikz}\\\\
&\eqdeuxeqref{ubscol}{bsbssplitmerge}&\begin{tikzpicture}
	\begin{pgfonlayer}{nodelayer}
		\node [style=newe] (0) at (0, 0) {$VU$};
		\node [style=none] (1) at (-1.5, 0) {};
		\node [style=none] (2) at (1.5, 0) {};
	\end{pgfonlayer}
	\begin{pgfonlayer}{edgelayer}
		\draw [style=noire] (1.center) to (0);
		\draw [style=noire] (0) to (2.center);
	\end{pgfonlayer}
\end{tikzpicture}

\end{longtable}\medskip

To prove \cref{mergegatesRB}, we have:\medskip

\begin{longtable}{RCL}
\input{gateURsurgateUB-label.tikz}&\eqeqref{bsbsmergesplit}&\input{bsbsgateURsurgateUB.tikz}\\\\
&\eqeqref{ubscol}&\input{RBbsNUbsRB.tikz}
\end{longtable}\medskip

To prove \cref{mergegatesBR}, we have:\medskip

\begin{longtable}{RCL}
\input{gateUBsurgateUR-label.tikz}&=&\input{xgateURsurgateUBx.tikz}\\\\
&\eqeqref{mergegatesRB}&\input{BRxbsNUbsxBR.tikz}\\\\
&\eqdeuxeqref{bsxcol}{xbscol}&\input{BRbsNUbsBR.tikz}
\end{longtable}\medskip

To prove \cref{mergegatesRR}, we have:\medskip

\begin{longtable}{RCL}
\input{gateURsurgateUR-label.tikz}&\eqdeuxeqref{negnegRBR}{neguRB}&\input{gateURsurneggateUBneg.tikz}\\\\
&\eqeqref{mergegatesRB}&\input{RRnbbsNUbsnbRR.tikz}
\end{longtable}\medskip

To prove \cref{mergegatesBB}, we have:\medskip

\begin{longtable}{RCL}
\input{gateUBsurgateUB-label.tikz}&\eqdeuxeqref{negnegBRB}{neguRB}&\input{neggateURnegsurgateUB.tikz}\\\\
&\eqeqref{mergegatesRB}&\input{BBnhbsNUbsnhBB.tikz}
\end{longtable}

\section{PGT Procedure}

\subsection{Full Description of the Procedure}\label{PBSoptprocillustree}

\begin{enumerate}
\setcounter{enumi}{-1}
\item During the whole procedure, every time there are two consecutive negations, we remove them using Equation \eqref{negnegRBR}, \eqref{negnegBRB} or their all-black version:
\begin{equation}\label{negnegnoir}\input{negneg.tikz}\ =\ \begin{tikzpicture}
	\begin{pgfonlayer}{nodelayer}
		\node [style=none] (0) at (-1.5, 0) {};
		\node [style=none] (1) at (1.5, 0) {};
	\end{pgfonlayer}
	\begin{pgfonlayer}{edgelayer}
		\draw [style=noire] (0.center) to (1.center);
	\end{pgfonlayer}
\end{tikzpicture}\end{equation}
\item\label{putinsnailform} Deform $D_0$ to put it in the form \eqref{snail} with $P$ gate-free. The goal of the following steps is to put $P$ in stair form.
\item\label{splitallPBS} Split all PBS of the form $\stikzfig{abbeamsplitter}$ into combinations of $\stikzfig{beamsplitterNRwB}$, $\stikzfig{beamsplitterwBNR}$, $\stikzfig{beamsplitterRNBw}$ and $\stikzfig{beamsplitterBwRN}$, using Equations \eqref{bsbigebre} to \eqref{bsbbnn}.
\item As long as there are two PBS connected by a black wire, with possibly a black negation on this wire, push this negation out (if present) using Equation \eqref{nbscol}, and cancel the PBS together using Equation \eqref{bsbsmergesplit}. It may be necessary to flip the PBS upside down using Equation \eqref{bsxcol} and/or \eqref{xbscol} in order to be able to apply Equations \eqref{nbscol} and \eqref{bsbsmergesplit}. Note also that to cancel the two PBS together one may have to use dinaturality:
\[\ptikzfig{tracebsNRwBgapbsRNBw-raccourci}{0.5}{0.8}1\ =\ \ptikzfig{traceBdanstraceRgap-label-raccourci}{0.5}{0.8}1\]
When there are not two such PBS anymore, all black wires are connected to at least one side of $P$ (possibly through negations), and the PBS are connected together with red and blue wires with possibly negations on them.
\item Remove all loops using the following equations:\newline
\begin{tabularx}{0.9\linewidth}{X@{\qquad\qquad\qquad\qquad}X}
\begin{equation}\label{bouclevideR}\vphantom{\input{bouclevideneg.tikz}}\input{bouclevideR-label.tikz}\ =\ \ptikzfig{diagrammevide}1{0.5}{0.465}\end{equation}&
\begin{equation}\label{bouclevideB}\vphantom{\input{bouclevideneg.tikz}}\input{bouclevideB-label.tikz}\ =\ \ptikzfig{diagrammevide}1{0.5}{0.465}\end{equation}\\
\begin{equation}\label{bouclevideN}\vphantom{\input{bouclevideneg.tikz}}\input{bouclevidecentree.tikz}\ =\ \ptikzfig{diagrammevide}1{0.5}{0.465}\end{equation}&
\begin{equation}\label{bouclevideneg}\input{bouclevideneg.tikz}\ =\ \ptikzfig{diagrammevide}1{0.5}{0.465}\end{equation}
\end{tabularx}
\newline Note that since $D_0$ is query-optimal, there cannot be loops containing gates at this point.
\item Deform $P$ to put it in the form \eqref{stairs} with $\sigma_1$ and $\sigma_2$ being wire permutations and the $C_i$ being trace-free and connected. It remains to transform the $C_i$ into staircases. Up to additional deformation of $P$ in order to reorder the input and output wires of the $C_i$, and to using Equations \eqref{bsxcol} and \eqref{xbscol}, every $C_i$ is of one of the following forms:
\[\ptikzfig{escalierbsnegpotentielsIOalenvers}{0.6}{0.6857}1\qquad\quad\ptikzfig{escalierbsnegpotentielsO}{0.6}{0.6857}1\qquad\quad\ptikzfig{escalierbsnegpotentielsI}{0.6}{0.6857}1\qquad\quad\ptikzfig{escalierbsnegpotentiels}{0.6}{0.6857}1\]
where $\ptikzfig{negpotentiel-m}{0.4}{0.8}1$ is either $\ntikzfig{cfilcourt-s}$ or $\sgtikzfig{cneg}$ with $a\in\{\vtype,\htype\}$, $\stikzfig{bs3pattessymsplit}$ is either $\stikzfig{beamsplitterNRwB}$ or $\stikzfig{beamsplitterwBNR}$ and $\stikzfig{bs3pattessymmerge}$ is either $\stikzfig{beamsplitterRNBw}$ or $\stikzfig{beamsplitterBwRN}$.
\item\label{pushnegationsillustree} Remove the negations in the middle of the $C_i$ by pushing them to the bottom by means of Equation \eqref{nbscol} and its following variants (all of the form ``a three-wire PBS with a negation on one of the three wires is equal to this PBS reflected vertically with negations on the other two wires''; note that Equations \eqref{nbscol}, \eqref{nbswBNRcomm}, \eqref{bsRNBwncomm} and \eqref{bsBwRNncomm} have to be applied from right to left, while Equations \eqref{bsNRwBnhcomm}, \eqref{bswBNRnhcomm}, \eqref{nhbsRNBwcomm} and \eqref{nhbsBwRNcomm} have to be applied from left to right):\newline
\begin{tabularx}{\linewidth}{XX}
&
\begin{equation}\label{nbswBNRcomm}\input{nbswBNR.tikz}\ =\ \input{bsNRwBnn.tikz}\end{equation}\\
\begin{equation}\label{bsRNBwncomm}\input{bsRNBwn.tikz}\ =\ \input{nnbsBwRN.tikz}\end{equation}&
\begin{equation}\label{bsBwRNncomm}\input{bsBwRNn.tikz}\ =\ \input{nnbsRNBw.tikz}\end{equation}\\
\begin{equation}\label{bsNRwBnhcomm}\input{bsNRwBnh.tikz}\ =\ \input{nbswBNRnb.tikz}\end{equation}&
\begin{equation}\label{bswBNRnhcomm}\input{bswBNRnh.tikz}\ =\ \input{nbsNRwBnb.tikz}\end{equation}\\
\begin{equation}\label{nhbsRNBwcomm}\input{nhbsRNBw.tikz}\ =\ \input{nbbsBwRNn.tikz}\end{equation}&
\begin{equation}\label{nhbsBwRNcomm}\input{nhbsBwRN.tikz}\ =\ \input{nbbsRNBwn.tikz}\end{equation}\end{tabularx}
\item\label{recomposestairs} Up to deforming $P$ in order to flip the $C_i$ upside down, and to using Equations \eqref{bsxcol} and \eqref{xbscol} wherever necessary, every $C_i$ is now of one of the following forms (note that it is easy to know of which form each $C_i$ should be, before deforming it, by looking at its input/output type):\newline
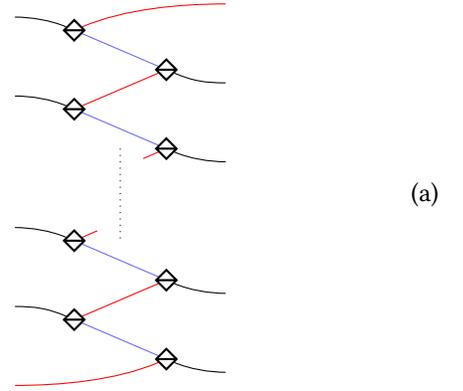
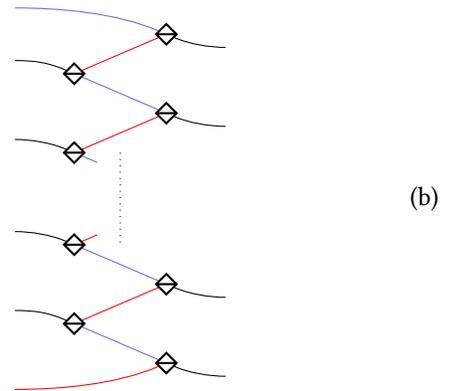
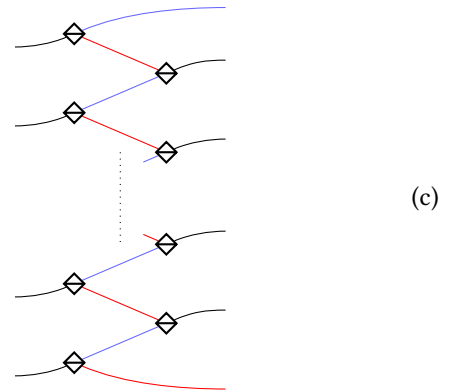
\begin{tabularx}{\linewidth}{X@{\qquad}X@{\qquad}X}\begin{eqnabc}\label{formeRR}\ptikzfig{escalierbs3pattesRR}{0.6}{0.8}1\end{eqnabc}&\begin{eqnabc}\label{formeBRI}\ptikzfig{escalierbs3pattesBRI}{0.6}{0.8}1\end{eqnabc}&\begin{eqnabc}\label{formeBRO}\ptikzfig{escalierbs3pattesBRO}{0.6}{0.8}1\end{eqnabc}
\end{tabularx}
\begin{tabularx}{0.7\linewidth}{X@{\quad}X}
\begin{eqnabc}\label{formeN}\ptikzfig{escalierbs3pattesN}{0.6}{0.8}1\end{eqnabc}&\begin{eqnabc}\label{formeBB}\quad\ \ptikzfig{escalierbs3pattesBB}{0.6}{0.8}1\end{eqnabc}\end{tabularx}

Transform each of them into one of the five kind of staircases depicted in Definition \ref{defNMF}, depending on its type:
\begin{itemize}
\item If $C_i$ is of the form \eqref{formeRR}, \eqref{formeBRI} or \eqref{formeBRO}, then repeatedly apply \cref{bsnrrn} or \eqref{bsrnnr} to put it in the form \ptikzfig{escalierbsrouge-label-abrege}{0.5}{0.8}1\nolinebreak, \ptikzfig{escalierbsrougemerge-label-abrege}{0.5}{0.8}1 or\smallskip \ptikzfig{escalierbsrougemergeinverse-label-abrege}{0.5}{0.8}1 respectively.
\item If $C_i$ is of the form \eqref{formeBB}, then repeatedly apply the following equation:
\begin{equation}\label{bsxNBBN}\input{defbsxNBBN.tikz}\ =\ \input{bsxNBBN-label.tikz}\end{equation}
to put it in the form \ptikzfig{escalierbsbleu-label-abrege}{0.5}{0.8}1.
\item If $C_i$ is of the form \eqref{formeN}, then repeatedly apply the following variant of Equation \eqref{bsbsmergesplit}:
\begin{equation}\label{bsbsmergesplitRBBR}
\input{RBbsNbsBR.tikz}\ =\ \input{swapRB-label.tikz}
\end{equation}
and Equation \eqref{bsbigebre}, as follows:\medskip

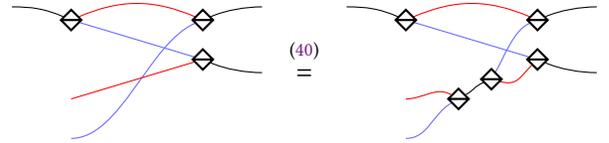
\begingroup
\par
\parshape0
\begin{longtable}{RCL}
\input{motifhautformeN.tikz}&\eqeqref{bsbsmergesplitRBBR}&\input{motifhautformeNmodifie.tikz}\\\\
&\eqeqref{bsbigebre}&\input{motifhautformeNmodifie2.tikz}
\end{longtable}\medskip\par
\endgroup
and finally apply \cref{bsbssplitmerge} once, in order to put it in the form \ptikzfig{escalierbsnoir-abrege}{0.5}{0.8}1\nolinebreak.

\end{itemize}
This gives us the desired diagram $D_1$ and finishes the procedure.
\end{enumerate}

\subsection{Derivations of the Ancillary Equations}\label{derivationsforPGTproc}

We have to derive Equations \eqref{negnegnoir} to \eqref{bsbsmergesplitRBBR} from the equations of Figure \ref{axiomsCPBS}. In order to benefit from some dependencies between the derivations, we treat the equations in the following order: \eqref{nbswBNRcomm}, \eqref{negnegnoir}, \eqref{bouclevideR}, \eqref{bouclevideB}, \eqref{bouclevideN}, \eqref{bsbsmergesplitRBBR}, \eqref{bouclevideneg}, \eqref{bsRNBwncomm}, \eqref{bsBwRNncomm}, \eqref{bsNRwBnhcomm}, \eqref{bswBNRnhcomm}, \eqref{nhbsRNBwcomm}, \eqref{nhbsBwRNcomm}, \eqref{bsxNBBN}.

For Equation \eqref{nbswBNRcomm}, the derivation is the following:\medskip

\begin{longtable}{RCL}\input{nbswBNR.tikz}&\eqeqref{bsxcol}&\input{nbsNRwBx.tikz}\\\\
&\eqeqref{nbscol}&\input{bswBNRnnx.tikz}\\\\
&\eqeqref{bsxcol}&\input{bsNRwBnn.tikz}\end{longtable}\medskip

For Equation \eqref{negnegnoir}:\medskip

\begin{longtable}{RCL}\input{negneg.tikz}&\eqeqref{bsbssplitmerge}&\input{nnbsRBbs.tikz}\\\\
&\eqeqref{nbscol}&\input{nbsBRnnRBbs.tikz}\\\\
&\eqeqref{nbswBNRcomm}&\input{bsRBnnnnRBbs.tikz}\\\\
&\eqtroiseqref{negnegRBR}{negnegBRB}{bsbssplitmerge}&\begin{tikzpicture}
	\begin{pgfonlayer}{nodelayer}
		\node [style=none] (0) at (-1.5, 0) {};
		\node [style=none] (1) at (1.5, 0) {};
	\end{pgfonlayer}
	\begin{pgfonlayer}{edgelayer}
		\draw [style=noire] (0.center) to (1.center);
	\end{pgfonlayer}
\end{tikzpicture}\end{longtable}\medskip

For Equation \eqref{bouclevideR}:\medskip

\begin{longtable}{RCL}\input{bouclevideR-label.tikz}&\eqeqref{idboxR}&\input{bouclevideIR-label.tikz}\\\\
&\eqeqref{bouclevideUR}&\ptikzfig{diagrammevide}1{0.5}{0.465}\end{longtable}\medskip

For Equation \eqref{bouclevideB}:\medskip

\begin{longtable}{RCL}\input{bouclevideB-label.tikz}&\eqeqref{negnegBRB}&\input{bouclevidenegnegBRB.tikz}\\\\
&\eqexpl{\text{dinaturality}}&\input{bouclevidenegnegRBR.tikz}\\\\
&\eqeqref{negnegRBR}&\input{bouclevideR-label.tikz}\\\\
&\eqeqref{bouclevideR}&\ptikzfig{diagrammevide}1{0.5}{0.465}\end{longtable}\medskip

For Equation \eqref{bouclevideN}:\medskip

\begin{longtable}{RCL}\input{bouclevide.tikz}&\eqeqref{bsbssplitmerge}&\input{bouclevidebsRBbs.tikz}\\\\
&\eqexpl{\text{dinaturality}}&\input{bouclevideRBbsNbsRB.tikz}\\\\
&\eqeqref{bsbsmergesplit}&\input{bouclevideBdansbouclevideR.tikz}\\\\
&\eqdeuxeqref{bouclevideB}{bouclevideR}&\ptikzfig{diagrammevide}1{0.5}{0.465}\end{longtable}\medskip

For Equation \eqref{bsbsmergesplitRBBR}:\medskip

\begin{longtable}{RCL}\input{RBbsNbsBR.tikz}&\eqeqref{bsxcol}&\input{RBbsbsRBx.tikz}\\\\
&\eqeqref{bsbsmergesplit}&\input{swapRB-label.tikz}\end{longtable}\medskip

For Equation \eqref{bouclevideneg}:\medskip

\begin{longtable}{RCL}\input{bouclevideneg.tikz}&\eqeqref{bsbssplitmerge}&\input{bouclevidenegbsRBbs.tikz}\\\\
&\eqeqref{nbscol}&\input{bouclevidebsBRnnRBbs.tikz}\\\\
&\eqexpl{\text{dinaturality}}&\input{bouclevidennRBbsNbsBR.tikz}\\\\
&\eqeqref{bsbsmergesplitRBBR}&\input{bouclevidennswapRB.tikz}\\\\
&=&\input{bouclevidenegnegBRB.tikz}\\\\
&\eqdeuxeqref{negnegBRB}{bouclevideN}&\ptikzfig{diagrammevide}1{0.5}{0.465}\end{longtable}\medskip

For Equation \eqref{bsRNBwncomm}:\medskip

\begin{longtable}{RCL}\input{bsRNBwn.tikz}&\eqeqref{bsbssplitmerge}&\input{RBbsnNbsRBbs.tikz}\\\\
&\eqeqref{nbscol}&\input{RBbsNbsBRnnRBbs.tikz}\\\\
&\eqeqref{bsbsmergesplitRBBR}&\input{xnnbsRNBw.tikz}\\\\
&\eqeqref{xbscol}&\input{nnbsBwRN.tikz}\end{longtable}\medskip

For Equation \eqref{bsBwRNncomm}:\medskip

\begin{longtable}{RCL}\input{bsBwRNn.tikz}&=&\input{xxbsBwRNn.tikz}\\\\
&\eqeqref{xbscol}&\input{xbsRNBwn.tikz}\\\\
&\eqeqref{bsRNBwncomm}&\input{xnnbsBwRN.tikz}\\\\
&\eqeqref{xbscol}&\input{nnbsRNBw.tikz}\end{longtable}\medskip

For Equation \eqref{bsNRwBnhcomm}:\medskip

\begin{longtable}{RCL}\input{bsNRwBnh.tikz}&\eqeqref{negnegBRB}&\input{bsNRwBnhnbnb.tikz}\\\\
&\eqeqref{nbswBNRcomm}&\input{nbswBNRnb.tikz}\end{longtable}\medskip

For Equation \eqref{bswBNRnhcomm}:\medskip

\begin{longtable}{RCL}\input{bswBNRnh.tikz}&\eqeqref{negnegRBR}&\input{bswBNRnhnbnb.tikz}\\\\
&\eqeqref{nbscol}&\input{nbsNRwBnb.tikz}\end{longtable}\medskip

For Equation \eqref{nhbsRNBwcomm}:\medskip

\begin{longtable}{RCL}\input{nhbsRNBw.tikz}&\eqeqref{negnegBRB}&\input{nhnbnbbsRNBw.tikz}\\\\
&\eqeqref{bsBwRNncomm}&\input{nbbsBwRNn.tikz}\end{longtable}\medskip

For Equation \eqref{nhbsBwRNcomm}:\medskip

\begin{longtable}{RCL}\input{nhbsBwRN.tikz}&\eqeqref{negnegRBR}&\input{nhnbnbbsBwRN.tikz}\\\\
&\eqeqref{bsRNBwncomm}&\input{nbbsRNBwn.tikz}\end{longtable}\medskip

For Equation \eqref{bsxNBBN}:\medskip

\begin{longtable}{RCL}\input{defbsxNBBN.tikz}&=&\input{defbsNNBBx.tikz}\\\\
&\eqeqref{bsnnbb}&\input{bsxNBBN-label.tikz}\end{longtable}

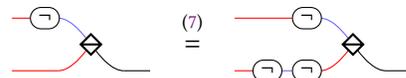
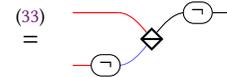
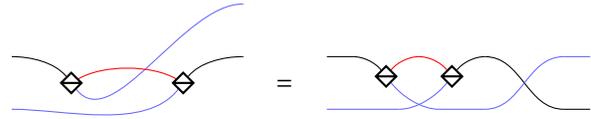
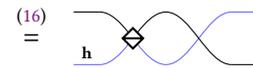
\end{document}